\documentclass[a4paper,11pt]{article}
\makeatletter
\renewcommand\normalsize{%
   \@setfontsize\normalsize{10.78pt}{12.5pt}%
   \abovedisplayskip 11\p@ \@plus3\p@ \@minus6\p@
   \abovedisplayshortskip \z@ \@plus3\p@
   \belowdisplayshortskip 6.5\p@ \@plus3.5\p@ \@minus3\p@
   \belowdisplayskip \abovedisplayskip
   \let\@listi\@listI}
\makeatother

\usepackage{amsmath,amssymb,amsthm,amsfonts,cases}
\usepackage{mathrsfs,booktabs}
\usepackage{url}
\usepackage{authblk}
\usepackage[usenames]{color}
\usepackage{graphicx}
\usepackage{float}
\usepackage{caption}
\usepackage{subcaption}
\usepackage{geometry}
\usepackage{multirow}
\usepackage{mathptmx}
\allowdisplaybreaks
\usepackage[compact]{titlesec}
\titlespacing{\section}{0pt}{2ex}{1ex}
\titlespacing{\subsection}{0pt}{1ex}{0ex}
\titlespacing{\subsubsection}{0pt}{0.5ex}{0ex}

\usepackage[onehalfspacing]{setspace}
\usepackage[normalem]{ulem}

\usepackage{mathtools}
\mathtoolsset{showonlyrefs=true}

\usepackage[round, authoryear]{natbib}

\usepackage{hyperref}

\newtheorem{remark}{Remark}[section]
\newtheorem{lemma}{Lemma}[section]
\newtheorem{theorem}{Theorem}[section]

\newtheorem{proposition}{Proposition}[section]
\newtheorem{corollary}{Corollary}[section]
\newtheorem{definition}{Definition}[section]
\newtheorem{assumption}{Assumption}[section]

\makeatletter\@addtoreset{equation}{section} \makeatother

\newcommand{\Proj}{\mathrm{Proj}}
\newcommand{\Dist}{\mathrm{Dist}}

\DeclareMathOperator*{\argmin}{arg\,min}
\DeclareMathOperator*{\argmax}{arg\,max}

\begin{document}

\title{Robust optimal investment and consumption strategies with portfolio constraints and stochastic environment}

\author[1,3]{Len Patrick Dominic M. Garces\thanks{Email: LenPatrickDominic.Garces@uts.edu.au}}
\author[2,3]{Yang Shen\thanks{Email: y.shen@unsw.edu.au; Corresponding Author}}
\affil[1]{\small School of Mathematical and Physical Sciences, University of Technology Sydney, Ultimo NSW 2007, Australia}
\affil[2]{\small School of Risk and Actuarial Studies, University of New South Wales, Syndey NSW 2052, Australia}
\affil[3]{\small ARC Centre of Excellence in Population Ageing Research, University of New South Wales, Sydney NSW 2052, Australia}

\date{27 June 2024}
\maketitle

\begin{abstract}
\begin{singlespace}
\noindent
    We investigate a continuous-time investment-consumption problem with model uncertainty in a general diffusion-based market with random model coefficients. We assume that a power utility investor is ambiguity-averse, with the preference to robustness captured by the homothetic multiplier robust specification, and the investor's investment and consumption strategies are constrained to closed convex sets. To solve this constrained robust control problem, we employ the stochastic Hamilton-Jacobi-Bellman-Isaacs equations, backward stochastic differential equations, and bounded mean oscillation martingale theory. Furthermore, we show the investor incurs (non-negative) utility loss, i.e. the loss in welfare, if model uncertainty is ignored. When the model coefficients are deterministic, we establish formally the relationship between the investor's robustness preference and the robust optimal investment-consumption strategy and the value function, and the impact of investment and consumption constraints on the investor's robust optimal investment-consumption strategy and value function. Extensive numerical experiments highlight the significant impact of ambiguity aversion, consumption and investment constraints, on the investor's robust optimal investment-consumption strategy, utility loss, and value function. Key findings include: 1) short-selling restriction always reduces the investor's utility loss when model uncertainty is ignored; 2) the effect of consumption constraints on utility loss is more delicate and relies on the investor's risk aversion level.
\end{singlespace}    
\end{abstract}


\section{Introduction}

Portfolio selection has been one of the central topics in modern finance. \citet{Markowitz1952} and \citet{Merton1969} are two groundbreaking works in this field and formulate portfolio selection as stochastic optimisation problems under the mean-variance criterion and the utility maximisation criterion, respectively. While \citet{Markowitz1952} focuses on the static mean-variance portfolio selection problem, it has been extended to a wide range of continuous-time versions (see, for example, \citet{ZhouLi2000}). The early investigation of \citet{Merton1969} on optimal investment-consumption strategies starts from a continuous-time setting. Under the continuous-time framework, portfolio selection problems are commonly tackled by stochastic control techniques, in which model coefficients in the state processes (for example, wealth and stochastic factors) are assumed to be known or can be estimated from real market data. For the latter scenario, a widely accepted approach proceeds as follows: 1) solving the stochastic control problem; and 2) plugging the estimated parameters into the optimal strategies obtained in the previous step. This approach leads to the so-called plug-in strategies. However, the performance of the plug-in strategies is highly sensitive to estimation errors. It has long been documented in empirical studies, for example \citet{Michaud1989}, that the mean-variance portfolios are unstable and perform poorly in terms of their out-of-sample mean and variance. As pointed out by \citet{Merton1980}, it is notoriously difficult to accurately estimate the sample mean of asset returns. When the number of stocks is large, \citet{Dubois-Veraart2015} show that the simple plug-in strategy can result in large utility loss due to accumulated estimation errors across the positions in multiple stocks.

Further exacerbating the difficulty of resolving the portfolio selection problem is the \textit{uncertainty} around the dynamics of stochastic factors influencing the financial market and, consequently, the investor's attitude towards this ambiguity (as illustrated, for example, by the famous \citet{Ellsberg1961} paradox). One formulation of this type of uncertainty is the situation where the investor has a best guess, or reference, model for the dynamics of the financial state variables/processes, but does not completely trust it. The investor then makes the investment and consumption decisions by considering models ``similar'' to the reference model, but the optimisation criterion involves a penalty that increases as the alternative model becomes 
more dissimilar from the investor's reference model, as usually measured by the relative entropy. By doing so, the investor seeks an optimal investment-consumption strategy that is robust to the uncertainty surrounding the financial model by taking the worst-case outcome, among all alternative models, of the optimisation problem. Indeed, this is the spirit of the robust control approach, based on the notion of multiplier preferences, pioneered by \citet{AndersonHansenSargent2003}.\footnote{We note that the robust control approach is one of the various ways in which model or Knightian uncertainty is incorporated in economic agents' decision making. We refer to the surveys by \citet{EpsteinSchneider2010}, \citet{HansenMarinacci2016}, \citet{GuidolinRinaldi2013}, and \citet{Marinacci2015}, and the references therein, for an overview of other methods of incorporating ambiguity in financial decision making. 
In particular, the framework of \citet{AndersonHansenSargent2003}, which involves a robust utility functional with a relative entropy penalty term, has been generalised via the notion of variational preferences introduced by \citet{MaccheroniMarinacciRustichini2006}. In turn, the variational preferences framework of \citet{MaccheroniMarinacciRustichini2006} also extends the multiple priors model of \citet{GilboaSchmeidler1989} and \citet{ChenEpstein2002} by including an explicit ``ambiguity index'' in the agent's objective function. We note that the variational preferences framework has been adapted into various settings in the context of robust optimal investment problems; see, for example, \citet{HHSchied2007}, \citet{LaevenStadje2014}, and \citet{yang2024optimal}.}

The robust control approach is one of the key remedies to overcome the difficulty caused by model uncertainty in portfolio selection problems. Under the robust control approach, portfolio selection problems are formulated as minimax problems, in which an investor aims to maximise an expected utility (or more generally, performance functional) in a worst-case scenario among a set of alternative probability measures. In the literature, there are two main specifications of robustness that are widely adopted to capture the investor's ambiguity aversion: multiplier robustness and locally constrained robustness. In both specifications, the relative entropy is a widely adopted measure to quantify the distance between a reference probability measure and alternative measures, which reflects the degree of model misspecification.\footnote{The use of relative entropy to quantify the distance between the probability measures implies that the decision-maker only considers alternative probability measures which are absolutely continuous with respect to the reference probability measure. Other notions of distance, such as the Wasserstein distance and optimal transport cost functions, have also been employed to characterise the so-called \textit{ambiguity sets} (or sets of alternative probability measures) in the context of optimisation and portfolio selection; see, for example, \citet{BlanchetMurthy2019} and \citet{PflugWozabal2007}.} In the former specification (see, for example, \citet{AndersonHansenSargent2003}, \citet{Maenhout2004}, and \citet{yan2020robust}), the relative entropy is added as a penalty term to the investor's objective function. In the latter specification (see, for example, \citet{Ait-SahaliaMatthys2019} and \citet{TrojaniVanini2002}), the relative entropy is constrained to not exceed a given level. Dating back to the pioneering work of \citet{AndersonHansenSargent2003} on the continuous-time robust control approach, robustness has been applied in a variety of decision-making problems in finance and economics and this idea is culminated in the monograph \cite{HansenSargent2011} by the two Nobel laureates: Lars Peter Hansen and Thomas Sargent. In the context of portfolio selection, \citet{UppalWang2003} extends the framework of \citet{AndersonHansenSargent2003} to allow investors to have varying levels of ambiguity aversion regarding the marginal distribution of any subset of the available investment options. Furthermore, \citet{Maenhout2004} modifies the methodology in \citet{AndersonHansenSargent2003} by proposing homothetic robustness which preserves the tractability of the corresponding robust portfolio selection problem.

The robust optimal control framework of \citet{AndersonHansenSargent2003} has since been extended and adapted into various settings (typically pertaining to assumptions on the financial market model and/or the investor's utility function) in the context of optimal investment and portfolio selection problems. \citet{Maenhout2006} considers the robust maxmimisation of (power) utility on terminal wealth with stochastic investment opportunities, as modelled by a mean-reverting risk premium process. 
%In this setting, the expected return of the risky asset is likewise mean-reverting and that returns are not independent and identically distributed. 
\citet{Liu2010} tackles a robust optimal investment-consumption problem on an infinite time horizon under recursive preferences, assuming also that the expected returns are modelled as a mean-reverting process.\footnote{\citet{Liu2011} further considers an optimal investment-consumption problem over a finite-time horizon when the expected returns are modelled as a hidden Markov model. However, in contrast to \citet{Liu2010} and \citet{Maenhout2004},  \citet{Liu2011} uses a locally constrained robustness approach (by imposing an essential upper bound on the drift distortion process), following the multiple priors formulation of \citet{GilboaSchmeidler1989} and \citet{ChenEpstein2002}.} %\footnote{We refer, for example, to \citet{Scheid2007} (and the references therein) for a distinction between robust utility functionals with and without a penalty term in the context of optimal investment problems. We emphasize that the robust control framework of \citet{AndersonHansenSargent2003}, as a special case of variational preferences introduced by \citet{MaccheroniMarinacciRustichini2006}, involves a robust utility functional with a penalty given by the relative entropy of alternative probability measures.}
\citet{FlorLarsen2014} employ the homothetic preferences framework of \citet{Maenhout2004} to the robust optimisation of a power utility investor's terminal wealth when the short-term interest rate is stochastic. Under a constant market price of risk, the expected return on the risky and risk-free assets are stochastic due to their dependence on the short rate. %\citet{Ait-SahaliaMatthys2019} consider a robust optimal consumption problem with recursive preferences over an infinite time horizon when the risky asset price is modelled as a jump-diffusion process. Here, perturbations from the reference measure induce distortions in the drift and the jump characteristics of the asset price process. 
\citet{WeiYangZhuang2023} further extend the financial market setting of the robust investment-consumption problem with recursive preferences by assuming the price of risky asset is driven by a jump-diffusion stochastic volatility model. As the market is incomplete, the authors assume that the investor has access to derivatives whose underlying is the risky asset. The investor's objective function is consequently formulated using the robust homothetic preferences of \citet{Maenhout2004}. In terms of assumptions on the financial market or investment opportunities available to the investor, the aforementioned papers allow for certain financial factors (e.g. the expected returns on the risky asset, the interest rate, the volatility of the risky asset price) to be modelled by a given stochastic process. All other model parameters are otherwise completely known (as a constant or a deterministic function of time). Furthermore, aside from admissibility conditions on the controls (investment and/or consumption decisions of the investor), no other constraints are imposed.

This paper considers a general robust portfolio selection problem under the multiplier robustness specification of \citet{AndersonHansenSargent2003} and \citet{Maenhout2004}. The generality comes from two aspects. First, both investment and consumption strategies are subject to constraints. More specifically, we assume that both investment and consumption strategies are constrained within closed convex sets. This setup includes many special cases, such as 1) the investor is not allowed to hold any short positions in the risky assets, 2) the investor is only allowed to consume up to her current level of wealth, and 3) the investor has to maintain subsistence consumption. When there is no model uncertainty, optimal investment and consumption problems with constraints on investment and consumption (or either one of the two) have been investigated by \citet{CvitanicKaratzas1992}, \citet{HuImkellerMuller2005,hu2022optimal}, \citet{KraftSeifriedSteffensen2013}, \citet{MaYiGuan2018}, and \citet{Zariphopoulou1994}, among others. However, our results show that when model ambiguity is factored into the investor's preferences, the investor's ambiguity aversion (i.e. preference for robustness) significantly impacts the optimal consumption strategy when the investment constraints are present. We establish this formally in Proposition \ref{prop:av-coefficnet}. Second, all model coefficients are allowed to be random. Thus, there is joint model uncertainty from multiple sources, such as equity return, stochastic interest rate, and stochastic volatility. In this general modelling framework, we utilise the stochastic Hamilton-Jacobi-Bellman-Isaacs (HJBI) equation and backward stochastic differential equations (BSDE) approach to tackle the constrained robust control problem. For that purpose, we first establish the existence and uniqueness results for the solutions to a set of quadratic BSDEs. By applying those results in conjunction with bounded mean oscillation (BMO) martingale theory (see \citet{Kazamaki2006}), we provide a verification theorem for the optimality of obtained strategies. We then study a special example in which all model coefficients are deterministic. In that case, the related quadratic BSDEs reduce to nonlinear ordinary differential equations (ODEs), which allows us to perform a comparative static analysis on the robust optimal strategies and value functions with respect to investment and consumption constraints. % \sout{We further our analysis by considering the same problem with a recursive utility.}

To the best of our knowledge, \citet{YangLiangZhou2019} is the only other work that studies the robust portfolio selection problem with constraints on both investment and consumption strategies and with random model coefficients. However, their characterisation of model uncertainty is different: the expected returns and the covariance matrix are uncertain and the uncertainty set is built on the unknown returns and covariance matrix.\footnote{This formulation of robust utility maximization problems involves taking the worst-case scenario over an uncertainty set around the ambiguous parameters (e.g. ambiguous returns or ambiguous volatility) or over a possibly non-dominated family of probability measures arising from the uncertainty set on the ambiguous parameters. We refer, for example, to \citet{BiaginiPinar2017}, \citet{EpsteinJi2013}, and \citet{LinRiedel2021} for detailed discussions of robust portfolio selection problems of this type.} 
%{\color{red} [I added a brief footnote on this type of robust control problem just so that we can say we are aware that this alternative approach exists.]} 
Furthermore, the robust utility maximisation problem considered by \citet{YangLiangZhou2019} does not explicitly involve the degree of the investor's ambiguity aversion. In our analysis, a key contribution we make to the literature is the analysis of the interplay between the investor's ambiguity aversion and the investment and portfolio constraints on the robust optimal strategies. Furthermore, as we employ the robust optimal control framework of \citet{AndersonHansenSargent2003} and \citet{Maenhout2004}, the core mathematical methodology of this paper is distinct from that of \citet{YangLiangZhou2019}. In particular, while \citet{YangLiangZhou2019} approach their problem via using the martingale argument that is firstly introduced in \citet{HuImkellerMuller2005}, we apply the stochastic HJBI equation and BSDEs to solve our robust control problem. 
%{\color{red} [The previous sentence is unclear; do we mean to say ``In particular, while \citet{YangLiangZhou2019} approach their problem via using a martingale argument, \textbf{we proceed, as in \citet{HuImkellerMuller2005}, by applying} the stochastic HJBI equation and BSDEs to solve our robust control problem''?]} 
Unlike \citet{YangLiangZhou2019} where the optimal strategies and value function are deterministic and can be characterised via the solutions to nonlinear ODEs, the robust optimal strategies and value function in the general case of our paper 
are stochastic and are thus represented by the solutions to nonlinear BSDEs.

The rest of the paper is structured as follows. Section \ref{sec:statement} introduces the mathematical model for the financial market and formulates the robust investment-consumption problem for a power utility investor with constraints on the investment and consumption strategies. %We assume that the financial market consists of a risk-free asset and $m \in \mathbb{N}$ risky assets whose price processes are modelled as geometric Brownian motions with stochastic expected returns and volatilities. 
In Section \ref{sec:main}, we solve the problem by the stochastic HJBI equation and represent the robust optimal investment and consumption strategies by unique solutions to corresponding BSDEs, as accomplished in Proposition \ref{prop:HJBI-solution}. We also establish a verification theorem (Theorem \ref{thm:main}) via the theory of BMO martingales and show in Proposition \ref{prop:utilloss} that the investor incurs a nonnegative utility loss when she ignores model uncertainty and pursues the resulting (sub-optimal) strategy. Section \ref{sec:specialdeterministic} discusses the special case of the robust investment-consumption problem with deterministic coefficients. With this simplification, the BSDEs characterising the value function and the optimal strategies reduce to ODEs; see Corollary \ref{cor:DeterministicCaseSolution}. Furthermore, in the deterministic case, we prove that, if short-selling is not allowed, the optimal consumption strategy is decreasing (resp. increasing) in the ambiguity aversion parameters when the coefficient of relative risk aversion is greater (resp. less) than unity. Furthermore, we formally establish the effect of investment and consumption constraints on the optimal consumption strategy and value function in Proposition \ref{prop:det-comparison}. These results are illustrated in a series of numerical experiments. Finally, Section \ref{sec:conclusion} concludes the paper.

\section{Statement of the Problem}
\label{sec:statement}

In this section, we introduce the stochastic models of the financial market and formulate the robust investment-consumption problem with portfolio constraints and under a stochastic environment.

We consider a probability space $(\Omega, \cal F, \mathbb F, \mathbb P)$, satisfying the usual conditions, which supports an $n$-dimensional standard Brownian motion $\{ W (t) \}_{t\in[0,T]}$ defined on this probability space. Here, 
$T$ denotes the terminal time of the time horizon, $\mathbb F := \{ {\cal F}_t \}_{t\in[0,T]}$ is a complete filtration, and $\mathbb P$ is a reference probability measure. In what follows, we study a financial market with $m+1$ primitive assets, including one risk-free asset and $m$ risky assets (i.e. stocks), where $m \leq n$,  We allow the financial market to be incomplete in the sense that $m < n$.
Suppose that the price process of the risk-free asset and the $i^{th}$ stock are respectively governed by
\begin{align*}
d S_0 (t) & = r (t) S_0 (t) dt, & S_0(0)  &= 1 , \\
d S_i (t) & = S_i(t) \bigg [ \mu_i(t) dt + \sum^n_{j=1} \sigma_{ij} (t) d W_j (t)\bigg ], & S_i(0)  &= s_i > 0 ,
\end{align*}
where $r (t)$ is the risk-free interest rate, $\mu_i(t)$ and $\sigma_{ij}(t)$ are the appreciation rate and the volatility, respectively, of the stock, for each $i = 1, 2,\ldots, m$ and $j = 1, 2,\ldots, n$. To simplify our notation, we denote by $B (t) : = (\mu_1 (t), \mu_2 (t), \ldots, \mu_m (t))^\top - r (t) {\bf 1}_m$ and $\sigma (t) : =[\sigma_{ij} (t)]_{i=1,2,\ldots,m;j=1,2,\ldots,n}$ the risk premium vector and the volatility matrix, respectively.

Suppose that an investor can invest in the aforementioned financial market. Denote by $\pi_i (t)$ the proportion of wealth allocated to the $i^{th}$ stock
at time $t$, for $i=1,2,\ldots,m$, and by $c (t)$ the proportion of wealth consumed at time $t$. Let $\pi(t) : = (\pi_1(t), \pi_2 (t), \ldots, \pi_m(t))^\top$ denote the vector of the proportions invested in the $m$ stocks available in the financial market. Hence, the investor's wealth process satisfies the following stochastic differential equation (SDE):
\begin{align*}
d X(t) = X(t) \left [ r(t) + \pi(t)^\top B(t) - c(t) \right ] dt
+ X(t) \pi(t)^\top \sigma(t) d W(t), \quad X(0) = x_0 > 0 .
\end{align*}
Here, $\{\pi(t)\}_{t\in[0,T]}$ and $\{c (t)\}_{t\in[0,T]}$ are referred to as the investment and consumption strategies, respectively.

This paper assumes that there are constraints on both investment and consumption strategies. More specifically, we assume that the investor's investment and consumption strategies are constrained in two closed convex sets $\Pi \in \mathbb R^m$ and $\mathscr C \in [0 , \infty]$, respectively. Without loss of generality, we can further specify $\mathscr C$ by an interval $[\underline c, \overline c]$ that is a subset of $[0 , \infty]$, i.e. $\mathscr C : = [\underline c, \overline c] \subset [0 , \infty]$. In other words, the consumption strategy is constrained between a (non-negative) floor $\underline c$ and a ceiling $\overline c$ (which is allowed to be infinity), corresponding to the lower and upper bounds of the consumed wealth proportion, respectively. This specification of  $\mathscr C$ will facilitate the investigation of the interaction between portfolio constraints and model uncertainty.

In practice, the investor is concerned about model misspecification that is caused by the deviation from the reference probability measure $\mathbb P$. To quantify such deviation, we define a set of alternative probability measures via the Radon-Nikodym derivative:
\begin{align*}
\frac{d {\mathbb Q}}{d {\mathbb P}} \bigg |_{{\cal F}_T} = \Lambda (T)
:= \exp \bigg \{ - \frac{1}{2} \int^{T}_0 |\phi(t)|^2 d t + \int^{T}_0 \phi^\top (t) d W (t) \bigg \} ,
\end{align*}
where $\{ \phi (t) \}_{t\in[0,T]}$ is called the distortion process. Whenever needed, we highlight the dependence of the alternative measure ${\mathbb Q}$ on $\phi$ by writing ${\mathbb Q} (\phi)$. Denote by ${\mathscr Q}$ the set of all alternative probability measures equivalent to $\mathbb P$.
%We can imagine that the alternative measure ${\mathbb Q}$ or the distortion process $\{ \phi (t) \}_{t\geq0}$ is chosen by a fictitious player, e.g. nature or god.

Under the measure ${\mathbb Q}$, the stochastic process $\{ W^{\mathbb Q} (t) \}_{t\in[0,T]}$ defined by
$W^{\mathbb Q} (t) := W (t) - \int^t_0 \phi (s) d s ,$
is an $n$-dimensional standard Brownian motion.

Note that the relative entropy between $\mathbb P$ and $\mathbb Q$ is computed as
\begin{align}
	{\mathbb E}^{\mathbb Q} \bigg [ \log \bigg ( \frac{d {\mathbb Q}}{d {\mathbb P}} \bigg |_{{\cal F}_T} \bigg ) \bigg ]
	&= {\mathbb E}^{\mathbb Q} \bigg [ - \frac{1}{2} \int^{T}_0 |\phi(t)|^2 d t + \int^{T}_0 \phi^\top (t) d W (t) \bigg ] \nonumber \\
	&= {\mathbb E}^{\mathbb Q} \bigg [ \frac{1}{2} \int^{T}_0 |\phi(t)|^2 d t + \int^{T}_0 \phi^\top (t) d W^{\mathbb Q} (t) \bigg ] = {\mathbb E}^{\mathbb Q} \bigg [ \frac{1}{2} \int^{T}_0 |\phi(t)|^2 d t \bigg ] .
\end{align}
This quantity will be related to a penalty term to be added to the investor's objective function.

To simplify our presentation and analysis, we denote by
$\theta (t) : = \sigma (t)^\top \Sigma (t)^{-1} B (t)$
the market price of risk, where $\Sigma (t) : = \sigma (t) \sigma (t)^\top$ is the covariance matrix of the risky assets,
and define
$p (t) := \sigma (t)^\top \pi (t)$
as the investor's risk exposure to Brownian shocks, with the $i^{th}$ entry representing the proportion of wealth exposed to the random shocks arising from the $i^{th}$ Brownian motion, 
for each $i = 1, 2, \ldots, n$. Define $\Gamma (t) : = \sigma (t)^\top \Pi$, which is also a convex set for each $(t, \omega) \in [0, T] \times \Omega$, inherited from the convexity of $\Pi$.
Thus, the risk exposure $p (t)$ is constrained to take values in the set $\Gamma (t)$. As in \citet{HuImkellerMuller2005}, we consider that the investor chooses the risk exposure
$p (t)$, at each $t \in [0, T]$. Whenever there is no risk of confusion, we also call $\{ p (t) \}_{t\in[0,T]}$ the investment strategy. 

Using $p (t)$ and $\theta (t)$, we rewrite the wealth equation under $\mathbb Q$ as follows:
\begin{align}\label{eq:wealth-equation}
d X(t) = X (t) \left [ r (t) + p (t)^\top ( \theta (t) + \phi (t) ) - c(t) \right ] dt + X (t) p (t)^\top d W^{\mathbb Q} (t) .
\end{align}
We assume that the investor aims to maximise the expected discounted utility of intertemporal consumption and terminal wealth, which
is subject to a penalty measured by the relative entropy between $\mathbb Q$ and $\mathbb P$. Therefore, the dynamic objective function
reads as follows:
\begin{align}\label{eq:problem}
V(t,x) = \sup_{(p, c) \in {\cal A}} \inf_{\phi \in \cal{B}}{\mathbb E}_{t,x}^{{\mathbb Q}}
\Bigg[ \int_t^T e^{-\int_t^s \rho(\nu) d\nu} U_1 (c(s) X (s)) ds
+ e^{-\int_t^T \rho(\nu) d\nu} U_2 (X(T)) + \sum^n_{i = 1} \int_t^T e^{-\int_t^s \rho(\nu) d\nu} \frac{\phi_i^2(s)}{2 \Psi_i(s)} ds \Bigg] ,
\end{align}
where ${\mathbb E}^{\mathbb Q}_{t,x} [\cdot]$ denotes the conditional expectation under $\mathbb Q$ given ${\cal F}_t$ and $X (t) = x$, ${\cal A}$ and ${\cal B}$ are the admissible sets of investment-consumption strategies and distortion processes (to be defined formally in Definition \ref{def:adm}), and $\rho (u)$ denotes
the investor's subjective discount rate. The third term in equation \eqref{eq:problem} is the so-called penalty term. Inspired by \citet{Maenhout2004}, we choose the normaliser on the penalty term as $\Psi_i (s)=\frac{\eta_i}{(1-\gamma)V(s,X(s))}$, in which $\eta_i \geq 0$ denotes the investor's ambiguity aversion coefficient corresponding to the $i^{th}$ Brownian motion.
The farther away the alternative measure $\mathbb Q$ is from the reference measure $\mathbb P$, the larger is the penalty term in \eqref{eq:problem}.

\begin{remark}
The infimum in the objective function \eqref{eq:problem} is taken with respect to the distortion process $\phi$ over the admissible set $\cal B$. This is equivalent to taking the infimum with respect to the probability measure $\mathbb Q (\phi)$ over a set of equivalent probability measures, denoted by ${\cal Q}_{\cal B}$, which is a subset of ${\cal Q}$.
\end{remark}

In this paper, we assume that the investor's preference is described by power utility functions
\begin{eqnarray*}
	U_1 (c x) = \frac{(cx)^{1-\gamma}}{1-\gamma} \quad \mbox{and} \quad U_2 (x) = \beta \frac{x^{1-\gamma}}{1-\gamma} ,
\end{eqnarray*}
where $\gamma \in (0, 1) \cup (1, \infty)$ denotes the coefficient of relative risk aversion and $\beta > 0$ is the weight of the utility derived from terminal wealth.

To simplify our presentation, we introduce
$H : = \mbox{diag} \big [ ( \eta_1, \eta_2, \cdots, \eta_n ) \big ] ,$
which implies that
$H^{-1} = \mbox{diag} \big [ \big ( \frac{1}{\eta_1}, \frac{1}{\eta_2}, \cdots, \frac{1}{\eta_n} \big ) \big ]$.
Thus, the penalty term can be rewritten concisely as follows:
\begin{align*}
	{\mathbb E}_{t}^{\mathbb Q}
	\Bigg[\sum^n_{i = 1} \int_t^T e^{-\int_t^s \rho(\nu) d\nu} \frac{\phi_i^2(s)}{2 \Psi_i(s)} ds \Bigg] =
	{\mathbb E}_{t}^{\mathbb Q}
	\Bigg[\frac{1-\gamma}{2} \int_t^T e^{-\int_t^s \rho(\nu) d\nu} \phi (s)^\top H^{-1} \phi (s) V (s, X (s)) ds \Bigg] .
\end{align*}

In the general model setting, we assume that the model coefficients satisfy the following conditions.

\begin{assumption}\label{ass:bound}
	The model coefficients $\{r (t)\}_{t\in[0,T]}$, $\{\rho (t)\}_{t\in[0,T]}$, $\{\mu_i(t)\}_{t\in[0,T]}$ and $\{\sigma_{ij}(t)\}_{t\in[0,T]}$ are $\mathbb F$-predictable and bounded processes, for $i=1,2,\ldots,m$ and $j=1,2,\ldots,n$, and the covariance matrix $\Sigma (t) = \sigma (t) \sigma (t)^\top$ is invertible, for any $t \in [0, T]$.
\end{assumption}

Let $k$ be a generic natural number. To simplify our notation, we define the following spaces of stochastic processes that are frequently used in the paper:
\begin{itemize}
	\item ${\cal L}^q_{\mathbb P} (0, T; \mathbb R^k)$: the space of $\mathbb R^k$-valued, $\mathbb F$-adapted processes such that
	\begin{align}
		| f |_{{\cal L}^q_{\mathbb P} (0, T; \mathbb R^k)} : = \left\{{\mathbb E} \bigg[ \bigg(\int^T_0 | f (t) |^2 d t\bigg)^{\frac{q}{2}} \bigg]\right\}^{\frac{1}{q}} < \infty , \quad & \mbox{where} \ q > 0 .
	\end{align}
	\item ${\cal S}^q_{\mathbb P} (0, T; \mathbb R^k)$: the space of $\mathbb R^k$-valued, $\mathbb F$-adapted, c\`adl\`ag processes $f$ such that
	\begin{align*}
		| f |_{{\cal S}^q_{\mathbb P} (0, T; \mathbb R^k)} : =
		\left\{
		\begin{aligned}
			& \left\{{\mathbb E} \bigg[ \sup_{t \in [0,T]} | f (t) |^q \bigg]\right\}^{\frac{1}{q}} < \infty , \quad & \mbox{if} \ q \in (0, \infty) , \\
			& \sup_{t \in [0,T]} \{ | f (t) |_\infty \} < \infty , \quad & \mbox{if} \ q = \infty ,
		\end{aligned}
	    \right.
	\end{align*}
where $| f (t) |_\infty$ denotes the essential supremum of $f(t)$ under $\mathbb P$.
\end{itemize}
If we replace the probability measure $\mathbb P$ by
an alternative measure $\mathbb Q$, we can define the spaces ${\cal L}^q_{\mathbb Q} (0, T; \mathbb R^k)$ and ${\cal S}^q_{\mathbb Q} (0, T; \mathbb R^k)$ similarly. 

Next, we define formally the admissible strategies that are of interest to us.

\begin{definition}\label{def:adm}
An investment-consumption strategy and a distortion process are admissible if
\begin{enumerate}
\item $u:=(p, c)$ and $\phi$ are $\mathbb F$-predictable processes taking values in $\Gamma \times \mathscr C$ and $\mathbb R^n$, respectively;
\item the wealth equation \eqref{eq:wealth-equation} associated with $u$ and $\phi$ admits a unique non-negative solution
such that $X (\cdot) \in {\cal S}^2_{\mathbb Q} (0, T; \mathbb R)$ and ${\mathbb E}^{\mathbb Q} [ \sup_{t \in [0, T]} |X (t)|^q ] < \infty$,
where $\mathbb Q := \mathbb Q (\phi) \in {\cal Q}$, for any $q \leq 2$;
\item the Radon-Nikodym derivative process $\{\Lambda (t)\}_{t\in[0,T]}$ associated with $\phi$ is a $\mathbb P$-martingale;
\item the following integrability conditions hold true: $p (\cdot) X (\cdot) \in {\cal L}^2_{\mathbb Q} (0, T; \mathbb R^n)$,  $c (\cdot) X (\cdot) \in {\cal L}^1_{\mathbb Q} (0, T; \mathbb R)$ and $\phi (\cdot) \in {\cal L}^2_{\mathbb Q} (0, T; \mathbb R^n)$, where $\mathbb Q = \mathbb Q (\phi) \in {\cal Q}$.
\end{enumerate}
The sets of all admissible investment-consumption strategies and distortion processes are denoted by ${\cal A}$ and ${\cal B}$, respectively.
\end{definition}

If it exists, an admissible investment-consumption strategy that achieves the supremum in Problem \eqref{eq:problem} is called a robust optimal investment-consumption strategy (or, in short, an optimal investment-consumption strategy). An admissible distortion process that achieves the infimum in Problem \eqref{eq:problem}, if it exists, is called an optimal distortion process.

In the remainder of the paper, we will frequently use the distance and projection operators, which are defined formally as follows. For any $v \in \mathbb R^m$, we denote by $\mbox{Dist}_{\Gamma} \{v\}
: = \min_{v^\prime\in{\Gamma}} |v-v^\prime|$ the distance between the vector $v$ and the closed convex set $\Gamma \subset \mathbb R^m$, and by $\mbox{Proj}_{\Gamma} \{v\} : = \argmin_{v^\prime\in{\Gamma}} |v - v^\prime|^2$ the orthogonal projection
of $v$ onto ${\Gamma}$. One can refer to \citet{HuImkellerMuller2005} for similar definitions of $\mbox{Dist}_{\Gamma} [\cdot]$ and $\mbox{Proj}_{\Gamma} [\cdot]$.

\section{Main Results}\label{sec:main}

In this section, we first identify the candidate of robust optimal investment and consumption strategies by solving the stochastic Hamilton-Jacobi-Bellman-Isaacs (HJBI) equation. Then we provide a verification theorem to confirm the optimality of the obtained candidate, and we discuss the utility loss when model uncertainty is ignored.

\subsection{Robust Optimal Strategies}

In this subsection, we solve Problem \eqref{eq:problem} by using the stochastic HJBI equation in conjunction with backward stochastic differential equations (BSDEs). To that end, we conjecture that the value function of the problem has a parametric form
$V (t, x) = G (t, x, Y (t))$, where $G$ belongs to the class of functions that are continuously differentiable
in the first argument and twice continuously differentiable in the second and third arguments, i.e. $G (\cdot, \cdot, \cdot) \in {\cal C}^{1, 2, 2} ([0, T] \times \mathbb R^2)$,
and $Y$ is the first component of the solution pair $(Y, Z)$ to a BSDE:
\begin{align}\label{eq:BSDE}
d Y (t) = - f (t, Y (t), Z (t)) d t + Z (t)^\top d W (t) , \quad Y (T) = \xi .
\end{align}
Note that the driver $f$ and the terminal value $\xi$ are unknown at the moment, and they will be determined in Proportion \ref{prop:HJB-solution} below.

To solve Problem \eqref{eq:problem}, we consider the following stochastic HJBI equation:
\begin{align}\label{eq:HJBI}
\sup_{u \in {\cal A}} \inf_{\phi \in {\cal B}} \bigg \{ {\cal L}^{(u, \phi)} [G (t, x, Y (t))]
+ \frac{1-\gamma}{2} \phi^\top H^{-1} \phi G (t, x, Y (t)) + U_1 (cx) \bigg \} = 0 ,
\quad G (T, x, \xi) = U_2 (x) ,
\end{align}
where the infinitesimal generator ${\cal L}^{(u, \phi)}$ acting on $G (t, x, Y (t))$ is defined as follows:
\begin{align}\label{eq:inf-generator}
	{\cal L}^{(u, \phi)} [G (t, x, Y (t))]
	:=& - \rho (t) G (t, x, Y (t)) + G_t (t, x, Y (t))
	+ x \left [ r (t) + p^\top (\theta (t) + \phi) - c \right ] G_x (t, x, Y (t)) \nonumber \\
	& - G_y (t, x, Y (t)) [f (t, Y (t), Z (t)) - Z^\top (t) \phi  ]
	+ \frac{1}{2} x^2 | p |^2 G_{xx} (t, x, Y (t)) \nonumber \\
	& + x p^\top Z (t) G_{xy} (t, x, Y (t))
	+ \frac{1}{2} |Z (t)|^2 G_{yy} (t, x, Y (t)) .
\end{align}
The stochasticity of \eqref{eq:HJBI} arises from its link to BSDE \eqref{eq:BSDE}. Unlike the classical version of stochastic HJB equations proposed by \citet{Peng1992}, the special structure of the value function allows us to separate the randomness of \eqref{eq:HJBI} from the time and state variables via BSDE \eqref{eq:BSDE}. One can refer
to \citet{RiederWopperer2012} and \citet{ShenWei2016} for similar stochastic HJB(I) equations.

The next proposition characterises the solution to the stochastic HJBI equation \eqref{eq:HJBI} and specifies the driver and the terminal value of BSDE \eqref{eq:BSDE}. In what follows, whenever there is no risk of confusion, 
we omit the time index ``$(t)$" of stochastic processes to simplify the notation. 

\begin{proposition}\label{prop:HJBI-solution}
The supremum and the infimum in the stochastic HJBI equation \eqref{eq:HJBI} are achieved by
\begin{align}\label{eq:p-c-star}
p^* (t) & = \bigg ( I + \frac{1}{\gamma} H \bigg )^{-\frac{1}{2}} \Proj_{\widehat \Gamma} \bigg \{ \bigg ( I + \frac{1}{\gamma} H \bigg )^{-\frac{1}{2}} \bigg [ \frac{\theta (t)}{\gamma} + \bigg(I - \frac{H}{1-\gamma}\bigg) \frac{Z (t)}{Y (t)} \bigg ] \bigg \} , \qquad
c^* (t) = \frac{1}{Y (t)} \vee \underline c \wedge \overline c , \\
\label{eq:phi-star}
\phi^* (t) & = -\frac{\gamma H}{1-\gamma} \frac{Z (t)}{Y (t)} - H \bigg ( I + \frac{1}{\gamma} H \bigg )^{-\frac{1}{2}} \Proj_{\widehat \Gamma} \bigg \{ \bigg ( I + \frac{1}{\gamma} H \bigg )^{-\frac{1}{2}} \bigg [ \frac{\theta (t)}{\gamma} + \bigg(I - \frac{H}{1-\gamma}\bigg) \frac{Z (t)}{Y (t)} \bigg ] \bigg \} ,
\end{align}
where ${\widehat \Gamma} := (I + \frac{H}{\gamma})^\frac{1}{2} \Gamma$ denotes an auxiliary constraint set in $\mathbb R^n$
and $\Proj_{\widehat \Gamma} \{\cdot\}$ denotes the orthogonal projection operator mapping an $\mathbb R^n$-valued vector onto ${\widehat \Gamma}$.

The stochastic HJBI equation \eqref{eq:HJBI} holds for all $(t, \omega, x) \in [0, T] \times \Omega \times \mathbb R^+$, if
\begin{align}\label{eq:G}
G (t, x, Y (t)) = \frac{x^{1-\gamma}}{1-\gamma} \times [Y (t)]^{\gamma} ,
\end{align}
with the terminal value $Y (T) = \beta^{\frac{1}{\gamma}}$,
and the driver of BSDE \eqref{eq:BSDE} is given by
\begin{align}\label{eq:driver}
f (t, Y, Z) =& \ \frac{Y}{Y \wedge \frac{1}{\underline c} \vee \frac{1}{\overline c}} + \frac{1}{\gamma} \bigg [ - \rho + (1-\gamma) r
+ \frac{1 - \gamma}{2 \gamma} \theta^\top \bigg ( I + \frac{1}{\gamma} H \bigg )^{-1} \theta \bigg ] Y
+ \frac{1}{\gamma} \theta^\top \bigg [ \bigg ( I + \frac{1}{\gamma} H \bigg )^{-1} - \gamma I \bigg ] Z \nonumber \\
& - \frac{1}{2\gamma(1-\gamma)} \frac{Z^\top (I + \frac{1}{\gamma}H)^{-1} H Z}{Y} - \frac{1-\gamma}{2} \Dist^2_{\widehat \Gamma} \bigg \{ \bigg ( I + \frac{1}{\gamma} H \bigg )^{-\frac{1}{2}} \bigg [ \frac{\theta}{\gamma} + \bigg(I - \frac{H}{1-\gamma}\bigg) \frac{Z}{Y} \bigg ] \bigg \} Y ,
\end{align}
where $\Dist_{\widehat \Gamma} \{ \cdot \}$ denotes the distance metric from an $\mathbb R^n$-valued vector to ${\widehat \Gamma}$.
\end{proposition}

\begin{proof}
First, applying the first-order condition to \eqref{eq:HJBI} with respect to (w.r.t.) $c$ and $\phi$ gives
% \begin{align*}
% - G_x (t, x, Y) + (xc)^{-\gamma} = 0 ,
% \end{align*}
% and
% \begin{eqnarray*}
% G_y (t, x, Y) Z^\top + x G_x (t, x, Y) p^\top + (1-\gamma) \phi^\top H^{-1} G (t, x, Y) = 0 .
% \end{eqnarray*}
\begin{align*}
    0 & = - G_x (t, x, Y) + (xc)^{-\gamma} , \\
    0 & = G_y (t, x, Y) Z^\top + x G_x (t, x, Y) p^\top + (1-\gamma) \phi^\top H^{-1} G (t, x, Y).
\end{align*}
By noting that $c$ is constrained in $[\underline c, \overline c]$ and there is no constraint on $\phi$, 
we solve the above two equations and obtain
\begin{align}\label{eq:c*-foc}
c^* & = \bigg\{ \frac{1}{x} \big [ G_x (t, x, Y) \big ]^{-\frac{1}{\gamma}} \bigg \} \vee \underline c \wedge \overline c , \\
\label{eq:phi*-foc}
\phi^* & = - \frac{1}{1-\gamma} H \bigg [ \frac{G_y (t, x, Y)}{G (t, x, Y (t))} Z + \frac{x G_x (t, x, Y)}{G (t, x, Y (t))} p \bigg ] .
\end{align}
Recalling the terminal value of the stochastic HJBI equation \eqref{eq:HJBI}, we try the {\it ansatz}:
$G (t, x, Y (t)) = \frac{x^{1-\gamma}}{1-\gamma} \times [Y (t)]^{\gamma} ,$
with $Y (T) = \beta^\frac{1}{\gamma}$. Thus, substituting the above {\it ansatz}
into \eqref{eq:c*-foc} and \eqref{eq:phi*-foc} gives 
\begin{align}\label{eq:c*-foc-solution}
c^* (t) & = \frac{1}{Y(t)} \vee \underline c \wedge \overline c , \\
\label{eq:phi*-foc-inter-solution}
\phi^* (t) & = - \frac{H}{1-\gamma} \bigg [ \gamma \frac{Z (t)}{Y (t)} + (1-\gamma) p (t) \bigg ] .
\end{align}

To solve the constrained maximisation problem w.r.t. $p$, we collect all the terms involving $p$ and $\phi^*$ in \eqref{eq:HJBI} as follows:
{\small
\begin{align*}
& \ x G_x (t, x, Y) p^\top [ \theta + \phi^* ] + \frac{1}{2} x^2 G_{xx} (t, x, Y) |p|^2
+ x G_{xy} (t, x, Y) p^\top Z + G_{y} (t, x, Y) Z^\top \phi^* + \frac{1-\gamma}{2} (\phi^*)^\top H^{-1} \phi^* G (t, x, Y) \\
%& = x G_x (t, x, Y) p^\top \theta + \frac{1}{2} x^2 G_{xx} (t, x, Y) |p|^2
%+ x G_{xy} (t, x, Y) p^\top Z - \frac{1-\gamma}{2} (\phi^*)^\top H^{-1} \phi^* G (t, x, Y) \\
%& = x G_x (t, x, Y) p^\top \theta + \frac{1}{2} x^2 G_{xx} (t, x, Y) |p|^2
%+ x G_{xy} (t, x, Y) p^\top Z \\
%& \quad - \frac{G (t, x, Y)}{2 (1-\gamma)} \bigg [ \frac{G_y (t, x, Y)}{G (t, x, Y (t))} Z + \frac{x G_x (t, x, Y)}{G (t, x, Y (t))} p \bigg ]^\top H \bigg [ \frac{G_y (t, x, Y)}{G (t, x, Y (t))} Z + \frac{x G_x (t, x, Y)}{G (t, x, Y (t))} p \bigg ] \\
& = \gamma x^{1-\gamma} Y^\gamma \bigg \{ - \frac{1}{2} p^\top \bigg (I + \frac{H}{\gamma} \bigg) p
+ p^\top \bigg [ \frac{\theta}{\gamma} + \bigg(I - \frac{H}{1-\gamma}\bigg) \frac{Z}{Y} \bigg ] \bigg \} - x^{1-\gamma} Y^\gamma \frac{\gamma^2}{2 (1-\gamma)^2} \frac{Z^\top H Z}{Y^2} \\
& = - x^{1-\gamma} Y^\gamma \frac{\gamma }{2} \bigg\{ p - \bigg ( I + \frac{1}{\gamma} H \bigg )^{-1} \bigg [ \frac{\theta}{\gamma} + \bigg(I - \frac{H}{1-\gamma}\bigg) \frac{Z}{Y} \bigg ]\bigg\}^\top \bigg (I + \frac{H}{\gamma} \bigg )
\bigg\{ p - \bigg ( I + \frac{1}{\gamma} H \bigg )^{-1} \bigg [ \frac{\theta}{\gamma} + \bigg(I - \frac{H}{1-\gamma}\bigg) \frac{Z}{Y} \bigg ]\bigg\} \\
&\quad + x^{1-\gamma} Y^\gamma
\bigg \{  \frac{1}{2\gamma} \theta^\top \bigg ( I + \frac{1}{\gamma} H \bigg )^{-1} \theta + \frac{1}{1-\gamma} \theta^\top \bigg [ \bigg ( I + \frac{1}{\gamma} H \bigg )^{-1} - \gamma I \bigg ]\frac{Z}{Y} + \frac{\gamma}{2} \frac{Z^\top}{Y}  \bigg[ I - \frac{1}{\gamma (1-\gamma)^2} \bigg ( I + \frac{1}{\gamma} H \bigg )^{-1} H \bigg] \frac{Z}{Y} \bigg \} .
\end{align*}
}
Thus, the constrained maximisation problem w.r.t. $p$ on $\Gamma$ is equivalent to
\begin{align}\label{eq:quad-min}
\min_{\widehat p \in \widehat \Gamma}
\left\| \widehat p - \bigg ( I + \frac{1}{\gamma} H \bigg )^{-\frac{1}{2}} \bigg [ \frac{\theta}{\gamma} + \bigg(I - \frac{H}{1-\gamma}\bigg) \frac{Z}{Y} \bigg ]\right\|^2 ,
\end{align}
where $\widehat p : = ( I + \frac{1}{\gamma} H )^{\frac{1}{2}} p$. Solving \eqref{eq:quad-min}
gives the maximiser
\begin{align}\label{eq:p*-foc-solution}
p^* (t) &= \bigg ( I + \frac{1}{\gamma} H \bigg )^{-\frac{1}{2}} {\widehat p}^* (t) \nonumber \\
&= \bigg ( I + \frac{1}{\gamma} H \bigg )^{-\frac{1}{2}} \Proj_{\widehat \Gamma} \bigg \{ \bigg ( I + \frac{1}{\gamma} H \bigg )^{-\frac{1}{2}} \bigg [ \frac{\theta}{\gamma} + \bigg(I - \frac{H}{1-\gamma}\bigg) \frac{Z}{Y} \bigg ] \bigg \} ,
\end{align}
where ${\widehat p}^* (t)$ is the unique solution to the quadratic minimisation problem \eqref{eq:quad-min}.

Plugging $p^* (t)$ into \eqref{eq:phi*-foc-inter-solution} leads to
\begin{align}\label{eq:phi*-foc-solution}
\phi^* (t)
%&= - \bigg ( H^{-1} + \frac{1}{\gamma} I \bigg )^{-1}
%\bigg [ \frac{1}{\gamma} \theta (t) + \frac{1}{1-\gamma} \frac{Z (t)}{Y (t)} \bigg ] \\
=& - \frac{H}{1-\gamma}
\bigg [ \gamma \frac{Z}{Y} + (1-\gamma) p^* (t) \bigg ] \nonumber \\
=& -\frac{\gamma H}{1-\gamma} \frac{Z}{Y} - H \bigg ( I + \frac{1}{\gamma} H \bigg )^{-\frac{1}{2}} \Proj_{\widehat \Gamma} \bigg \{ \bigg ( I + \frac{1}{\gamma} H \bigg )^{-\frac{1}{2}} \bigg [ \frac{\theta}{\gamma} + \bigg(I - \frac{H}{1-\gamma}\bigg) \frac{Z}{Y} \bigg ] \bigg \} .
\end{align}
%Denote by
%\begin{eqnarray*}
%\Gamma_{q1} : = \big [ (1 - q) I + q \Gamma^{-1} \big ]^{-1} ,
%\quad \Gamma_{q2} : = \bigg [ \frac{1 - q}{q} \Gamma + I \bigg ]^{-1} ,
%\end{eqnarray*}
%and
%\begin{eqnarray*}
%{\widehat \Gamma}_q : = \frac{1}{2} q (1 - q) \Gamma_{q1}^2 + \Gamma_{q2}^2 .
%\end{eqnarray*}
Substituting \eqref{eq:c*-foc-solution}, \eqref{eq:p*-foc-solution}, and \eqref{eq:phi*-foc-solution} into the stochastic HJBI equation \eqref{eq:HJBI}, we have
\begin{align}
 \frac{\gamma}{1-\gamma} x^{1-\gamma} Y^{\gamma - 1} \bigg \{ &- f + \frac{Y}{Y \wedge \frac{1}{\underline c} \vee \frac{1}{\overline c}}  + \frac{1}{\gamma} \bigg [ - \rho + (1-\gamma) r
+ \frac{1 - \gamma}{2 \gamma} \theta^\top \bigg ( I + \frac{1}{\gamma} H \bigg )^{-1} \theta \bigg ] Y \nonumber \\
&  + \frac{1}{\gamma} \theta^\top \bigg [ \bigg ( I + \frac{1}{\gamma} H \bigg )^{-1} - \gamma I \bigg ] Z
- \frac{1}{2\gamma(1-\gamma)} \frac{Z^\top (I + \frac{1}{\gamma}H)^{-1} H Z}{Y}  \nonumber \\
&  - \frac{1-\gamma}{2} \Dist^2_{\widehat \Gamma} \bigg \{ \bigg ( I + \frac{1}{\gamma} H \bigg )^{-\frac{1}{2}} \bigg [ \frac{\theta}{\gamma} + \bigg(I - \frac{H}{1-\gamma}\bigg) \frac{Z}{Y} \bigg ] \bigg \} Y \bigg \} = 0 .
\end{align}
Therefore, the driver of the BSDE \eqref{eq:BSDE} is given by \eqref{eq:driver}. The proof is now completed.
\end{proof}

\begin{remark}\label{rmk:implication}
	As in Merton's classic work \cite{Merton1971}, there are two components in the optimal investment strategy, namely myopic demand and hedging demand. Indeed, when the investment constraint is absent, $( I + \frac{1}{\gamma} H )^{-1} \frac{\theta (t)}{\gamma}$ and $( I + \frac{1}{\gamma} H )^{-1} (I - \frac{H}{1-\gamma}) \frac{Z (t)}{Y (t)} $ correspond to the myopic demand and the hedging demand, respectively. In such a special case, the former component can be considered as an adjusted Merton ratio, which reduces to the Merton ratio $\frac{\theta (t)}{\gamma}$ if $H$ is a zero matrix (when there is no ambiguity aversion or when the investor is ambiguity-neutral). In other words, the adjustment arises due to the investor's ambiguity aversion to model uncertainty.
\end{remark}

\begin{remark}\label{rmk:no-ambiguity}
When the investor is ambiguity-neutral, i.e. $H = 0_{n \times n}$, the optimal investment and consumption strategies and the optimal distortion process are
\begin{align}
p^*_0 (t) =  \Proj_{\widehat \Gamma} \bigg \{ \frac{1}{\gamma} \theta (t)
+ \frac{Z_0 (t)}{Y_0 (t)} \bigg \} , \quad c^*_0 (t) = \frac{1}{Y_0 (t)} \vee \underline c \wedge \overline c , \quad \phi^*_0 (t) = 0_n ,
\end{align}
and BSDE \eqref{eq:BSDE} becomes
\begin{align}\label{eq:BSDE-0}
	d Y_0 (t) = - f_0 (t, Y_0 (t), Z_0 (t)) d t + Z_0 (t)^\top d W (t) , \quad Y_0 (T) = \beta^\frac{1}{\gamma} ,
\end{align}
with the following driver:
\begin{align}
f_0 (t, Y_0, Z_0) = \frac{Y_0}{Y_0 \wedge \frac{1}{\underline c} \vee \frac{1}{\overline c}} + \frac{1}{\gamma} \bigg [ - \rho + (1-\gamma) r
+ \frac{1 - \gamma}{2 \gamma} |\theta|^2 \bigg ] Y_0  + \frac{1 - \gamma}{\gamma}  \theta^\top Z_0 - \frac{1-\gamma}{2} \Dist^2_{\widehat \Gamma} \bigg \{ \frac{\theta}{\gamma} + \frac{Z_0}{Y_0} \bigg \} Y_0 .
\end{align}
Here, $(Y_0, Z_0)$ is the unique solution to the above BSDE \eqref{eq:BSDE-0}.
\end{remark}

\begin{remark}
	As can be seen from the driver of BSDE \eqref{eq:BSDE}, the investment (resp. consumption) constraint affects not only the optimal investment (resp. consumption) strategy itself, but also the optimal consumption (resp. investment) strategy. This subtle interaction arises from the solution pair $(Y, Z)$ of BSDE \eqref{eq:BSDE}, which depends on both the investment and consumption constraints $\widehat \Gamma$ and $\mathscr C$.
\end{remark}

\subsection{Verification Theorem}

In this subsection, we provide a verification theorem 
to confirm that the investment-consumption strategy \eqref{eq:p-c-star} and the distortion process \eqref{eq:phi-star} obtained in Proposition \ref{prop:HJBI-solution} are indeed the optimal ones. For later use, we formally define bounded mean oscillation (BMO) martingales. One can refer to \citet{Kazamaki2006} for a systematic introduction of BMO martingales.

\begin{definition}
	We say that a local martingale $\{ M (t) \}_{t \in [0, T]}$ a BMO martingale under $\mathbb P$ if there exists a constant $K > 0$ such that for any $\mathbb F$-stopping time $\tau \in [0, T]$, the following inequality holds:
	\begin{align}\label{eq:BMO-ineq}
		{\mathbb E} [ (M (T) - M (\tau))^2 | {\cal F}_\tau ] \leq K .
	\end{align}
\end{definition}

The next lemma shows the existence and uniqueness of
a solution to BSDE \eqref{eq:BSDE} and the properties of the unique solution.

\begin{lemma}\label{lem:BSDE}
Under Assumption \ref{ass:bound}, BSDE \eqref{eq:BSDE} admits a unique solution $(Y, Z) \in {\cal S}^\infty_{\mathbb P} (0, T; \mathbb R) \times {\cal L}^2_{\mathbb P} (0, T; \mathbb R^n)$. Moreover,
the solution satifies the following properties: (1) $Y$ is bounded above zero, i.e. there exists $\varepsilon > 0$ such that $Y (t) \geq \varepsilon$, for any $t \in [0, T]$, $\mathbb P$-almost surely;
(2) the It\^o integral $\int^\cdot_0 Z^\top (s) d W (s)$ is a BMO martingale under $\mathbb P$.
\end{lemma}

\begin{proof}
We divide the proof into two cases: (i) $\gamma > 1$ and (ii) $0 < \gamma < 1$.

{\it Case (i) $\gamma > 1$}. Choose an arbitrary positive constant $\delta > 0$ and consider a truncated version of BSDE \eqref{eq:BSDE} as below:
\begin{align}\label{eq:trunc-BSDE}
d Y (t) = - f_\delta (t, Y (t), Z (t)) d t
+ Z (t)^\top d W (t) , \quad Y (T) = \ \beta^\frac{1}{\gamma} ,
\end{align}
where the driver is given by 
\begin{align}
f_\delta (t, Y, Z) =&\ \frac{Y}{Y\wedge \frac{1}{\underline c} \vee \frac{1}{\overline c}}
+ \frac{1}{\gamma} \bigg [ - \rho + (1-\gamma) r 
+ \frac{1 - \gamma}{2 \gamma} \theta^\top  \bigg ( I + \frac{1}{\gamma} H \bigg )^{-1} \theta \bigg ] Y + \theta^\top \bigg [ \frac{1}{\gamma} \bigg ( I + \frac{1}{\gamma} H \bigg )^{-1} - I \bigg ] Z \\
& - \frac{1}{2 \gamma(1-\gamma)} \frac{Z^\top ( I + \frac{1}{\gamma} H )^{-1} H Z}{Y \vee \delta}  - \frac{1-\gamma}{2} \Dist^2_{\widehat \Gamma} \bigg \{ \bigg ( I + \frac{1}{\gamma} H \bigg )^{-\frac{1}{2}} \bigg [ \frac{\theta}{\gamma} + \bigg(I - \frac{H}{1-\gamma}\bigg) \frac{Z}{Y \vee \delta} \bigg ] \bigg \} Y .
\end{align}
Clearly, BSDE \eqref{eq:trunc-BSDE} has a bounded terminal value. 
Moreover, by simple algebra, we can show the driver 
$f_\delta (t, y, z)$ has (at most) linear growth in $y$ and quadratic growth in $z$, that is, 
\begin{align}
|f_\delta (t, y, z)| \leq k_1 |y| + k_2 |z| 
+ k_3 |z|^2 + k_4 , \quad \forall (t, \omega, y, z) \in [0,T] \times \Omega \times \mathbb R \times \mathbb R^n ,
\end{align}
where $k_i$ are generic positive constants, for $i = 1, 2, 3, 4$. It then follows from \citet{Kobylanski2000} that \eqref{eq:trunc-BSDE} admits a unique solution $(Y_\delta, Z_\delta) \in {\cal S}^\infty_{\mathbb P} (0, T; \mathbb R) \times {\cal L}^2_{\mathbb P} (0, T; \mathbb R^n)$ and $\int^\cdot_0 Z_\delta (s)^\top d W (s)$ is a BMO martingale under $\mathbb P$. Thus, we can define a new probability
measure $\mathbb P_\delta$ equivalent to $\mathbb P$ as follows:
\begin{align}
\frac{d \mathbb P_\delta}{d \mathbb P} \bigg |_{{\cal F}_T} = {\cal E} \bigg \{ - \int^T_0 \bigg (
\theta (t)^\top \bigg [ \frac{1}{\gamma} \bigg ( I + \frac{1}{\gamma} H \bigg )^{-1} - I \bigg ]
- \frac{1}{2 \gamma(1-\gamma)} \frac{Z_\delta (t)^\top ( I + \frac{1}{\gamma} H )^{-1} H }{Y_\delta (t) \vee \delta} \bigg ) d W (t) \bigg \} .
\end{align}
Under the measure ${\mathbb P}_\delta$, the following process 
$W^{\mathbb P_\delta} (t) := W (t) - \int^t_0 Z_\delta (s) d s ,$
is an $n$-dimensional standard Brownian motion.

By changing the probability measure from $\mathbb P$ to $\mathbb P_\delta$, BSDE \eqref{eq:trunc-BSDE} is transformed to
\begin{align*}
    % \label{eq:trunc-BSDE-tranform}
    d Y_\delta (t) =& - \bigg \{ \frac{Y_\delta}{Y_\delta \wedge \frac{1}{\underline c} \vee \frac{1}{\overline c}} + \frac{1}{\gamma} \bigg [ - \rho + (1-\gamma) r
    + \frac{1 - \gamma}{2 \gamma} \theta^\top \bigg ( I + \frac{1}{\gamma} H \bigg )^{-1} \theta \bigg ] Y_\delta \\
    & - \frac{1-\gamma}{2} \Dist^2_{\widehat \Gamma} \bigg \{ \bigg ( I + \frac{1}{\gamma} H \bigg )^{-\frac{1}{2}} \bigg [ \frac{\theta}{\gamma} + \bigg(I - \frac{H}{1-\gamma}\bigg) \frac{Z_\delta}{Y_\delta \vee \delta} \bigg ] \bigg \} Y_\delta\bigg \} d t
    + Z^\top_\delta d W^{\mathbb P_\delta} , \quad Y_\delta (T) =  \beta^{\frac{1}{\gamma}} .
\end{align*}
Here we have used $(Y_\delta, Z_\delta)$ to highlight the unique solution to \eqref{eq:trunc-BSDE} associated with $\delta$. 

Let ${\mathbb E}^{\mathbb P_\delta}_t [ \cdot ]$ denote the conditional expectation under ${\mathbb P_\delta}$ given ${\cal F}_t$. By the Feynman-Kac formula, we obtain an expectation representation for $Y_\delta$:
\begin{align}\label{eq:Y-lowerbound}
Y_\delta (t) = {\mathbb E}^{\mathbb P_\delta}_t \bigg [ \beta^{\frac{1}{\gamma}} e^{\int^T_t \widehat \Theta_\delta (s) d s} \bigg ]
\geq {\mathbb E}^{\mathbb P_\delta}_t \bigg [ \beta^{\frac{1}{\gamma}} e^{\int^T_t \Theta (s) d s} \bigg ] \geq \varepsilon ,
\end{align}
where
\begin{align}
	\widehat \Theta_\delta (t) : =&\ \frac{1}{Y_\delta (t) \wedge \frac{1}{\underline c} \vee \frac{1}{\overline c}} + \frac{1}{\gamma} \bigg [ - \rho (t) + (1-\gamma) r (t)
	+ \frac{1 - \gamma}{2 \gamma} \theta (t)^\top \bigg ( I + \frac{1}{\gamma} H \bigg )^{-1} \theta (t) \bigg ] \\
	& - \frac{1-\gamma}{2} \Dist^2_{\widehat \Gamma} \bigg \{ \bigg ( I + \frac{1}{\gamma} H \bigg )^{-\frac{1}{2}} \bigg [ \frac{\theta (t)}{\gamma} + \bigg(I - \frac{H}{1-\gamma}\bigg) \frac{Z_\delta (t)}{Y_\delta (t) \vee \delta} \bigg ] \bigg \}
\end{align}
and
\begin{align}
	\Theta (t) : = \frac{1}{\gamma} \bigg [ - \rho (t) + (1-\gamma) r (t)
	+ \frac{1 - \gamma}{2 \gamma} \theta (t)^\top \bigg ( I + \frac{1}{\gamma} H \bigg )^{-1} \theta (t) \bigg ] .
\end{align}
In \eqref{eq:Y-lowerbound}, $\varepsilon > 0$ is a lower bound of $Y_\delta (t)$, depending only on $\beta$ and $\gamma$ and the bounds of various other model parameters. Obviously, $\varepsilon > 0$ is independent of our choice of $\delta$. Therefore, by setting $\delta = \varepsilon$
from the very beginning of the proof for {\it Case (i)} gives the desired results.

{\it Case (ii) $0 < \gamma < 1$}. In this case, we rewrite the driver of BSDE \eqref{eq:BSDE} as follows:
\begin{align}\label{eq:driver-rew}
	f (t, Y, Z) =& \ \frac{Y}{Y \wedge \frac{1}{\underline c} \vee \frac{1}{\overline c}} + \frac{1}{\gamma} \big [ - \rho + (1-\gamma) r \big ] Y
	- \frac{1-\gamma}{2} [p^* (t, Y, Z)]^\top \bigg (I + \frac{H}{\gamma} \bigg) p^* (t, Y, Z) Y \nonumber \\
	& + (1-\gamma) [p^* (t, Y, Z)]^\top \bigg [ \frac{\theta}{\gamma} + \bigg(I - \frac{H}{1-\gamma}\bigg) \frac{Z}{Y} \bigg ] Y - \frac{\gamma}{2 (1-\gamma)} \frac{Z^\top H Z}{Y} ,
\end{align}
where 
\begin{align}
p^* (t, Y, Z) = \bigg ( I + \frac{1}{\gamma} H \bigg )^{-\frac{1}{2}} \Proj_{\widehat \Gamma} \bigg \{ \bigg ( I + \frac{1}{\gamma} H \bigg )^{-\frac{1}{2}} \bigg [ \frac{\theta}{\gamma} + \bigg(I - \frac{H}{1-\gamma}\bigg) \frac{Z}{Y} \bigg ] \bigg \} .
\end{align}
As in the proof of {\it Case (i)}, if we truncate the driver \eqref{eq:driver-rew} by $\widetilde \delta$ as follows:
\begin{align}
	f_{\widetilde \delta} (t, Y, Z) =& \ \frac{Y}{Y \wedge \frac{1}{\underline c} \vee \frac{1}{\overline c}} + \frac{1}{\gamma} \big [ - \rho + (1-\gamma) r \big ] Y
	- \frac{1-\gamma}{2} [p^*_{\widetilde \delta} (t, Y, Z)]^\top \bigg (I + \frac{H}{\gamma} \bigg) p^*_{\widetilde \delta} (t, Y, Z) Y \\
	& + (1-\gamma) [p^*_{\widetilde \delta} (t, Y, Z)]^\top \bigg [ \frac{\theta}{\gamma} + \bigg(I - \frac{H}{1-\gamma}\bigg) \frac{Z}{Y \vee \widetilde \delta} \bigg ] Y - \frac{\gamma}{2 (1-\gamma)} \frac{Z^\top H Z}{Y \vee \widetilde \delta} ,
\end{align}
where 
\begin{align}
p^*_{\widetilde \delta} (t, Y, Z) = \bigg ( I + \frac{1}{\gamma} H \bigg )^{-\frac{1}{2}} \Proj_{\widehat \Gamma} \bigg \{ \bigg ( I + \frac{1}{\gamma} H \bigg )^{-\frac{1}{2}} \bigg [ \frac{\theta}{\gamma} + \bigg(I - \frac{H}{1-\gamma}\bigg) \frac{Z}{Y \vee \widetilde \delta} \bigg ] \bigg \} ,
\end{align}
then we obtain a truncated BSDE:
\begin{align}\label{eq:trunc-BSDE-2}
d Y (t) = - f_{\widetilde \delta} (t, Y (t), Z (t)) d t + Z (t)^\top d W (t) , \quad Y (T) = \beta^{\frac{1}{\gamma}} .
\end{align}
It can be shown as in {\it Case (i)} that BSDE \eqref{eq:trunc-BSDE-2} has a bounded terminal value and a quadratic driver. Thus, by \citet{Kobylanski2000}, we have that \eqref{eq:trunc-BSDE-2} admits a unique solution such that $(Y_{\widetilde \delta}, Z_{\widetilde \delta}) \in {\cal S}^\infty_{\mathbb P} (0, T; \mathbb R) \times {\cal L}^2_{\mathbb P} (0, T; \mathbb R^n)$ and $\int^\cdot_0 Z_{\widetilde \delta} (s) d W (s)$ is a BMO martingale under $\mathbb P$, where the solution is written as $(Y_{\widetilde \delta}, Z_{\widetilde \delta})$ to highlight its dependence on ${\widetilde \delta}$.

Then we define an equivalent probability measure $\mathbb P_{\widetilde \delta}$ as follows:
\begin{align}
	\frac{d \mathbb P_{\widetilde \delta}}{d \mathbb P} \bigg |_{{\cal F}_T} = {\cal E} \bigg \{ - \int^T_0 \bigg (
	- \frac{\gamma}{2 (1-\gamma)} \frac{Z_{\widetilde \delta} (t)^\top H }{Y_{\widetilde \delta} (t) \vee \widetilde \delta} \bigg ) d W (t) \bigg \} .
\end{align}
Under $\mathbb P_{\widetilde \delta}$, the following process 
$W^{\mathbb P_{\widetilde \delta}} (t) := W (t) - \int^t_0 Z_{\widetilde \delta} (s) d s ,$
is an $n$-dimensional standard Brownian motion,
and BSDE \eqref{eq:BSDE} is transformed to
\begin{align}
    \begin{split}
        \label{eq:trunc-BSDE-tranform-rew}
        d Y_{\widetilde \delta} (t) = & - Y_{\widetilde \delta} \bigg \{ \frac{1}{Y_{\widetilde \delta} (t) \wedge \frac{1}{\underline c} \vee \frac{1}{\overline c}} + \frac{1}{\gamma} \big [ - \rho + (1-\gamma) r \big ] 
		- \frac{1-\gamma}{2} [p^*_{\widetilde \delta} (t, Y_{\widetilde \delta} , Z_{\widetilde \delta}]^\top \bigg (I + \frac{H}{\gamma} \bigg) p^*_{\widetilde \delta} (t, Y_{\widetilde \delta} , Z_{\widetilde \delta})  \\
		& + (1-\gamma) [p^*_{\widetilde \delta} (t, Y_{\widetilde \delta}, Z_{\widetilde \delta})]^\top \bigg [ \frac{\theta}{\gamma} + \bigg(I - \frac{H}{1-\gamma}\bigg) \frac{Z_{\widetilde \delta}}{Y_{\widetilde \delta} \vee \widetilde \delta} \bigg ] \bigg \} d t
		+ Z^\top_{\widetilde \delta} d W^{\mathbb P_{\widetilde \delta}}, \quad Y_{\widetilde \delta} (T) =  \beta^{\frac{1}{\gamma}} .
    \end{split}
\end{align}
The properties of orthogonal projection operators \citep[see e.g.][]{Beck2004} imply that
\begin{align}
	[p^*_{\widetilde \delta} (t, y, z)]^\top \bigg (I + \frac{H}{\gamma} \bigg) p^*_{\widetilde \delta} (t, y, z) = \left\| \Proj_{\widehat \Gamma} \bigg \{ \bigg ( I + \frac{1}{\gamma} H \bigg )^{-\frac{1}{2}} \bigg [ \frac{\theta}{\gamma} + \bigg(I - \frac{H}{1-\gamma}\bigg) \frac{z}{y \vee \widetilde \delta} \bigg ] \bigg \} \right\|^2 \geq 0
\end{align}
and
\begin{align}
	[p^*_{\widetilde \delta} (t, y, z)]^\top \bigg [ \frac{\theta}{\gamma} + \bigg(I - \frac{H}{1-\gamma}\bigg) \frac{z}{y} \bigg ] 
 \geq 
 %\left\| \Proj_{\widehat \Gamma} \bigg \{ \bigg ( I + \frac{1}{\gamma} H \bigg )^{-\frac{1}{2}} \bigg [ \frac{\theta (t)}{\gamma} + \bigg(I - \frac{H}{1-\gamma}\bigg) \frac{Z (t)}{Y (t)} \bigg ] \bigg \} \right\|^2 = 
 [p^*_{\widetilde \delta} (t, y, z)]^\top \bigg (I + \frac{H}{\gamma} \bigg) p^*_{\widetilde \delta} (t, y, z) , \quad \forall (t, \omega, y, z) \in [0,T] \times \Omega \times \mathbb R \times \mathbb R^n .
\end{align}

Let ${\mathbb E}^{\mathbb P_{\widetilde \delta}}_t [\cdot]$ be the conditional expectation under 
$\mathbb P_{\widetilde \delta}$ given ${\cal F}_t$. Note 
the geometric form of the driver in \eqref{eq:trunc-BSDE-tranform-rew}
implies that $Y_{\widetilde \delta}$ is non-negative. 
Furthermore, following the similar derivations as in {\it Case (i)}, we have the following representation for $Y_{\widetilde \delta}$:
{\small
\begin{align}\label{eq:Y-lowerbound-rew}
	Y_{\widetilde \delta} (t) &= {\mathbb E}^{\mathbb P_{\widetilde \delta}}_t \bigg [ \beta^{\frac{1}{\gamma}} e^{\int^T_t \widetilde \Theta (\nu) d \nu} + \int^T_t e^{\int^s_t \widetilde \Theta (\nu) d \nu} Y_{\widetilde \delta} (s) \bigg \{ \frac{1}{Y_{\widetilde \delta} (s) \wedge \frac{1}{\underline c} \vee \frac{1}{\overline c}}
	- \frac{1-\gamma}{2} [p^*_{\widetilde \delta} (s, Y_{\widetilde \delta} (s), Z_{\widetilde \delta} (s))]^\top \bigg (I + \frac{H}{\gamma} \bigg) p^*_{\widetilde \delta} (s, Y_{\widetilde \delta} (s), Z_{\widetilde \delta} (s)) \\
	& \qquad\qquad\qquad\qquad\quad + (1-\gamma) [p^*_{\widetilde \delta} (s, Y_{\widetilde \delta} (s), Z_{\widetilde \delta} (s))]^\top \bigg [ \frac{\theta (s)}{\gamma} + \bigg(I - \frac{H}{1-\gamma}\bigg) \frac{Z_{\widetilde \delta} (s)}{Y_{\widetilde \delta} (s) \vee \widetilde \delta} \bigg ] \bigg \} d s \bigg ] \\
	&\geq {\mathbb E}^{\mathbb P_{\widetilde \delta}}_t \bigg [ \beta^{\frac{1}{\gamma}} e^{\int^T_t \widetilde \Theta (\nu) d \nu} + \int^T_t e^{\int^s_t \widetilde \Theta (\nu) d \nu} Y_{\widetilde \delta} (s) \bigg \{ \frac{1}{Y_{\widetilde \delta} (s) \wedge \frac{1}{\underline c} \vee \frac{1}{\overline c}}
	+ \frac{1-\gamma}{2} [p^*_{\widetilde \delta} (s, Y_{\widetilde \delta} (s), Z_{\widetilde \delta} (s))]^\top \bigg (I + \frac{H}{\gamma} \bigg) p^*_{\widetilde \delta} (s, Y_{\widetilde \delta} (s), Z_{\widetilde \delta} (s)) \bigg \} d s \bigg ] \\
	&\geq {\mathbb E}^{\mathbb P_{\widetilde \delta}}_t \bigg [ \beta^{\frac{1}{\gamma}} e^{\int^T_t \widetilde \Theta (\nu) d \nu} \bigg ] \geq \widetilde \varepsilon,
\end{align}}
where $\widetilde \Theta : = \frac{1}{\gamma} \big [ - \rho + (1-\gamma) r \big ]$ and $\widetilde \varepsilon > 0$ is the lower bound of $Y_{\widetilde \delta} (t)$ and is independent of our choice of $\widetilde \delta$. Again, as in the proof of
{\it Case (i)}, by setting $\widetilde \delta = \widetilde \varepsilon$ from the very beginning of the proof for {\it Case (ii)} completes the proof.
\end{proof}

The next proposition establishes the martingale optimality principle for 
Problem \eqref{eq:problem} which is essentially a robust control problem. 

\begin{proposition}\label{prop:MOPR}
% Let $v (\cdot,\cdot) \in {\cal C}^{1,2} (O) \cap {\cal C} ({\overline O})$ with boundary condition $v (T,x) = U_2 (x)$
    % for all $x \in {\mathbb R}^+$, where $O : = [0, T] \times {\mathbb R}^+$ and ${\overline O}$ denotes the closure of $O$.
Suppose that i) for each $t \in [0, T]$, $v (t,x)$ is twice continuously differentiable in $x$; ii) for each $x \in {\mathbb R}^+$, $v (t,x)$ is 
a continuous $\mathbb F$-adapted process;
iii) for each $x \in {\mathbb R}^+$, the boundary condition $v (T,x) = U_2 (x)$ holds true. 
Denote by $X^{u,\phi} (\cdot)$ the wealth process associated with 
the investment-consumption strategy $u (\cdot) = (p (\cdot), c (\cdot))$ 
and the distortion process $\phi (\cdot)$,
and define a cost process $\{N^{u,\phi} (t)\}_{t\in[0,T]}$ associated with
$u (\cdot)$ and $\phi (\cdot)$ as follows:
\begin{align*}
N^{u,\phi} (t) := \int^t_0 e^{- \int^s_0 \rho (\nu) d \nu} R (s, X^{u,\phi} (s), u (s), \phi (s)) d s + e^{- \int^t_0 \rho (\nu) d \nu}  v (t, X^{u,\phi} (t)) , \quad \forall t \in [0, T] ,
\end{align*}
where
\begin{align}
R (s, X^{u,\phi} (s), u (s), \phi (s)) : = \frac{1-\gamma}{2} \phi (s)^\top H^{-1} \phi (s) \cdot v (s, X^{u,\phi} (s)) + U_1 (c (s) X (s)) .
\end{align}
Suppose the following conditions hold true:
\begin{enumerate}
\item for any $u (\cdot) \in \cal A$, there exists $\mathbb Q (\phi) \in {\cal Q}_{\cal B}$, i.e. $\phi (\cdot) \in \cal B$, such that the process $\{ N^{u,\phi} (t) \}_{t \in [0, T]}$ is a $\mathbb Q (\phi)$-supermartingale;
\item there exists $u^* (\cdot) \in \cal A$ and $\mathbb Q^* := \mathbb Q (\phi^*) \in {\cal Q}_{\cal B}$, i.e. $\phi^* (\cdot) \in \cal B$, such that
\begin{align}\label{eq:opt-measure}
& \inf_{\phi \in {\cal B}} {\mathbb E}^{\mathbb Q (\phi)}
\bigg [ \int^T_0 e^{- \int^s_0 \rho (\nu) d \nu} R (s, X^{u^*, \phi} (s), u^* (s), \phi (s)) d s
+ e^{- \int^T_0 \rho (\nu) d \nu} U_2 (X^{u^*,\phi} (T)) \bigg ] \nonumber \\
& = {\mathbb E}^{\mathbb Q^*} \bigg [ \int^T_0 e^{- \int^s_0 \rho (\nu) d \nu} R (s, X^{u^*, \phi^*} (s), u^* (s), \phi^* (s)) d s + e^{- \int^T_0 \rho (\nu) d \nu}  U_2 (X^{u^*, \phi^*} (T)) \bigg ]
\end{align}
and $\{ N^{u^*, \phi^*} (t) \}_{t \in [0, T]}$ is a $\mathbb Q^*$-martingale.
\end{enumerate}
Then $u^* (\cdot) : = (p^* (\cdot), c^*(\cdot))$ is a (robust) optimal investment-consumption strategy and $\phi^* (\cdot)$ is an optimal distortion process,
and the value function is given by $V (t, x) = v (t, x)$.
\end{proposition}

\begin{proof}
By Condition 1, we have that for any $u (\cdot) \in {\cal A}$, there exists $\phi (\cdot) \in {\cal B}$ such that
\begin{align*}
& \ {\mathbb E}^{\mathbb Q (\phi)}_{t,x} \bigg [ \int^T_0 e^{- \int^s_0 \rho (\nu) d \nu}  R (s, X^{u, \phi} (s), u (s), \phi (s)) d s
+ e^{- \int^T_0 \rho (\nu) d \nu} U_2 (X^{u,\phi} (T)) \bigg ] \\
&= {\mathbb E}^{\mathbb Q (\phi)}_{t,x} \bigg [ N^{u,\phi} (T) \bigg ]
\leq N^{u,\phi} (t) = \int^t_0 e^{- \int^s_0 \rho (\nu) d \nu} R (s, X^{u,\phi} (s), u (s), \phi (s)) d s + e^{- \int^t_0 \rho (\nu) d \nu}  v (t, x) .
\end{align*}
Taking the infimum and supreme w.r.t. $\phi (\cdot) \in \cal B$ and $u (\cdot) \in {\cal A}$, respectively, 
and rearranging the two sides of the above inequalities, we obtain
\begin{align*}
V (t, x) = \sup_{u \in {\cal A}} \inf_{\phi \in \cal B} {\mathbb E}^{\mathbb Q (\phi)}_{t,x} \bigg [ \int^T_t e^{- \int^s_t \rho (\nu) d \nu}  R (s, X^{u, \phi} (s), u (s), \phi (s)) d s
+ e^{- \int^T_t \rho (\nu) d \nu} U_2 (X^{u,\phi} (T)) \bigg ] \leq v (t,x) .
\end{align*}
Moreover, it follows from Condition 2 that
\begin{align*}
& \ {\mathbb E}^{\mathbb Q^*}_{t,x} \bigg [ \int^T_0 e^{- \int^s_0 \rho (\nu) d \nu} R (s, X^{u^*, \phi^*} (s), u^* (s), \phi^* (s)) d s
+ e^{- \int^T_0 \rho (\nu) d \nu} U_2 (X^{u^*,\phi^*} (T)) \bigg ] \\
& = {\mathbb E}^{\mathbb Q^*}_{t,x} \bigg [ N^{u^*,\phi^*} (T) \bigg ]
= N^{u^*,\phi^*} (t) = \int^t_0 e^{- \int^s_0 \rho (\nu) d \nu}  R (s, X^{u^*,\phi^*} (s), u^* (s), \phi^* (s)) d s + e^{- \int^t_0 \rho (\nu) d \nu} v (t, x) .
\end{align*}
Thus,
\begin{align*}
V (t, x) &= \sup_{u \in {\cal A}} \inf_{\phi \in \cal B} {\mathbb E}^{\mathbb Q (\phi)}_{t,x} \bigg [ \int^T_t e^{- \int^s_t \rho (\nu) d \nu} R (s, X^{u, \phi} (s), u (s), \phi (s)) d s
+ e^{- \int^T_t \rho (\nu) d \nu}  U_2 (X^{u,\phi} (T)) \bigg ] \\
&\geq \inf_{\phi \in \cal B} {\mathbb E}^{\mathbb Q (\phi)}_{t,x} \bigg [ \int^T_t e^{- \int^s_t \rho (\nu) d \nu}  R (s, X^{u^*, \phi} (s), u^* (s), \phi (s)) d s
+ e^{- \int^T_t \rho (\nu) d \nu}  U_2 (X^{u^*,\phi} (T)) \bigg ] \\
& = {\mathbb E}^{\mathbb Q^*}_{t,x} \bigg [ \int^T_t e^{- \int^s_t \rho (\nu) d \nu}  R (s, X^{u^*, \phi^*} (s), u^* (s), \phi^* (s)) d s
+ e^{- \int^T_t \rho (\nu) d \nu} U_2 (X^{u^*,\phi^*} (T)) \bigg ] = v (t, x) .
\end{align*}
This completes the proof.
\end{proof}

Before we verify that the solution to the stochastic HJBI equation \eqref{eq:HJBI} gives the robust optimal strategies and the value function to Problem \eqref{eq:problem}, we present three useful lemmas which validate that the strategies obtained in Proposition \ref{prop:HJBI-solution} are admissible.

\begin{lemma}\label{lem:RN}
The Radon-Nikodym derivative process $\{\Lambda (t)\}_{t\in[0,T]}$ associated with $\phi^*$, i.e.
\begin{align}
\Lambda (t)
= \exp \bigg \{ - \frac{1}{2} \int^{t}_0 |\phi^*(s)|^2 d s + \int^{T}_0 \phi^* (s)^\top d W (s) \bigg \} ,
\end{align}
is a uniformly integrable $\mathbb P$-martingale.
\end{lemma}

\begin{proof}
Recall the relation \eqref{eq:phi*-foc-inter-solution} and the property of the solution $(Y (\cdot), Z (\cdot))$. It suffices to show 
that the stochastic integral $\int^\cdot_0 p^* (t)^\top d W (t)$ is a BMO martingale under $\mathbb P$, which ensures that $\int^\cdot_0 \phi^* (t)^\top d W (t)$ is also a BMO martingale under $\mathbb P$ and thus the desired result of the current lemma holds true.

To proceed, we need to derive the following inequality as used to define BMO martingales (see Eq. \eqref{eq:BMO-ineq}):
	\begin{align}
	&\ {\mathbb E} \bigg[ \bigg(\int^T_0 p^* (t)^\top d W (t) - \int^\tau_0 p^* (t)^\top d W (t) \bigg )^2 \bigg | {\cal F}_\tau \bigg ] \\
    %     & = {\mathbb E} \bigg[ \int^T_\tau |p^* (t)|^2 d t \bigg | {\cal F}_\tau \bigg ] = 
         & = {\mathbb E} \bigg[ \int^T_\tau \bigg | \bigg ( I + \frac{1}{\gamma} H \bigg )^{-\frac{1}{2}} \Proj_{\widehat \Gamma} \bigg \{ \bigg ( I + \frac{1}{\gamma} H \bigg )^{-\frac{1}{2}} \bigg [ \frac{\theta}{\gamma} + \bigg(I - \frac{H}{1-\gamma}\bigg) \frac{Z}{Y} \bigg ] \bigg \} \bigg |^2 d t \bigg | {\cal F}_\tau \bigg ] \\
         & \leq 
         {\mathbb E} \bigg[ \int^T_\tau \bigg | \bigg ( I + \frac{1}{\gamma} H \bigg )^{-1} \bigg [ \frac{\theta}{\gamma} + \bigg(I - \frac{H}{1-\gamma}\bigg) \frac{Z}{Y} \bigg ] \bigg |^2 d t \bigg | {\cal F}_\tau \bigg ] \\
         & =  
         {\mathbb E} \bigg[ \bigg ( \int^T_\tau \bigg \{ \bigg ( I + \frac{1}{\gamma} H \bigg )^{-1} \bigg [ \frac{\theta}{\gamma} + \bigg(I - \frac{H}{1-\gamma}\bigg) \frac{Z}{Y} \bigg ] \bigg \}^\top d W (t) \bigg )^2 \bigg | {\cal F}_\tau \bigg ] \\
         & \leq K .
	\end{align}
 In the above derivation, the second and fourth lines are due to the It\^o isometry, 
 the third line is obtained by using the nonexpansiveness property of projection operators \citep[see][Theorem 9.9]{Beck2004},
 and the last line holds since Lemma \ref{lem:BSDE} implies that $\int^\cdot_0 \{ ( I + \frac{1}{\gamma} H )^{-1} [ \frac{\theta}{\gamma} + (I - \frac{H}{1-\gamma}) \frac{Z}{Y} ] \}^\top d W (t)$ is a BMO martingale under $\mathbb P$.

 The desired result is an immediate consequence of Theorem 2.3 in \citet{Kazamaki2006}. The proof is completed.
\end{proof}

\begin{lemma}\label{lem:SDE}
Under Assumption \ref{ass:bound}, SDE \eqref{eq:wealth-equation} associated with $u^* = (p^*,c^*)$ and $\phi^*$ admits a unique non-negative solution $X^* \in {\cal S}^2_{\mathbb Q^*} (0, T; \mathbb R)$, where
$\mathbb Q^* := \mathbb Q (\phi^*)$. Moreover,
it holds true that ${\mathbb E} [ \sup_{t \in [0, T]} |X^* (t)|^q ] < \infty$, for any $q \leq 2$.
\end{lemma}

\begin{proof}
Associated with the strategies $u^*=(p^*, c^*)$ and $\phi^*$, the wealth process $X^*$ satisfies
\begin{align}\label{eq:wealth-Q}
d X^* (t) = X^* (t) \left [ A (t) dt + p^* (t)^\top d W^{\mathbb Q^*} (t) \right ] ,
\end{align}
where
\begin{align}
A (t) : =& \ r (t) - \frac{1}{Y (t)} \vee \underline c \wedge \overline c + \Proj_{\widehat \Gamma} \bigg \{ \bigg ( I + \frac{1}{\gamma} H \bigg )^{-\frac{1}{2}} \bigg [ \frac{\theta (t)}{\gamma} + \bigg(I - \frac{H}{1-\gamma}\bigg) \frac{Z (t)}{Y (t)} \bigg ] \bigg \}^\top  \bigg ( I + \frac{1}{\gamma} H \bigg )^{-\frac{1}{2}} D (t) ,
\end{align}
and 
\begin{align}
D (t) : = \theta (t) -\frac{\gamma H}{1-\gamma} \frac{Z (t)}{Y (t)} - H \bigg ( I + \frac{1}{\gamma} H \bigg )^{-\frac{1}{2}} \Proj_{\widehat \Gamma} \bigg \{ \bigg ( I + \frac{1}{\gamma} H \bigg )^{-\frac{1}{2}} \bigg [ \frac{\theta (t)}{\gamma} + \bigg(I - \frac{H}{1-\gamma}\bigg) \frac{Z (t)}{Y (t)} \bigg ] \bigg \} .
\end{align}
Although \eqref{eq:wealth-Q} is a linear SDE, it does not satisfy the standard regularity conditions (e.g. Lipschitz condition) because both $A (t)$ and $p^* (t)$ are unbounded stochastic processes. Thus, the standard SDE theory is not applicable here. To overcome this difficulty, we use measure change techniques and BMO martingale theory
to show that \eqref{eq:wealth-Q} has a unique solution that is square-integrable under $\mathbb Q^*$.

First, it can be shown as in the proof of Lemma \ref{lem:RN} that 
$\int^\cdot_0 D (t)^\top d W (t)$ is a BMO martingale under $\mathbb P$. Then applying Theorem 3.3 in \citet{Kazamaki2006} and the fact that $\mathbb Q^*$ is a uniformly integrable measure equivalent 
to $\mathbb P$ as follows: 
\begin{align}
\frac{d {\mathbb Q}^*}{d \mathbb P} \bigg |_{{\cal F}_T} = {\cal E} \bigg \{ - \int^T_0 \phi^* (t)^\top d W (t) \bigg \} ,
\end{align}
we obtain $\int^\cdot_0 D (t)^\top d W^{\mathbb Q^*} (t)$ is a BMO martingale under $\mathbb Q^*$ as well.

With the above result, we can now define a new probability measure $\widehat {\mathbb Q}^*$ equivalent to $\mathbb Q^*$ as below
\begin{align}
\frac{d \widehat {\mathbb Q}^*}{d \mathbb Q^*} \bigg |_{{\cal F}_T} = \Pi (T) := {\cal E} \bigg \{ - \int^T_0 D (t)^\top d W^{\mathbb Q^*} (t) \bigg \} .
\end{align}
By Girsanov's theorem, we know that the process $\{W^{\widehat {\mathbb Q}^*} (t)\}_{t \in [0, T]}$ defined by
$W^{\widehat {\mathbb Q}^*} (t) : = W^{\mathbb Q^*} (t) + \int^t_0 D (s) d s $
is a Brownian motion under $\widehat {\mathbb Q}^*$. Moreover, we note that
\begin{align}
\frac{1}{\Pi (T)}
= {\cal E} \bigg \{ \int^T_0 D (t)^\top d W^{\widehat {\mathbb Q}^*} (t) \bigg \} .
\end{align}
Under $\widehat {\mathbb Q}^*$, we consider a Dol\'eans-Dade exponential process $\{{\widehat X}^* (t)\}_{t\in[0,T]}$ defined by
\begin{align}
{\widehat X}^* (t) : = {\cal E} \bigg \{ \int^t_0 p^* (s)^\top d W^{\widehat {\mathbb Q}^*} (s) \bigg \} ,
\end{align}
which is the unique solution to the following SDE:
\begin{align}\label{eq:SDE-hat-X}
d {\widehat X}^* (t) = {\widehat X}^* (t) p^* (t) d W^{\widehat {\mathbb Q}^*} (t)
= {\widehat X}^* (t) \big [ {\bar A} (t) d t + p^* (t)^\top d W^{\mathbb Q} (t) \big ] ,
\end{align}
where ${\bar A} (t) : = A (t) - (r (t) - \frac{1}{Y(t)} \vee \underline c \wedge \overline c)$.

Thus, we have $X^* (t) := {\widehat X}^* (t) \cdot e^{\int^t_0 ( r (s) - \frac{1}{Y (s)} \vee \underline c \wedge \overline c ) d s}$
is a solution to the following SDE:
\begin{align}
d X^* (t) &= d [ {\widehat X}^* (t) \cdot e^{\int^t_0 ( r (s) - \frac{1}{Y (s)} \vee \underline c \wedge \overline c ) d s} ] \nonumber \\
&= [ {\widehat X}^* (t) \cdot e^{\int^t_0 ( r (s) - \frac{1}{Y (s)} \vee \underline c \wedge \overline c ) d s} ]
\big [ A (t) d t + p^* (t)^\top d W^{\mathbb Q^*} (t) \big ] = X^* (t) \big [ A (t) d t + p^* (t)^\top d W^{\mathbb Q^*} (t) \big ] ,
\end{align}
which coincides with \eqref{eq:wealth-Q}. This confirms the existence of a solution to \eqref{eq:wealth-Q}.

Furthermore, the uniqueness of the solution $X^* (t)$ can be verified by contradiction. Indeed, if there is another solution denoted by
$X^\prime (t)$ to SDE \eqref{eq:wealth-Q}, then ${\widehat X}^\prime (t) : = X^\prime (t) e^{- \int^t_0 ( r (s) - \frac{1}{Y (s)} \vee \underline c \wedge \overline c ) d s}$
is also a solution to \eqref{eq:SDE-hat-X}. This contradicts with the uniqueness of the solution to \eqref{eq:SDE-hat-X}. Therefore,
$X^* (t) = {\widehat X}^* (t) \cdot e^{\int^t_0 ( r (s) - \frac{1}{Y (s)} \vee \underline c \wedge \overline c ) d s}$ must be the unique solution to SDE \eqref{eq:wealth-Q}.

Next we show that $X^*$ is square integrable under $\mathbb Q^*$. For that purpose, we derive
\begin{align}
{\mathbb E}^{\mathbb Q^*} \bigg [ \sup_{t \in [0, T]} |X^* (t)|^2 \bigg ]
&= {\mathbb E}^{\mathbb Q^*} \bigg [ \sup_{t \in [0, T]} |{\widehat X}^* (t) \cdot e^{\int^t_0 ( r (s) - \frac{1}{Y (s)} \vee \underline c \wedge \overline c ) d s}|^2 \bigg ] \nonumber \\
&\leq C {\mathbb E}^{\mathbb Q^*} \bigg [ \sup_{t \in [0, T]} |{\widehat X}^* (t)|^2 \bigg ]
= C {\mathbb E}^{\widehat{\mathbb Q}^*} \bigg [ \frac{1}{\Pi (T)} \sup_{t \in [0, T]} |{\widehat X}^* (t)|^2 \bigg ] \nonumber \\
&\leq C \bigg \{ {\mathbb E}^{\widehat{\mathbb Q}^*} \bigg [ \frac{1}{\Pi^2 (T)} \bigg ] \bigg \}^\frac{1}{2}
\bigg \{ {\mathbb E}^{\widehat{\mathbb Q}^*} \bigg [ \sup_{t \in [0, T]} |{\widehat X}^* (t)|^4 \bigg ] \bigg \}^\frac{1}{2} \nonumber \\
&\leq C \bigg \{ {\mathbb E}^{\widehat{\mathbb Q}^*} \bigg [ \frac{1}{\Pi^2 (T)} \bigg ] \bigg \}^\frac{1}{2}
\bigg \{ {\mathbb E}^{\widehat{\mathbb Q}^*} \bigg [ |{\widehat X}^* (T)|^4 \bigg ] \bigg \}^\frac{1}{2} < \infty .
\end{align}
Here, the third line is due to the Cauchy-Schwarz inequality, the first inequality in the last line is ensured by Doob's martingale inequality, and the second one in the last line is derived by applying the reverse H\"older inequality (see \citet{Kazamaki2006}) to $\frac{1}{\Pi (T)}$ and ${\widehat X}^* (T)$. The non-negativity of $X^*$ follows immediately from the constructed form $X^* (t) = {\cal E} \{ \int^t_0 p^* (s) d W^{\widehat {\mathbb Q}} (s) \} \cdot e^{\int^t_0 ( r (s) - \frac{1}{Y (s)} \vee \underline c \wedge \overline c ) d s}$.

Moreover, by It\^o's formula, we have
\begin{align}\label{eq:wealth-1-gamma}
	d |X^* (t)|^{q} = |X^* (t)|^{q} \left \{ \left [ q A (t) + \frac{1}{2} q (q-1) |p^* (t)|^2 \right ] dt +  q \, p^* (t)^\top d W^{\mathbb Q^*} (t) \right \} .
\end{align}
Let ${\widetilde {\mathbb Q}^*}$ be another new probability measure equivalent to $\mathbb Q^*$ defined as follows:
\begin{align}
	\frac{d {\widetilde {\mathbb Q}^*}}{d \mathbb Q^*} = {\cal E} \bigg \{ - \int^T_0  \bigg \{ D (t)
	+ \frac{1}{2} (q-1) p^* (t) \bigg \}^\top d W^{\mathbb Q^*} (t) \bigg \} .
\end{align}
Thus, we obtain the relationship $|X^* (t)|^{q} = {\cal E} \{ \int^t_0 q \cdot p^* (s)^\top d W^{\widetilde {\mathbb Q}^*} (s) \} \cdot e^{\int^t_0 q ( r (s) - \frac{1}{Y (s)} \vee \underline c \wedge \overline c ) d s}$. Following similar derivations as presented above, we can show ${\mathbb E} [ \sup_{t \in [0, T]} |X^* (t)|^q ] < \infty$, for any $q \leq 2$. The proof is now completed.
\end{proof}

\begin{lemma}
Under Assumption \ref{ass:bound}, the strategies $u^* = (p^*, c^*)$ and $\phi^*$ obtained in Proposition \ref{prop:HJBI-solution} are admissible, that is, $u^* = (p^*, c^*) \in {\cal A}$ and $\phi^* \in {\cal B}$.
\end{lemma}

\begin{proof}
In Lemmas \ref{lem:RN} and \ref{lem:SDE}, we have verified Conditions 3 and 2 of Definition \ref{def:adm}, respectively.
We only need to show Conditions 1 and 4 of Definition \ref{def:adm}.
First, it follows from the explicit expressions of $u^* = (p^*, c^*)$ and 
$\phi^*$ that the former is constrained in $\widehat \Gamma \times \cal C$ and the latter takes value in $\mathbb R^n$. Hence, Condition 1 is met. 

It remains to show Condition 4. In Lemma \ref{lem:RN}, we have proved 
that $\int^\cdot_0 (\phi^* (s))^\top d W (s)$ is a BMO martingale under $\mathbb P$. By Theorem 3.3 in \citet{Kazamaki2006}, we can show that the same is true under $\mathbb Q$.
Moreover, the BMO property implies that the optimal distortion process $\phi^*$ also satisfies the square-integrability under $\mathbb Q$, that is, $\phi^* \in {\cal L}^2_{\mathbb Q} (0, T; \mathbb R^n)$.

Finally, we check other conditions related to $u^* = (p^*, c^*)$. For that purpose, we note that
1) due to Lemma \ref{lem:SDE}, the wealth process
associated with  $(p^*, c^*)$  satisfies $X^* \in {\cal S}^2_{\mathbb Q} (0, T; \mathbb R)$,
which, together with the boundedness of the optimal consumption strategy $c^*$, ensures that $X^* c^* \in {\cal L}^1_{\mathbb Q} (0, T; \mathbb R)$;
and 2) it has been shown as in the proof of Lemma \ref{lem:RN} that $\int^\cdot_0 (p^* (s))^\top d W (s)$ is a BMO martingale under $\mathbb P$. Then applying Lemma 4.1 in \citet{DelbaenTang2010}
leads to $X^* p^* \in {\cal L}^2_{\mathbb Q} (0, T; \mathbb R^n)$.
Combining the above facts gives Condition 4. Therefore, we obtain that $u^*=(p^*, c^*) \in {\cal A}$ and $\phi^* \in {\cal B}$. This completes the proof.
\end{proof}

In what follows, we present the main result of this subsection, namely the verification theorem. This theorem is proved by applying the martingale optimality principle established in Proposition \ref{prop:MOPR}.

\begin{theorem}\label{thm:main}
	Under Assumption \ref{ass:bound}, the robust optimal investment-consumption strategy and the optimal distortion process are given by \eqref{eq:p-c-star} and \eqref{eq:phi-star}
	and the value function is given by \eqref{eq:G}.
\end{theorem}

\begin{proof}
For simplicity, we define
\begin{align}
v (t, x) := G (t, x, Y (t)) = \frac{x^{1-\gamma}}{1-\gamma} \cdot [Y (t)]^{\gamma} ,
\end{align}
and recall the following two processes defined in Proposition \ref{prop:MOPR}:
\begin{align}
R (s, x, u, \phi) & : = \frac{1-\gamma}{2} \phi^\top H^{-1} \phi \cdot G (s, x, Y (s)) + U_1 (c x) , \\
N^{u,\phi} (t) & := \int^t_0 e^{- \int^s_0 \rho (\nu) d \nu} R (s, X^{u,\phi} (s), u (s), \phi (s)) d s + e^{- \int^t_0 \rho (\nu) d \nu}  v (t, X^{u,\phi} (t)) .
\end{align}

%Let
%\begin{align*}
%\Pi (t, x, Y; u, \phi) &:= {\cal L}^{u, \phi} [G (t, x, Y)]
%+ \frac{1-\gamma}{2} \phi^\top H^{-1} \phi G (t, x, Y) + \frac{(cx)^{1-\gamma}}{1-\gamma} \\
%&= \bigg \{ - \rho G (t, x, Y (t)) + G_t (t, x, Y (t))
%+ x \left [ r (t) + p^\top (\theta (t) + \phi) - c \right ] G_x (t, x, Y (t)) \\
%& \quad - G_y (t, x, Y (t)) [f (t, Y (t), Z (t)) - Z^\top (t) \phi  ]
%+ \frac{1}{2} x^2 | p |^2 G_{xx} (t, x, Y (t)) \\
%& \quad + x p^\top Z (t) G_{xy} (t, x, Y (t))
%+ \frac{1}{2} |Z (t)|^2 G_{yy} (t, x, Y (t)) \bigg \} + \frac{1-\gamma}{2} \phi^\top H^{-1} \phi G (t, x, Y) + \frac{(cx)^{1-\gamma}}{1-\gamma}
%\end{align*}

First, we verify Condition 1 in Proposition \ref{prop:MOPR}. In what follows, we suppress the dependence of the cost and wealth processes on $u$ and $\phi$ and write $N (t) = N^{u,\phi} (t)$ and $X (t) = X^{u,\phi} (t)$ whenever there is no risk of confusion. Applying It\^o's formula to $N (t)$ gives
\begin{align*}
d N (t) = e^{- \int^t_0 \rho (s) d s} \big [ \Xi (t, X (t), Y (t); u (t), \phi (t)) d t
+ \widetilde \Xi (t, X (t), Y (t), Z (t); u (t), \phi (t)) d W^{\mathbb Q} (t) \big ] ,
\end{align*}
where
{\small
\begin{align*}
\Xi (t, X (t), Y (t), Z (t); u (t), \phi (t)) :=& - \rho (t) G (t, X (t), Y (t))
+ \left [ r (t) + p^\top (t) (\theta (t) + \phi (t)) - c (t) \right ] X (t) G_x (t, X (t), Y (t)) \\
& - G_y (t, X (t), Y (t)) [f (t, Y (t), Z (t)) - Z^\top (t) \phi (t) ]
+ \frac{1}{2} X^2 (t) | p (t) |^2 G_{xx} (t, X (t), Y (t)) \\
& + X (t) p^\top (t) Z (t) G_{xy} (t, X (t), Y (t))
+ \frac{1}{2} |Z (t)|^2 G_{yy} (t, X (t), Y (t)) \\
& + \frac{1-\gamma}{2} \phi^\top (t) H^{-1} \phi (t) G (t, X (t), Y (t)) + U_1 (c (t)X(t))
\end{align*}
}
and
\begin{align*}
\widetilde \Xi (t, X (t), Y (t), Z (t); u (t), \phi (t)) : =& \ G_x (t, X (t), Y (t)) X(t)p^\top (t) + G_y (t, X (t), Y (t)) Z^\top (t) \\
=& \ [X (t)]^{1-\gamma} [Y (t)]^{\gamma} p^\top (t) + \frac{\gamma}{1-\gamma} [X (t)]^{1-\gamma} [Y (t)]^{\gamma-1} Z^\top (t) .
\end{align*}
For any $u$, we set
\begin{align}
\phi (t) = \Phi (t, u) = - H \bigg [ \frac{\gamma}{1-\gamma} \frac{Z (t)}{Y (t)} + p \bigg ] .
\end{align}
Then we have
{\small
\begin{align*}
\Xi (t, X (t), Y (t), Z (t); u, \Phi (t, u))
& = - \rho (t) G
+ \left [ r (t) + p^\top \theta (t) - c \right ] X (t) G_x - G_y f (t, Y (t), Z (t)) + \frac{1}{2} X^2 (t) | p |^2 G_{xx} \\
& \quad + X (t) p^\top Z (t) G_{xy}
+ \frac{1}{2} |Z (t)|^2 G_{yy} - \frac{1-\gamma}{2} \Phi^\top (t, u) H^{-1} \Phi (t, u) G + U_1 (c X (t)) \\
& = - \rho (t) G + r (t) X (t) G_x - G_y f + \frac{1}{2} |Z (t)|^2 G_{yy} + \big [ - c X (t) G_x + U_1 (c X (t)) \big ] \\
& \quad + p^\top \theta (t) X (t) G_x + \frac{1}{2} X^2 (t) | p |^2 G_{xx} + X (t) p^\top Z (t) G_{xy}
- \frac{1-\gamma}{2} \Phi^\top (t, u) H^{-1} \Phi (t, u) G .
\end{align*}
}
It follows from the proof of Proposition \ref{prop:HJBI-solution} that
\begin{align*}
u^*(t) = (p^*(t), c^*(t)) = \argmax_{u \in \widehat \Gamma \times \mathscr C} \Xi (t, X (t), Y (t), Z (t); u, \Phi (t, u)) ,
\end{align*}
where $p^*(t)$ and $c^*(t)$ are given by Eq. \eqref{eq:p-c-star}, and
\begin{align}\label{eq:opt-strategy-cond}
\Xi (t, X (t), Y (t), Z (t); u (t), \Phi (t, u (t))) \leq \Xi (t, X (t), Y (t), Z (t); u^* (t), \Phi (t, u^* (t))) = 0 , \quad
\forall u (\cdot) \in {\cal A} .
\end{align}

To confirm that for any $u (\cdot) \in {\cal A}$, there always exists $\phi (t) = \Phi (t, u (t))$ such that the cost process $N$ is a supermartingale under ${\mathbb Q} (\phi)$, we now apply localization techniques. To that end, we need to show
that $\int^\cdot_0 \widetilde \Xi (t, X (t), Y (t), Z (t); u (t), \phi (t)) d W^{\mathbb Q} (t)$ is a local martingale under ${\mathbb Q} (\phi)$. Since $Y \in {\cal S}^\infty_{\mathbb Q (\phi)} (0, T; \mathbb R)$, we only need to show that $\int^\cdot_0 |X (t)|^{2(1-\gamma)} |p (t)|^2 d t < \infty$ and $\int^\cdot_0 [X (t)]^{1-\gamma} Z^\top (t) d t < \infty$ hold ${\mathbb Q} (\phi)$-almost surely. We only prove the first inequality as the second one can be proved similarly. For that purpose, using Markov's inequality, we derive
\begin{align}
	{\mathbb Q} \bigg ( \sup_{t \in [0,T]} |X (t)|^{-2\gamma} \geq a \bigg ) \leq \frac{{\mathbb E} \left[\sup_{t \in [0,T]} |X (t)|^{-2\gamma}\right]}{a} .
\end{align}
Sending $a$ to $\infty$ and combining with ${\mathbb E} [ \sup_{t \in [0, T]} |X^* (t)|^{-2\gamma} ] < \infty$ gives
\begin{align}
{\mathbb Q} \bigg ( \sup_{t \in [0,T]} |X (t)|^{-2\gamma} \geq \infty \bigg ) = 0 \quad \mbox{and} \quad
{\mathbb Q} \bigg ( \sup_{t \in [0,T]} |X (t)|^{-2\gamma} < \infty \bigg ) = 1 .
\end{align}
By using $p X \in {\cal L}^2_{\mathbb Q} (0, T; \mathbb R^n)$, we obtain
\begin{align}
\int^T_0 |X (t)|^{2(1-\gamma)} |p (t)|^2 d t \leq \bigg(\sup_{t \in [0, T]} |X (t)|^{-2\gamma}\bigg) \cdot \int^T_0 |X (t)|^{2} |p (t)|^2 d t
< \infty , \quad {\mathbb Q}\mbox{-a.s.}
\end{align}
Therefore, we can confirm that $\int^\cdot_0 \widetilde \Xi (t, X (t), Y (t), Z (t); u (t), \phi (t)) d W^{\mathbb Q} (t)$ is a local martingale under ${\mathbb Q} (\phi)$. Then, applying localization techniques leads to that $N$ is a supermartingale under ${\mathbb Q} (\phi)$.

Next we verify Condition 2 in Proposition \ref{prop:MOPR}. For that purpose, we show that the cost process $N$ is a ${\mathbb Q} (\phi^*)$-martingale. It follows from Proposition \ref{prop:HJB-solution} that for $u^* = (p^*, c^*)$ and $\phi^* = \Phi (t, u^*)$. Lemma \ref{lem:SDE} ensures ${\mathbb E}^{\mathbb Q (\phi^*)} [ \sup_{t \in [0, T]} |X^* (t)|^{2 (1-\gamma)} ] < \infty$, i.e. $|X^*|^{1-\gamma} \in {\cal S}^2_{\mathbb Q (\phi^*)} (0, T; \mathbb R)$. Moreover, since $p^*$ and $Z$ satisfy the BMO property, and $Y$ is bounded 
and bounded above zero,
we have $\int^\cdot_0 [X^* (t)]^{1-\gamma} [Y (t)]^{\gamma} p^* (t)^\top d W^{{\mathbb Q} (\phi^*)} (t)$ and $\int^\cdot_0 [X^* (t)]^{1-\gamma} [Y (t)]^{\gamma-1} Z (t)^\top d W^{{\mathbb Q} (\phi^*)} (t)$ are ${\mathbb Q} (\phi^*)$-martingales (see Lemma 1.4 in \citet{DelbaenTang2010}).
Therefore, associated with $u^*$ and $\phi^*$, the cost process $N$ is a martingale under ${\mathbb Q} (\phi^*)$.

It remains to verify \eqref{eq:opt-measure}, which can be considered as an optimal control problem with the state process $X^{u^*,\phi}$.
Hence, we only need to show that $\phi^*$ is the optimal strategy of the minimisation problem \eqref{eq:opt-measure}. It is equivalent to showing that
\begin{align}\label{eq:opt-measure-cond1}
\Xi (t, X (t), Y (t), Z (t); u^* (t), \phi (t)) \geq \Xi (t, X (t), Y (t), Z (t); u^* (t), \phi^* (t)) = 0
\end{align}
and
\begin{align}\label{eq:opt-measure-cond2}
\phi^* (t) = \argmin_{\phi \in \mathbb R^n} \Xi (t, X (t), Y (t), Z (t); u^* (t), \phi) .
\end{align}
From \eqref{eq:opt-strategy-cond}, the last equality in \eqref{eq:opt-measure-cond1} is true.
It suffices to show that $\phi^*$ is the minimiser of $\Xi (t, X (t), Y (t), Z (t); u^* (t), \phi)$.
Note the Hessian matrix of $\Xi$ w.r.t. $\phi$, i.e.
\begin{align}
\frac{\partial^2}{\partial \phi^2} \Xi (t, X (t), Y (t), Z (t); u^*(t), \phi)
= H^{-1} [X (t)]^{1-\gamma} [Y (t)]^{\gamma}
\end{align}
is positive-definite. Then applying the first-order condition to $\Xi$, that is, 
\begin{align*}
\frac{\partial}{\partial \phi} \Xi (t, X (t), Y (t), Z (t); u^*, \phi)
= X (t) [p^* (t)]^\top G_x + G_y Z^\top (t) + (1-\gamma) \phi^\top (t) H^{-1} G = 0
\end{align*}
gives the minimiser
\begin{align}
\phi (t) = - H \bigg [ \frac{\gamma}{1-\gamma} \frac{Z (t)}{Y (t)} + p^* (t) \bigg ]
\end{align}
which coincides with $\phi^* (t) = \Phi (t, u^* (t))$. The proof is now completed.
\end{proof}

\subsection{Utility Loss}
\label{sec:UL}

To end this section, we consider the situation where the investor ignores model uncertainty and adopts the sub-optimal investment-consumption strategy $(p_0, c_0) = (p_0^*, c^*_0)$ as defined in Remark \ref{rmk:no-ambiguity}. Therefore, the optimisation problem becomes
\begin{align}\label{eq:pro-inf}
	\widetilde V (t,x) = \inf_{\mathbb Q\in\mathcal{Q}}{\mathbb E}_{t,x}^{\mathbb Q}
	\Bigg[& \int_t^T e^{-\int_t^s \rho(\nu) d\nu} U_1 (c_0 (s) X^{(p_0, c_0, \phi)} (s)) ds
	+ e^{-\int_t^T \rho(\nu) d\nu} U_2 (X^{(p_0, c_0, \phi)}  (T)) \nonumber \\
	& + \frac{1-\gamma}{2} \int_t^T e^{-\int_t^s \rho(\nu) d\nu} \phi (s)^\top H^{-1} \phi (s) \widetilde V (s, X^{(p_0, c_0, \phi)} (s)) ds  \Bigg] ,
\end{align}
where $X^{(p_0, c_0, \phi)}$ denotes the wealth process associated with $(p_0, c_0)$ and $\phi$.

To establish rigorous results, we impose the following assumption in the rest of this subsection:
\begin{assumption}\label{ass:bound-constraint}
	The closed convex set $\widehat \Gamma$ is bounded.
\end{assumption}
Note that all the results obtained prior to Subsection \ref{sec:UL} do not rely on Assumption \ref{ass:bound-constraint}. In Section \ref{sec:specialdeterministic} for the deterministic case, we do not need Assumption \ref{ass:bound-constraint} either. In the remaining part of the paper, we will state explicitly if 
Assumption \ref{ass:bound-constraint} is needed.

The solution to the optimisation problem \eqref{eq:pro-inf} can be 
derived by solving the following stochastic HJB equation
\begin{align}\label{eq:HJB}
	\inf_{\phi \in {\cal B}} \bigg \{ {\cal L}^{(p_0, c_0, \phi)} [\widetilde G (t, x, \widetilde Y (t))]
	+ \frac{1-\gamma}{2} \phi^\top H^{-1} \phi G (t, x, \widetilde Y (t)) + U_1 (c_0 x) \bigg \} = 0 ,
	\quad \widetilde G (T, x, \xi) = U_2 (x) ,
\end{align}
where
\begin{align}\label{eq:inf-generator-sub}
	{\cal L}^{(p_0, c_0, \phi)} [\widetilde G (t, x, \widetilde Y (t))]
	=& - \rho (t) \widetilde G (t, x,\widetilde  Y (t)) + \widetilde G_t (t, x, \widetilde Y (t))
	+ x \left [ r (t) + p^\top_0 (t) (\theta (t) + \phi) - c_0 (t) \right ] \widetilde G_x (t, x, \widetilde Y (t)) \nonumber \\
	& - \widetilde G_y (t, x, \widetilde Y (t)) [\widetilde f (t, \widetilde Y (t), \widetilde Z (t)) - \widetilde Z^\top (t) \phi  ]
	+ \frac{1}{2} x^2 | p_0 (t) |^2 \widetilde G_{xx} (t, x, \widetilde Y (t)) \nonumber \\
	& + x p^\top_0 (t) \widetilde Z (t) \widetilde G_{xy} (t, x, \widetilde Y (t))
	+ \frac{1}{2} |\widetilde Z (t)|^2 \widetilde G_{yy} (t, x, \widetilde Y (t)) ,
\end{align}
and $(\widetilde Y, \widetilde Z)$ is the solution to the following BSDE
\begin{align}\label{eq:BSDE-wt}
	d \widetilde Y (t) = - \widetilde f (t, \widetilde Y (t), \widetilde Z (t)) d t + \widetilde Z (t)^\top d W (t) , \quad \widetilde Y (T) = \beta^\frac{1}{\gamma} ,
\end{align}
with the driver given by
{\small
\begin{align}\label{eq:driver-wt}
	\widetilde f (t, \widetilde Y, \widetilde Z) =& \ \frac{1-\gamma}{\gamma} \bigg \{ - \frac{\widetilde Y}{Y_0 \wedge \frac{1}{\underline c} \vee \frac{1}{\overline c}} +  \frac{1}{1-\gamma} \bigg [ \frac{\widetilde Y}{Y_0 \wedge \frac{1}{\underline c} \vee \frac{1}{\overline c}} \bigg ]^{1-\gamma} \bigg \}  - \frac{1-\gamma}{2} \Proj^\top_{\widehat \Gamma} \bigg \{ \frac{\theta}{\gamma} + \frac{Z_0}{Y_0} \bigg \}  \bigg ( I + \frac{H}{\gamma} \bigg ) \Proj_{\widehat \Gamma} \bigg \{ \frac{\theta}{\gamma} + \frac{Z_0}{Y_0} \bigg \} \widetilde Y \nonumber \\
	& + \frac{1}{\gamma} \bigg [ - \rho + (1-\gamma) r
	\bigg ] \widetilde Y + (1-\gamma) \Proj^\top_{\widehat \Gamma} \bigg \{ \frac{\theta}{\gamma} + \frac{Z_0}{Y_0} \bigg \}
	\bigg [ \frac{\theta}{\gamma} + \frac{\widetilde Z}{\widetilde Y} \bigg ( I -\frac{1}{1-\gamma} H \bigg ) \bigg ] \widetilde Y - \frac{1-\gamma}{2} \frac{\widetilde Z^\top ( I + \frac{\gamma}{(1-\gamma)^2} H ) \widetilde Z}{\widetilde Y} . 
\end{align}
}
Here $(Y_0, Z_0)$ denotes the unique solution to BSDE \eqref{eq:BSDE-0}.

\begin{lemma}\label{lem:BSDE-wt}
	Under Assumptions \ref{ass:bound} and \ref{ass:bound-constraint}, BSDE \eqref{eq:BSDE-wt} admits a unique solution $(\widetilde Y, \widetilde Z) \in {\cal S}^\infty_{\mathbb P} (0, T; \mathbb R) \times {\cal L}^2_{\mathbb P} (0, T; \mathbb R^n)$. Moreover,
	$\widetilde Y$ is bounded above zero, i.e. there exists $\varepsilon > 0$ such that $\widetilde Y \geq \varepsilon$;
	the It\^o integral $\int^\cdot_0 \widetilde Z (s)^\top d W (s)$ is a BMO martingale under $\mathbb P$.
\end{lemma}

\begin{proof}
	Under Assumptions \ref{ass:bound} and \ref{ass:bound-constraint}, all model coefficients and the projection operators
	in the driver \eqref{eq:driver-wt} are bounded. Then we can apply similar truncation techniques as in the proof of Lemma \ref{lem:BSDE} to derive the desired results, which is thus omitted here.
\end{proof}

\begin{remark}
	The main reason for imposing Assumption \ref{ass:bound-constraint} is that without this assumption, the boundedness and BMO properties of the solution $(\widetilde Y, \widetilde Z)$ may not hold since $Z_0$ is unbounded in general. It is worth pointing out $Z_0$ becomes zero in the deterministic case. Therefore, we do not need Assumption \ref{ass:bound-constraint} in Section \ref{sec:specialdeterministic}.
\end{remark}

Following the same procedure of solving the stochastic HJBI equation \eqref{eq:HJBI}, we derive the solution to the stochastic HJB equation \eqref{eq:HJB} in the following proposition.

\begin{proposition}\label{prop:HJB-solution}
	Under Assumptions \ref{ass:bound} and \ref{ass:bound-constraint}, the infimum in the stochastic HJB equation \eqref{eq:HJB} is achieved by
\begin{align}\label{eq:phi-star-wt}
	\widetilde \phi^* (t) = -\frac{\gamma H}{1-\gamma} \frac{\widetilde Z (t)}{\widetilde Y (t)}
	- H \Proj_{\widehat \Gamma} \bigg \{ \frac{\theta (t)}{\gamma} + \frac{Z_0 (t)}{Y_0 (t)} \bigg \} ,
\end{align}
and the solution to the stochastic HJB equation \eqref{eq:HJB} is given by
\begin{align}\label{eq:G-wt}
	\widetilde v (t, x) = \widetilde G (t, x, \widetilde Y (t)) = \frac{x^{1-\gamma}}{1-\gamma} \times [\widetilde Y (t)]^{\gamma} ,
\end{align}
where $(\widetilde Y, \widetilde Z)$ and $(Y_0, Z_0)$ are the unique solutions to BSDEs \eqref{eq:BSDE-wt} and \eqref{eq:BSDE-0}, respectively.
\end{proposition}

\begin{proof}
By the first-order condition with respect to $\phi$, we have
\begin{align*}
	\widetilde G_y (t, x, \widetilde Y) \widetilde Z^\top + x \widetilde G_x (t, x, \widetilde Y) p^\top_0 (t) + (1-\gamma) \phi^\top H^{-1} \widetilde G (t, x, \widetilde Y) = 0 .
\end{align*}
Plugging $p_0$ and $\widetilde G$ into the above equation gives
\begin{align}
	\widetilde \phi^* (t) = - H \bigg [ \frac{\gamma}{1-\gamma} \frac{\widetilde Z (t)}{\widetilde Y (t)} + p_0 (t) \bigg ]
	= -\frac{\gamma H}{1-\gamma} \frac{\widetilde Z (t)}{\widetilde Y (t)} - H \Proj_{\widehat \Gamma} \bigg \{ \frac{1}{\gamma} \theta (t)
	+ \frac{Z_0 (t)}{Y_0 (t)} \bigg \} .
\end{align}
Substituting $(p_0, c_0)$ and $\widetilde \phi^*$ into the stochastic HJB equation \eqref{eq:HJB} and combining with \eqref{eq:driver-wt}, we immediately have that the equality in \eqref{eq:HJB} holds, that is,
\begin{align}
\frac{\gamma}{1-\gamma} x^{1-\gamma} \widetilde Y^{\gamma - 1} \bigg \{& - \widetilde f + \frac{1-\gamma}{\gamma} \bigg \{ - \frac{\widetilde Y}{Y_0 \wedge \frac{1}{\underline c} \vee \frac{1}{\overline c}} +  \frac{1}{1-\gamma} \bigg [ \frac{\widetilde Y}{Y_0 \wedge \frac{1}{\underline c} \vee \frac{1}{\overline c}} \bigg ]^{1-\gamma} \bigg \} \\
 & - \frac{1-\gamma}{2} \Proj^\top_{\widehat \Gamma} \bigg \{ \frac{\theta}{\gamma} + \frac{Z_0}{Y_0} \bigg \}  \bigg ( I + \frac{H}{\gamma} \bigg ) \Proj_{\widehat \Gamma} \bigg \{ \frac{\theta}{\gamma} + \frac{Z_0}{Y_0} \bigg \} \widetilde Y + \frac{1}{\gamma} \bigg [ - \rho + (1-\gamma) r
	\bigg ] \widetilde Y \nonumber \\
	&  + (1-\gamma) \Proj^\top_{\widehat \Gamma} \bigg \{ \frac{\theta}{\gamma} + \frac{Z_0}{Y_0} \bigg \}
	\bigg [ \frac{\theta}{\gamma} + \frac{\widetilde Z}{\widetilde Y} \bigg ( I -\frac{1}{1-\gamma} H \bigg ) \bigg ] \widetilde Y - \frac{1-\gamma}{2} \frac{\widetilde Z^\top ( I + \frac{\gamma}{(1-\gamma)^2} H ) \widetilde Z}{\widetilde Y} \bigg \} = 0 .
\end{align}
The verification theorem can be established similarly as that for the stochastic HJBI equation \eqref{eq:HJBI} and thus is omitted here.
\end{proof}

We define the utility loss as the percentage of the initial wealth that the investor will lose if adopting the sub-optimal investment-consumption strategy $(p_0, c_0)$, that is, the time-dependent solution $L (t)$ to the equation $v (t, x (1 - L (t))) = \widetilde v (t, x) .$

\begin{proposition}
\label{prop:utilloss}
Under Assumptions \ref{ass:bound} and \ref{ass:bound-constraint}, if the model uncertainty is ignored, the utility loss is given by
\begin{align}\label{eq:UL}
	L (t) = 1 - \left[\frac{\widetilde Y (t)}{Y (t)} \right]^\frac{\gamma}{1-\gamma} .
\end{align}
Moreover, the utility loss always satisfies $L (t) \in [0,1]$, for any $t \in [0, T]$ and $\gamma \in (0, 1) \cup (1, \infty)$.
\end{proposition}

\begin{proof}
The utility loss \eqref{eq:UL} is an immediate result by definition and Propositions \ref{prop:HJBI-solution} and \ref{prop:HJB-solution}. Moreover, we note that
\begin{align}
	\widetilde v (t, x) &= J (t, x; p_0, c_0, \widetilde \phi^*) = \inf_{\phi \in \cal B} J (t, x; p_0, c_0, \phi)
	\leq J (t, x; p_0, c_0, \phi^*) \nonumber \\
	&\leq \sup_{(p,c) \in \cal A} J (t, x; p, c, \phi^*)
	= \sup_{(p,c) \in \cal A} \inf_{\phi \in \cal B} J (t, x; p, c, \phi) = v (t, x) .
\end{align}
Combining with Propositions \ref{prop:HJBI-solution} and \ref{prop:HJB-solution}, we have
$\frac{1}{1-\gamma} [ \widetilde Y (t) ]^\gamma \leq \frac{1}{1-\gamma}  [ Y (t) ]^\gamma ,$
which leads to
\begin{align}
\left\{
	\begin{aligned}
		\left[\frac{\widetilde Y (t)}{Y (t)}\right]^\gamma \leq 1 ,  \quad \mbox{if} \ \gamma < 1 , \\
		\left[\frac{\widetilde Y (t)}{Y (t)}\right]^\gamma \geq 1 , \quad \mbox{if} \ \gamma > 1 .
	\end{aligned}
\right.
\end{align}	
Thus, $[\widetilde Y (t)/Y (t)]^\frac{\gamma}{1-\gamma} \leq 1$ for any $\gamma \in (0, 1) \cup (1, \infty)$, implying $L (t) \geq 0$. On the other hand, combining Lemma \ref{lem:BSDE} with Lemma \ref{lem:BSDE-wt}, we have $[\widetilde Y (t)/Y (t)]^\frac{\gamma}{1-\gamma} \geq 0$ and thereby $L (t) \leq 1$. This completes the proof.
\end{proof}

\section{Special Case in Deterministic Setting}
\label{sec:specialdeterministic}

In this section, we consider the robust optimal investment-consumption problem when 
all the model coefficients including $\{r (t)\}_{t\in[0,T]}$, $\{\rho (t)\}_{t\in[0,T]}$, $\{\mu_i(t)\}_{t\in[0,T]}$ and $\{\sigma_{ij}(t)\}_{t\in[0,T]}$ for $i=1,2,\ldots,m$ and $j=1,2,\ldots,n$, are deterministic functions of time $t$. As in the previous section, we omit the dependence of these deterministic functions on time index $t$ below whenever needed for ease of notation.
In the deterministic case, it is obvious that $Z (t) \equiv 0$ and BSDE \eqref{eq:BSDE} reduces to an ODE
\begin{align}\label{eq:Y-sol-deter-diff}
- \frac{d Y (t)}{d t} = \frac{Y (t)}{Y (t) \wedge \frac{1}{\underline c} \vee \frac{1}{\overline c}} + Q (t) Y (t) , \qquad Y(T) = \beta^{\frac{1}{\gamma}},
\end{align}
where 
\begin{align}
Q (t) : = \frac{1}{\gamma} \bigg [ - \rho + (1-\gamma) r
+ \frac{1 - \gamma}{2 \gamma} \theta^\top \bigg ( I + \frac{1}{\gamma} H \bigg )^{-1} \theta \bigg ] - \frac{1-\gamma}{2} \Dist^2_{\widehat \Gamma} \bigg [ \bigg ( I + \frac{1}{\gamma} H \bigg )^{-\frac{1}{2}} \frac{\theta}{\gamma}  \bigg ] . 
\end{align}
The solution to \eqref{eq:Y-sol-deter-diff} is given by
\begin{align}\label{eq:Y-sol-deter}
Y (t) = \beta^{\frac{1}{\gamma}} \exp \bigg ( \int^T_t Q (s) d s \bigg ) + \int^T_t \frac{Y (s)}{Y (s) \wedge \frac{1}{\underline c} \vee \frac{1}{\overline c}} \exp \bigg ( \int^s_t Q (\nu) d \nu \bigg ) d s .
\end{align}
%{\color{cyan} [Should there be an extra pair of brackets in the exponents, since both terms in the exponents are being integrated?]}
As in Section \ref{sec:main}, when the investor is ambiguity-neutral, i.e. $H = 0_{n \times n}$, the solution to the above ODE \eqref{eq:Y-sol-deter-diff}, denoted by $Y_0 (t)$, has the following explicit form 
\begin{align}\label{eq:Y0-sol-deter}
Y_0 (t) = \beta^{\frac{1}{\gamma}} \exp \bigg ( \int^T_t Q_0 (s) d s \bigg ) + \int^T_t \frac{Y (s)}{Y (s) \wedge \frac{1}{\underline c} \vee \frac{1}{\overline c}} \exp \bigg ( \int^s_t Q_0 (\nu) d \nu \bigg ) d s ,
\end{align}
where 
\begin{align}
Q_0 (t) : = \frac{1}{\gamma} \bigg [ - \rho + (1-\gamma) r
+ \frac{1 - \gamma}{2 \gamma} |\theta|^2 \bigg ] - \frac{1-\gamma}{2} \Dist^2_{\widehat \Gamma} \bigg (  \frac{\theta}{\gamma}  \bigg ) . 
\end{align}

% In the deterministic case, the robust optimal investment-consumption strategies, the optimal distortion process, and the value function that are presented in
% Proposition \ref{prop:HJBI-solution} can be simplified as
% \begin{align}\label{eq:inv-con-deter}
% p^* (t) = \bigg ( I + \frac{1}{\gamma} H \bigg )^{-\frac{1}{2}} \Proj_{\widehat \Gamma} \bigg \{ \bigg ( I + \frac{1}{\gamma} H \bigg )^{-\frac{1}{2}} \frac{\theta (t)}{\gamma} \bigg \} , \qquad c^* (t) = \frac{1}{Y (t)} \vee \underline c \wedge \overline c ,
% \end{align}
% \begin{align}
% \phi^* (t) = - H \bigg ( I + \frac{1}{\gamma} H \bigg )^{-\frac{1}{2}} \Proj_{\widehat \Gamma} \bigg \{ \bigg ( I + \frac{1}{\gamma} H \bigg )^{-\frac{1}{2}} \frac{\theta (t)}{\gamma} \bigg \} ,
% \end{align}
% and
% \begin{align}\label{eq:value-fun-deter}
% V (t, x) = \frac{x^{1-\gamma}}{1-\gamma} \times [Y (t)]^{\gamma} ,
% \end{align}
% where $Y (t)$ is given by \eqref{eq:Y-sol-deter}.

Moreover, it can shown that BSDE \eqref{eq:BSDE-wt} becomes a Bernoulli equation
\begin{align}\label{eq:Bernoulli}
	- \frac{d \widetilde Y (t)}{d t} 
 = \frac{1}{\gamma} \bigg \{ \bigg [ \frac{\widetilde Y (t)}{Y_0 (t) \wedge \frac{1}{\underline c} \vee \frac{1}{\overline c}} \bigg ]^{1-\gamma} + \widetilde Q (t) \widetilde Y (t) \bigg \}   , \qquad \widetilde{Y}(T) = \beta^{\frac{1}{\gamma}},
\end{align}
where 
\begin{align}
\widetilde Q (t) : = (1-\gamma) \bigg \{ - \frac{1}{Y_0 \wedge \frac{1}{\underline c} \vee \frac{1}{\overline c}} - \frac{\gamma}{2} \Proj^\top_{\widehat \Gamma} \bigg ( \frac{\theta}{\gamma} \bigg ) \bigg ( I + \frac{H}{\gamma} \bigg )  \Proj_{\widehat \Gamma} \bigg ( \frac{\theta}{\gamma} \bigg )  
	+ \frac{1}{1-\gamma} \bigg [ - \rho + (1-\gamma) r
	\bigg ] + \Proj^\top_{\widehat \Gamma} \bigg ( \frac{\theta}{\gamma}  \bigg ) \theta \bigg \} .
\end{align}
% {\color{blue}
% [Correction: The above ODE should be corrected as follows:
% \begin{align*}
% 	- \frac{d \widetilde Y (t)}{d t} 
%  = \frac{1}{\gamma} \bigg \{ \bigg [ \frac{\widetilde Y}{Y_0 \wedge \frac{1}{\underline c} \vee \frac{1}{\overline c}} \bigg ]^{1-\gamma} + \Psi (t) \widetilde Y\bigg \}   ,
% \end{align*}
% where 
% \begin{align}
% \Psi (t) : = (1-\gamma) \bigg \{ - \frac{1}{Y_0 \wedge \frac{1}{\underline c} \vee \frac{1}{\overline c}} - \frac{\gamma}{2} {\color{magenta} \Proj^\top_{\widehat \Gamma} \bigg ( \frac{\theta}{\gamma} \bigg )  \bigg ( I + \frac{H}{\gamma} \bigg ) \Proj_{\widehat \Gamma} \bigg ( \frac{\theta}{\gamma} \bigg ) } 
% 	+ \frac{1}{1-\gamma} \bigg [ - \rho + (1-\gamma) r
% 	\bigg ] + \Proj^\top_{\widehat \Gamma} \bigg ( \frac{\theta}{\gamma}  \bigg ) \theta \bigg \} .
% \end{align}
% ]
% }
A simple application of transformation techniques gives the explicit solution to the Bernoulli equation \eqref{eq:Bernoulli}:
\begin{align}\label{eq:Y-wt-sol-deter}
	\widetilde Y (t) = \bigg\{ \beta \exp \bigg ( \int^T_t \widetilde Q (s) d s \bigg ) + \int^T_t \frac{1}{[Y_0 (s) \wedge \frac{1}{\underline c} \vee \frac{1}{\overline c}]^{1-\gamma}} \exp \bigg ( \int^s_t \widetilde Q (\nu) d \nu \bigg ) d s \bigg \}^{\frac{1}{\gamma}} .
\end{align}

The following corollary summarises the optimal solution 
to Problem \eqref{eq:problem} in the deterministic case.

\begin{corollary}
\label{cor:DeterministicCaseSolution}
When all the model coefficients are deterministic, the robust optimal investment-consumption strategy, the optimal distortion process, the value function 
and the utility loss are given by 
\begin{align}\label{eq:inv-con-deter}
p^* (t) & = \bigg ( I + \frac{1}{\gamma} H \bigg )^{-\frac{1}{2}} \Proj_{\widehat \Gamma} \bigg \{ \bigg ( I + \frac{1}{\gamma} H \bigg )^{-\frac{1}{2}} \frac{\theta (t)}{\gamma} \bigg \} , \qquad c^* (t) = \frac{1}{Y (t)} \vee \underline c \wedge \overline c \\
\label{eq:distortion-deter}
\phi^* (t) & = - H \bigg ( I + \frac{1}{\gamma} H \bigg )^{-\frac{1}{2}} \Proj_{\widehat \Gamma} \bigg \{ \bigg ( I + \frac{1}{\gamma} H \bigg )^{-\frac{1}{2}} \frac{\theta (t)}{\gamma} \bigg \}, \\
\label{eq:value-fun-deter}
V (t, x) & = \frac{x^{1-\gamma}}{1-\gamma} \times [Y (t)]^{\gamma} , \\
\label{eq:utility-loss-deter}
L (t) & = 1 - \left[\frac{\widetilde{Y}(t)}{Y(t)}\right]^\frac{\gamma}{1-\gamma} ,
\end{align}
in which $Y$ and $\widetilde Y$ are computed explicitly as in \eqref{eq:Y-sol-deter} and \eqref{eq:Y-wt-sol-deter}.
\end{corollary}

\begin{proof}
Applying the results obtained in Section \ref{sec:main} and the explicit solutions \eqref{eq:Y-sol-deter} and \eqref{eq:Y-wt-sol-deter} immediately yields the desired results
\eqref{eq:inv-con-deter}, \eqref{eq:distortion-deter}, \eqref{eq:value-fun-deter}, and \eqref{eq:utility-loss-deter}.
\end{proof}

\subsection{Comparative Statics}\label{subsec:CS}

In this subsection, we investigate the impact of the ambiguity aversion on the investor's optimal strategies and value function. To that end, we focus on a special portfolio constraint, namely, no-short-selling constraint $\Gamma : = {\mathbb R}_+^n$. It is obvious that ${\widehat \Gamma} = ( I + \frac{1}{\gamma} H)^{\frac{1}{2}} \Gamma = {\mathbb R}_+^n$. 

First, we note that the optimal investment strategy and the optimal distortion process
can be further specified by
$p^* (t) = ( \gamma I + H )^{-1} \theta^+ (t)$
and
$\phi^* (t) = \big [ \big( I + \frac{H}{\gamma} \big)^{-1} - I \big ] \theta^+ (t) ,$
where $\theta^+ (t)$ denotes an $n$-dimensional vector with each entry being the positive part of the corresponding entry in the market price of risk vector $\theta (t)$. Obviously, in the deterministic case, only the investment constraint affects the optimal investment strategy and the optimal distortion process. With the no-short-selling constraint, both the optimal investment strategy and the optimal distortion process become more conservative compared with the case without such constraint. The optimal distortion process is now an $n$-dimensional vector with all entries being non-positive. Hence, in the worst-case scenario, the investor is only concerned with the lower-than-expected market price of risk due to Brownian shocks. This is consistent with intuition because the investor can only hold long positions in stocks due to the no-short-selling constraint.
Moreover, everything else being identical, the more ambiguity-averse the investor is,
the smaller norms both the optimal investment strategy $p^* (t)$ and the optimal distortion process $\phi^* (t)$ have.

In the following proposition, we investigate the effect of the ambiguity aversion coefficient and the investment/consumption constraints
on the robust optimal consumption strategy and the value function.

\begin{proposition}\label{prop:av-coefficnet}
    Suppose that all the model coefficients are deterministic and short-selling is not allowed, i.e. $\widehat{\Gamma} = \mathbb{R}^n_+$. Then the following assertions hold true.
	\begin{enumerate}
		\item[(i)] If $\gamma > 1$, then the robust optimal consumption strategy $c^*$ is a decreasing function of $\eta_i$, for each $i = 1, 2, \ldots, n$.
		\item[(i')] If $0 < \gamma < 1$, then the robust optimal consumption strategy $c^*$ is an increasing function of $\eta_i$, for each $i = 1, 2, \ldots, n$.
		\item[(ii)] For any $\gamma \in (0,1) \cup (1, \infty)$, then the value function $v (t, x)$ is a decreasing function of $\eta_i$, for each $i = 1, 2, \ldots, n$.
	\end{enumerate}
\end{proposition}

\begin{proof}
Before proving the three assertions, we note ${\widehat \Gamma} = {\mathbb R}_+^n$ is a (convex) cone, in which the Pythagorean theorem holds true. Thus, the last two terms in the exponential functions in \eqref{eq:Y-sol-deter} can be combined as follows:
\begin{align}
& \frac{1 - \gamma}{2 \gamma^2} \theta^\top \bigg ( I + \frac{1}{\gamma} H \bigg )^{-1} \theta - \frac{1 - \gamma}{2} \Dist^2_{\widehat \Gamma} \bigg [ \bigg ( I + \frac{1}{\gamma} H \bigg )^{-\frac{1}{2}} \frac{\theta}{\gamma}  \bigg ] \\
& = \frac{1 - \gamma}{2 \gamma^2} \Proj^2_{\widehat \Gamma} \bigg \{ \bigg ( I + \frac{1}{\gamma} H \bigg )^{-\frac{1}{2}} \theta \bigg \} = \frac{1 - \gamma}{2 \gamma^2} \left| \bigg ( I + \frac{1}{\gamma} H \bigg )^{-\frac{1}{2}} \theta^+ \right|^2 = \frac{1 - \gamma}{2 \gamma} \sum^n_{i = 1} \frac{(\theta^+_i)^2}{\gamma + \eta_i} .
\end{align}

{\it Assertions (i) and (i').}
Denote by $g (y) := \frac{y}{y \wedge \frac{1}{\underline c} \vee \frac{1}{\overline c}}$. Taking partial derivative with respect to $\eta_i$ on both sides of ODE \eqref{eq:Y0-sol-deter} gives
\begin{align}
	- \frac{d}{d t} \frac{\partial Y (t)}{\partial \eta_i} = \bigg \{ g^\prime (Y (t)) + \frac{1}{\gamma} \big [ - \rho + (1-\gamma) r \big ] + \frac{1 - \gamma}{2 \gamma} \sum^n_{j = 1} \frac{(\theta^+_j)^2}{\gamma + \eta_j} \bigg \} \frac{\partial Y (t)}{\partial \eta_i}
	- \frac{1 - \gamma}{2 \gamma} \frac{(\theta^+_i)^2}{(\gamma + \eta_i)^2}  Y (t) , \quad \frac{\partial Y (T)}{\partial \eta_i}  = 0 .
\end{align}
Therefore,
\begin{align}
	\frac{\partial Y (t)}{\partial \eta_i} = - \frac{1 - \gamma}{2 \gamma} \frac{(\theta^+_i)^2}{(\gamma + \eta_i)^2} \int^T_t Y (s) \exp \bigg \{ g^\prime (Y (s)) + \frac{1}{\gamma} \big [ - \rho + (1-\gamma) r \big ] + \frac{1 - \gamma}{2 \gamma} \sum^n_{j = 1} \frac{(\theta^+_j)^2}{\gamma + \eta_j} \bigg \} d s 
\end{align}
is positive (resp. negative) if $\gamma > 1$ (resp. $0 < \gamma < 1$). 
In other words, $Y (t)$ is an increasing (resp. a decreasing) function of $\eta_i$, for each $i = 1, 2, \ldots, n$, if $\gamma > 1$ (resp. $0 < \gamma < 1$).
Since the optimal consumption strategy is related
to the reciprocal of $Y$ as in \eqref{eq:inv-con-deter},
it decreases (resp. increases) with the ambiguity aversion coefficient $\eta_i$ if $\gamma > 1$ (resp. $0 < \gamma < 1$).

{\it Assertion (ii).} The explicit expression of the value function \eqref{eq:value-fun-deter} implies that its monotonicity w.r.t. $\eta_i$ is jointly determined by that of $[Y (t)]^\gamma$ and the sign of $1-\gamma$.
Therefore, the value function always decreases with the ambiguity aversion coefficient $\eta_i$ regardless of the value of the coefficient of relative risk aversion $\gamma$ as long as it is in $(0,1) \cup (1, \infty)$. 
%{\color{cyan} [Do you mean ``...regardless of the value of the coefficient of relative risk aversion $\gamma$?]}.
\end{proof}

The proposition above suggests a relationship between the investor's robust optimal consumption strategy (and the elasticity of intertemporal substitution implied by the optimal path) and their preference for robustness when she is only allowed to take long positions in the investment strategy. For a power utility investor with coefficient of relative risk aversion $\gamma > 0$, her elasticity of intertemporal substitution (EIS) is $\frac{1}{\gamma}$. Thus, an investor with a high (resp. low) level of risk aversion, $\gamma > 1$ (resp. $0 < \gamma < 1$), tends to prioritize (resp. defer) present consumption due to a low (resp. high) EIS.
However, when the investor is ambiguity-averse, she recognises the greater uncertainty around how the financial risk factors and environment will evolve in the future. %especially in relation to the real interest rate. 
Thus, an ambiguity-averse investor who has a high (resp. low) level of risk aversion tends to have a lower (resp. higher) overall level of consumption compared to an ambiguity-neutral investor who has a high (resp. low) level of risk aversion. Furthermore, as the investor's preference for robustness increases, her overall level of consumption decreases (resp. increases). As a result, an ambiguity-averse investor tends to have a lower level of welfare (as measured by the value function) compared to an ambiguity-neutral investor and the welfare decreases as she has a higher preference for robustness.

As a final note in this subsection, we examine how the portfolio (investment/consumption) constraints affect the investor's optimal consumption strategy and value function. To facilitate our comparison, we define
\begin{itemize}
	\item $c^*_{\rm C_1} (t)$ and $V_{\rm C_1} (t, x)$: Robust optimal consumption strategy and value function when there are both investment and upper consumption constraints;
	\item $c^*_{\rm C_2} (t)$ and $V_{\rm C_2} (t, x)$: Robust optimal consumption strategy and value function when there are both investment and lower consumption constraints;
	\item $c^*_{\rm C_3} (t)$ and $V_{\rm C_3} (t, x)$: Robust optimal consumption strategy and value function when there is only investment constraint;
	\item $c^*_{\rm C_4} (t)$ and $V_{\rm C_4} (t, x)$: Robust optimal consumption strategy and value function when there is only upper consumption constraint;
	\item $c^*_{\rm C_5} (t)$ and $V_{\rm C_5} (t, x)$: Robust optimal consumption strategy and value function when there is only lower consumption constraint;
	\item $c^*_{\rm NC} (t)$ and $V_{\rm NC} (t, x)$: Robust optimal consumption strategy and value function when there is neither investment constraint nor lower/upper consumption constraints.
\end{itemize}
The comparison results are summarised in the following proposition.

\begin{proposition} \label{prop:det-comparison}
    Suppose all the model coefficients are deterministic. Then the following assertions hold true.
	\begin{enumerate}
		\item[(i)] If $\gamma > 1$, then the robust consumption rules satisfy $c^*_{\rm C_1} (t)  \leq  c^*_{\rm C_4} (t)$, $c^*_{\rm C_2} (t) \leq c^*_{\rm C_5} (t)$ and $c^*_{\rm C_3} (t)  \leq  c^*_{\rm NC} (t)$;
		\item[(i')] If $0 < \gamma < 1$, then the robust consumption rules satisfy $c^*_{\rm C_1} (t) \geq c^*_{\rm C_4} (t)$, $c^*_{\rm C_2} (t) \geq c^*_{\rm C_5} (t)$, and $c^*_{\rm C_3} (t) \geq c^*_{\rm NC} (t)$;
		\item[(ii)] For any $\gamma \in (0,1) \cup (1, \infty)$, then the value functions satisfy $V_{\rm C_1} (t, x) \leq V_{\rm C_4} (t, x)$, $V_{\rm C_2} (t, x) \leq V_{\rm C_5} (t, x)$, and $V_{\rm C_3} (t, x) \leq  V_{\rm NC} (t, x)$;
		\item[(iii)] If $\gamma > 1$, then the value functions satisfy $V_{\rm C_1} (t, x) \geq V_{\rm C_3} (t, x)$, $V_{\rm C_2} (t, x) \leq V_{\rm C_3} (t, x)$, $V_{\rm C_4} (t, x) \geq V_{\rm NC} (t, x)$,
		and $V_{\rm C_5} (t, x) \leq V_{\rm NC} (t, x)$;
		\item[(iii')] If $0 < \gamma < 1$, then the value functions satisfy $V_{\rm C_1} (t, x) \leq V_{\rm C_3} (t, x)$, $V_{\rm C_2} (t, x) \geq V_{\rm C_3} (t, x)$, $V_{\rm C_4} (t, x) \leq V_{\rm NC} (t, x)$,
		and $V_{\rm C_5} (t, x) \geq V_{\rm NC} (t, x)$.
	\end{enumerate}
\end{proposition}

\begin{proof}
To facilitate the following comparison, we highlight
all quantities related to the portfolio constraint cases and the no portfolio constraint case by using subscripts ``${\rm C}_i$" and ``NC", for $i = 1, 2, 3, 4, 5$. In what follows, it suffices to compare $Y_{{\rm C}_i}$ and $Y_{\rm NC}$. Let us consider the effect of the investment constraint first. To that end, we consider the ODE for $\Delta_3 Y : = Y_{\rm C_3} - Y_{\rm NC}$:
\begin{align}
	- \frac{d \Delta_3 Y (t)}{d t} =& \ \frac{1}{\gamma} \bigg [ - \rho + (1-\gamma) r
	+ \frac{1 - \gamma}{2 \gamma} \theta^\top \bigg ( I + \frac{1}{\gamma} H \bigg )^{-1} \theta \bigg ] \Delta_3 Y
	- \frac{1-\gamma}{2} \Dist^2_{\widehat \Gamma} \bigg [ \bigg ( I + \frac{1}{\gamma} H \bigg )^{-\frac{1}{2}} \frac{\theta}{\gamma}  \bigg ] Y_{\rm C_3} , \qquad \Delta_3 Y (T) = 0 ,
\end{align}
whose solution is given by
\begin{align}
	\Delta_3 Y (t) = - \frac{1-\gamma}{2} \int^T_t \Dist^2_{\widehat \Gamma} \bigg [ \bigg ( I + \frac{1}{\gamma} H \bigg )^{-\frac{1}{2}} \frac{\theta}{\gamma}  \bigg ] Y_{\rm C_3} (s)
	\exp \bigg \{ \int^s_t \frac{1}{\gamma} \bigg [ - \rho + (1-\gamma) r
	+ \frac{1 - \gamma}{2 \gamma} \theta^\top \bigg ( I + \frac{1}{\gamma} H \bigg )^{-1} \theta \bigg ] d \nu \bigg \} d s .
\end{align}
Thus, if $\gamma > 1$ (resp. $0 < \gamma < 1$), then $Y_{{\rm C}_3} (t) \geq Y_{\rm NC} (t)$ (resp. $Y_{\rm C_3} (t) \leq Y_{\rm NC} (t)$). Similarly, we can show if $\gamma > 1$ (resp. $0 < \gamma < 1$), then $Y_{\rm C_1} (t) \geq Y_{\rm C_4} (t)$ and $Y_{\rm C_2} (t) \geq Y_{\rm C_5} (t)$ (resp. $Y_{\rm C_1} (t) \leq Y_{\rm C_4} (t)$ and $Y_{\rm C_2} (t) \leq Y_{\rm C_5} (t)$).

Therefore, we have
\begin{enumerate}
	\item[\it (i)] if $\gamma > 1$, then $c^*_{\rm C_3} (t) = \frac{1}{Y_{\rm C_3} (t)} \leq \frac{1}{Y_{\rm NC} (t)} = c^*_{\rm NC} (t)$, $c^*_{\rm C_1} (t) = \frac{1}{Y_{\rm C_1} (t)} \wedge \overline c \leq \frac{1}{Y_{\rm C_4} (t)} \wedge \overline c = c^*_{\rm C_4} (t)$, and $c^*_{\rm C_2} (t) = \frac{1}{Y_{\rm C_2} (t)} \vee \underline c \leq \frac{1}{Y_{\rm C_5} (t)} \vee \underline c = c^*_{\rm C_5} (t)$;
	\item[\it (i')] if $0 < \gamma < 1$, then $c^*_{\rm C_3} (t) = \frac{1}{Y_{\rm C_3} (t)} \geq \frac{1}{Y_{\rm NC} (t)} = c^*_{\rm NC} (t)$, $c^*_{\rm C_1} (t) = \frac{1}{Y_{\rm C_1} (t)} \wedge \overline c \geq \frac{1}{Y_{\rm C_4} (t)} \wedge \overline c = c^*_{\rm C_4} (t)$, and $c^*_{\rm C_2} (t) = \frac{1}{Y_{\rm C_2} (t)} \vee \underline c \geq \frac{1}{Y_{\rm C_5} (t)} \vee \underline c = c^*_{\rm C_5} (t)$;
	\item[\it (ii)] For any $\gamma \in (0,1) \cup (1, \infty)$, $V_{\rm C_3} (t, x) = \frac{x^{1-\gamma}}{1-\gamma} \times [Y_{\rm C_3} (t)]^{\gamma} \leq \frac{x^{1-\gamma}}{1-\gamma} \times [Y_{\rm NC} (t)]^{\gamma} = V_{\rm NC} (t, x)$, $V_{\rm C_1} (t, x) = \frac{x^{1-\gamma}}{1-\gamma} \times [Y_{\rm C_1} (t)]^{\gamma} \leq \frac{x^{1-\gamma}}{1-\gamma} \times [Y_{\rm C_4} (t)]^{\gamma} = V_{\rm C_4} (t, x)$, and $V_{\rm C_2} (t, x) = \frac{x^{1-\gamma}}{1-\gamma} \times [Y_{\rm C_2} (t)]^{\gamma} \leq \frac{x^{1-\gamma}}{1-\gamma} \times [Y_{\rm C_5} (t)]^{\gamma} = V_{\rm C_5} (t, x)$.
\end{enumerate}

We next consider the effect of the upper and lower consumption constraints.
Differencing the ODEs for $Y_{\rm C_4}$ and $Y_{\rm NC}$ and denoting $\Delta_4 Y : = Y_{\rm C_4} - Y_{\rm NC}$, we obtain
\begin{align}
	- \frac{d \Delta_4 Y (t)}{d t} = \frac{Y_{\rm C_4}}{Y_{\rm C_4} \vee \frac{1}{\overline c}} - 1 + \frac{1}{\gamma} \bigg [ - \rho + (1-\gamma) r
	+ \frac{1 - \gamma}{2 \gamma} \theta^\top \bigg ( I + \frac{1}{\gamma} H \bigg )^{-1} \theta \bigg ] \Delta_4 Y ,
	\qquad \Delta_4 Y (T) = 0 ,
 \end{align}
which admits the following solution
\begin{align}
	\Delta_4 Y (t) = \int^T_t \bigg ( \frac{Y_{\rm C_4} (s)}{Y_{\rm C_4} (s) \vee \frac{1}{\overline c}} - 1 \bigg )
	\exp \bigg \{ \int^s_t \frac{1}{\gamma} \bigg [ - \rho + (1-\gamma) r
	+ \frac{1 - \gamma}{2 \gamma} \theta^\top \bigg ( I + \frac{1}{\gamma} H \bigg )^{-1} \theta \bigg ] d \nu \bigg \} d s \leq 0
\end{align}
implying $Y_{\rm C_4} (t) \leq Y_{\rm NC} (t)$. In the same vein, we can show that $Y_{\rm C_1} (t) \leq Y_{\rm C_3} (t)$, $Y_{\rm C_2} (t) \geq Y_{\rm C_3} (t)$, and $Y_{\rm C_5} (t) \geq Y_{\rm NC} (t)$.

Therefore, we have
\begin{enumerate}
	\item[\it (iii)] if $\gamma > 1$, then $V_{\rm C_4} (t, x) = \frac{x^{1-\gamma}}{1-\gamma} \times [Y_{\rm C_4} (t)]^{\gamma} \geq \frac{x^{1-\gamma}}{1-\gamma} \times [Y_{\rm NC} (t)]^{\gamma} = V_{\rm NC} (t, x)$, $V_{\rm C_1} (t, x) = \frac{x^{1-\gamma}}{1-\gamma} \times [Y_{\rm C_1} (t)]^{\gamma} \geq \frac{x^{1-\gamma}}{1-\gamma} \times [Y_{\rm C_3} (t)]^{\gamma} = V_{\rm C_3} (t, x)$, $V_{\rm C_2} (t, x) = \frac{x^{1-\gamma}}{1-\gamma} \times [Y_{\rm C_2} (t)]^{\gamma} \leq \frac{x^{1-\gamma}}{1-\gamma} \times [Y_{\rm C_3} (t)]^{\gamma} = V_{\rm C_3} (t, x)$
	and $V_{\rm C_5} (t, x) = \frac{x^{1-\gamma}}{1-\gamma} \times [Y_{\rm C_5} (t)]^{\gamma} \leq \frac{x^{1-\gamma}}{1-\gamma} \times [Y_{\rm NC} (t)]^{\gamma} = V_{\rm NC} (t, x)$;
	\item[\it (iii')] if $0 < \gamma < 1$, then $V_{\rm C_4} (t, x) = \frac{x^{1-\gamma}}{1-\gamma} \times [Y_{\rm C_4} (t)]^{\gamma} \leq \frac{x^{1-\gamma}}{1-\gamma} \times [Y_{\rm NC} (t)]^{\gamma} = V_{\rm NC} (t, x)$, $V_{\rm C_1} (t, x) = \frac{x^{1-\gamma}}{1-\gamma} \times [Y_{\rm C_1} (t)]^{\gamma} \leq \frac{x^{1-\gamma}}{1-\gamma} \times [Y_{\rm C_3} (t)]^{\gamma} = V_{\rm C_3} (t, x)$, $V_{\rm C_2} (t, x) = \frac{x^{1-\gamma}}{1-\gamma} \times [Y_{\rm C_2} (t)]^{\gamma} \geq \frac{x^{1-\gamma}}{1-\gamma} \times [Y_{\rm C_3} (t)]^{\gamma} = V_{\rm C_3} (t, x)$
	and $V_{\rm C_5} (t, x) = \frac{x^{1-\gamma}}{1-\gamma} \times [Y_{\rm C_5} (t)]^{\gamma} \geq \frac{x^{1-\gamma}}{1-\gamma} \times [Y_{\rm NC} (t)]^{\gamma} = V_{\rm NC} (t, x)$.
\end{enumerate}
This completes the proof.
\end{proof}

An economic insight can be immediately gained from Proposition \ref{prop:det-comparison}. That is, if the welfare is measured by the value function, the investor is always worse off whenever there is an investment constraint. However, depending on the investor's risk aversion level, the robust optimal consumption strategy represented as a proportion of wealth, can be either larger or smaller. To be more specific, if the investor has a high (resp. low) risk aversion level, then the optimal consumption proportion is smaller (resp. larger) when an investment constraint is present.
 
It is worth highlighting that the impact of consumption constraints on the robust optimal consumption strategy is more subtle. Though we have partly addressed the monotonicity issue of $Y$ with respect to the lower and upper consumption constraints, it is unclear how these constraints would affect the robust optimal consumption strategy, which is determined by $c^* (t) = \frac{1}{Y(t)} \vee \underline c \wedge \overline c$ through the combined effects of $Y (t)$, floor $\underline c$ and ceiling $\overline c$. Thus, 
we will revisit the impact 
of consumption constraints on $c^*$ in the next subsection.

\subsection{Numerical Illustrations}

In this subsection, we provide numerical illustrations of the results established in the previous subsection. We consider a time horizon of $T = 3$ and a financial market with $m = 2$ risky assets whose price processes are driven by $n = 3$ independent Brownian motions. 
The financial market parameters, utility function parameters, subjective discount factor, and ambiguity aversion parameters used in these illustrations are provided in Table \ref{tab:det-par-values}. These parameter values are comparable to those used in similar studies; see e.g. \citep{BalterMahayniSchweizer2021,  KraftSeifriedSteffensen2013, Maenhout2004}. In the numerical implementation, we first set the constraints of consumption rate to be $[\underline{c}, \overline{c}] = [0, \infty)$, 
which is essentially equivalent to the cases without consumption constraints and is considered in the cases $\mbox{C}_3$ and $\mbox{NC}$ in Subsection \ref{subsec:CS}. Doing so allows us to focus on the impact of investment constraints first. Later we will consider more stringent consumption constraints as other cases (i.e. $\mbox{C}_1$, $\mbox{C}_2$, $\mbox{C}_4$, and $\mbox{C}_5$) discussed in Subsection \ref{subsec:CS}.

\begin{table}[h]
\caption{Parameter values used for numerical illustrations in the case of deterministic coefficients.}
\label{tab:det-par-values}
\centering
\begin{tabular}{@{}cc@{}}
\toprule
\textbf{Parameter} & \textbf{Value} \\ \midrule
$\gamma$ & $\gamma \in \{0.9, 4\}$ \\
$\mu_i$ ($i=1,2$), $r$ & $\mu_1 = 0.09$, $\mu_2 = 0.11$, $r = 0.05$ \\
$\sigma = \left[\begin{array}{ccc} \sigma_{11} & \sigma_{12} & \sigma_{13} \\ \sigma_{21} & \sigma_{22} & \sigma_{33} \end{array}\right]$ & $\sigma = \left[\begin{array}{ccc} 0.050 & 0.066 & 0.082 \\ 0.058 & 0.0740 & 0.090 \end{array}\right]$ \\
$\beta$ & $\beta = 1$ \\
$\rho$ & $\rho = 0.015$ \\
$\eta_i$ ($i=1,2,3$) & $\eta_1 = 1$, $\eta_2 = 3$, $\eta_3 = 5$ \\ \bottomrule
\end{tabular}
\end{table}

Given the parameter values above, the market price of risk vector $\theta(t)$ is computed as \[\theta(t) = \theta = \sigma^\top(\sigma\sigma^\top)^{-1}(\mu - r\mathbf{1}_m)= \left[\begin{array}{ccc} 4.7396 & 0.8333 & -3.0729\end{array}\right]^\top.\] Since $\theta (t)$ is a constant vector, the optimal investment (more precisely, risk exposure) strategy $p^*(t)$ and the optimal distortion process $\phi^*(t)$ are also constant vectors.

The first set of illustrations compares the optimal investment strategy, the optimal distortion process, the optimal consumption strategy, the utility loss, and the value function when short-selling is allowed ($\Gamma = \mathbb{R}^n$) and not allowed ($\Gamma = \mathbb{R}^n_+$). We also compare the results when the investor is ambiguity-averse or ambiguity-neutral; see Table \ref{tab:det-comparison-cases} for a summary of the cases considered. Table \ref{tab:det-opt-strat-distort-comparison} shows the optimal investment strategy $p^*(t) = p^*$ and the optimal distortion vector $\phi^*(t) = \phi^*$ for the different cases considered. Figure \ref{fig:det-comparison} exhibits the optimal consumption strategy, the utility loss over time, and the value function over time for a fixed wealth of $x=1$. Clearly, $\phi^*(t) = 0$ when the investor is ambiguity-neutral.

\begin{table}[h]
\caption{The different cases considered for the short-selling constraint and the investor's ambiguity aversion.}
\label{tab:det-comparison-cases}
\centering
\begin{tabular}{@{}ccc@{}}
\toprule
Case & Short-selling & Ambiguity-averse \\ \midrule
Case 1 & Allowed & Yes ($H \neq 0_{n\times n}$) \\
Case 2 & Allowed & No ($H = 0_{n\times n}$) \\
Case 3 & Not allowed & Yes ($H \neq 0_{n\times n}$) \\
Case 4 & Not allowed & No ($H = 0_{n\times n}$) \\ \bottomrule
\end{tabular}
\end{table}

\begin{table}[h]
\caption{The optimal investment strategy $p^*(t) = p^*$ and the optimal distortion vector $\phi^*(t) = \phi^*$ for the cases considered in Table \ref{tab:det-comparison-cases}}
\label{tab:det-opt-strat-distort-comparison}
\centering
\begin{tabular}{@{}crrrrrrrrr@{}}
\toprule
\textbf{} & \multicolumn{4}{c}{\textbf{$p^*(t) = p^*$}} & \multicolumn{1}{c}{} & \multicolumn{4}{c}{$\phi^*(t) = \phi^*$} \\ \cmidrule(lr){2-5} \cmidrule(l){7-10} 
 & \multicolumn{1}{c}{Case 1} & \multicolumn{1}{c}{Case 2} & \multicolumn{1}{c}{Case 3} & \multicolumn{1}{c}{Case 4} & \multicolumn{1}{c}{} & \multicolumn{1}{c}{Case 1} & \multicolumn{1}{c}{Case 2} & \multicolumn{1}{c}{Case 3} & \multicolumn{1}{c}{Case 4} \\ \cmidrule(r){1-5} \cmidrule(l){7-10} 
\multirow{3}{*}{$\gamma = 4$} & 0.9479 & 1.1849 & 0.9479 & 1.1849 &  & -0.9479 & 0 & -0.9479 & 0 \\
 & 0.1190 & 0.2083 & 0.1190 & 0.2083 &  & -0.3571 & 0 & -0.3571 & 0 \\
 & -0.3414 & -0.7682 & 0 & 0 &  & 1.7072 & 0 & 0 & 0 \\ \cmidrule(r){1-5} \cmidrule(l){7-10} 
\multirow{3}{*}{$\gamma = 0.9$} & 2.4945 & 5.2662 & 2.4945 & 5.2662 &  & -2.4945 & 0 & -2.4945 & 0 \\
 & 0.2137 & 0.9259 & 0.2137 & 0.9259 &  & -0.6410 & 0 & -0.6410 & 0 \\
 & -0.5208 & -3.4144 & 0 & 0 &  & 2.6042 & 0 & 0 & 0 \\ \bottomrule
\end{tabular}
\end{table}

% \begin{table}[]
% \caption{The optimal risk exposure strategy $p^*(t) = p^*$ and the optimal distortion vector $\phi^*(t) = \phi^*$ when short-selling is allowed and not allowed. 
% {\color{red} [Comment: It sounds a bit confusing to use "No short-selling" and "No Investment Constraint" because "No short-selling" sounds like no short-selling constraint. in the table. Also this is inconsistent with other places when we call the two cases "short-selling allowed/permitted" and "short-selling not allowed".]}}
% \label{tab:det-opt-strat-distort-comparison}
% \centering
% \begin{tabular}{@{}cccc@{}}
% \toprule
% \multicolumn{2}{c}{No short-selling} & \multicolumn{2}{c}{No Investment Constraint} \\
% $p^*$ & $\phi^*$ & $p^*$ & $\phi^*$ \\ \midrule
% 0.9479 & -0.9479 & 0.9479 & -0.9479 \\
% 0.1190 & -0.3571 & 0.1190 & -0.3571 \\
% 0 & 0 & -0.3414 & 1.7072 \\ \bottomrule
% \end{tabular}
% \end{table}

% {\color{red} [\textsc{Insert discussion here}]}

Let us firstly focus on the case with a high degree of risk aversion ($\gamma = 4$). When the investor is ambiguity-averse ($H \neq 0_{n\times n}$), we observe from Table \ref{tab:det-opt-strat-distort-comparison} that the imposition of the short-selling constraint eliminates the short position of wealth exposed to the third Brownian motion (risk factor), while retaining the positive positions in the other two Brownian motions. Moreover, the short-selling constraint also has an impact on the optimal distortion process. Indeed, due to no exposure to the third Brownian motion, the third entry of $\phi^* (t)$ becomes zero, that is, the investor is not concerned about the uncertainty from the third Brownian motion. Without the ability to enter a net short position with respect to the third risk factor to offset the net long positions in the other two risk factors, as reflected in Figure \ref{fig:det-opt-consumption-comparison}, the investor also reduces her consumption since she has a high degree of risk aversion ($\gamma = 4$). 
%Furthermore, \ref{fig:det-opt-consumption-utility-loss-comparison}(a) also shows that the ambiguity-averse investor has a more conservative consumption path than the ambiguity-neutral one. 
However, if the ambiguity-averse investor ignores model uncertainty, she will face a smaller utility loss when the short-selling constraint is in place. As expected, the investor is better off (in terms of the value function as shown in Figure \ref{fig:det-valuefunction-comparison}) when short-selling is allowed, although the gap decreases to zero when approaching the terminal time. Furthermore, the investor considers alternative models that are closer to the reference model, as evidenced by the lower relative entropy implied by the optimal distortion vector when short-selling is not allowed. 

% {\color{red} [Insert discussion for when the agent is not ambiguity-averse ($H = \mathbf{0}$).]}

\begin{figure}[h!]
    \centering
    \begin{subfigure}[b]{0.30\textwidth}
        \includegraphics[width = \textwidth]{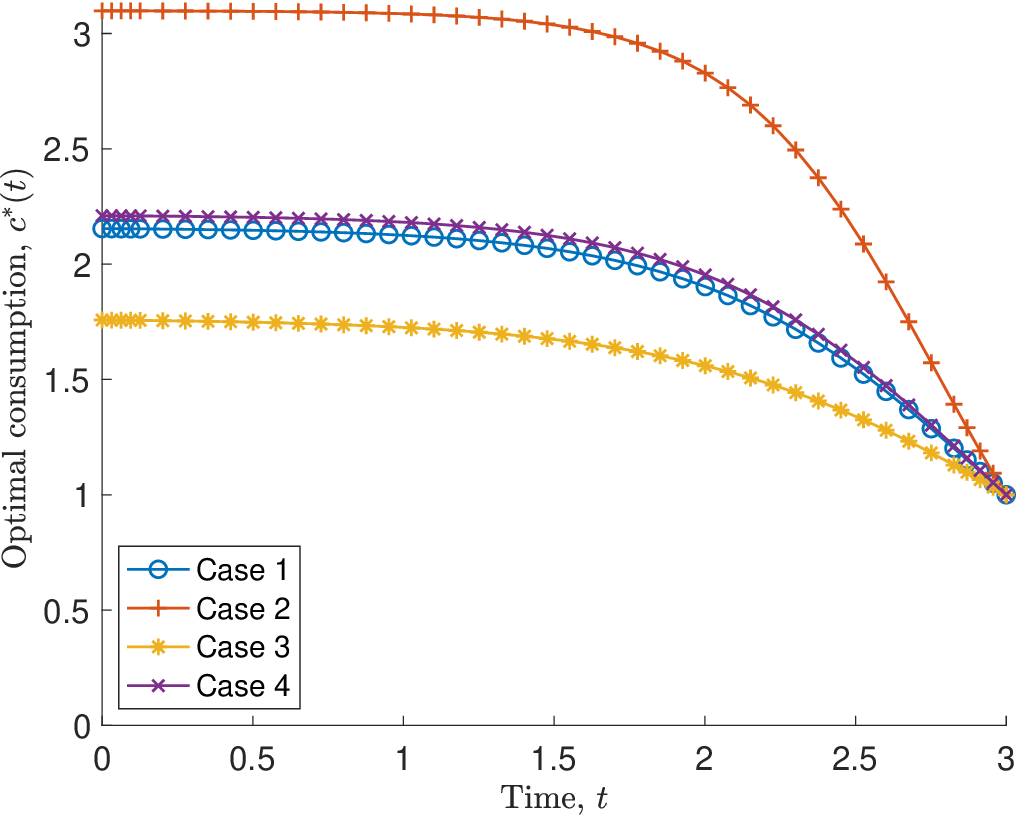}
        \caption{$\gamma = 4, c^*(t)$}
        \label{fig:det-opt-consumption-comparison}
    \end{subfigure}
    \hfill
    \begin{subfigure}[b]{0.30\textwidth}
        \includegraphics[width = \textwidth]{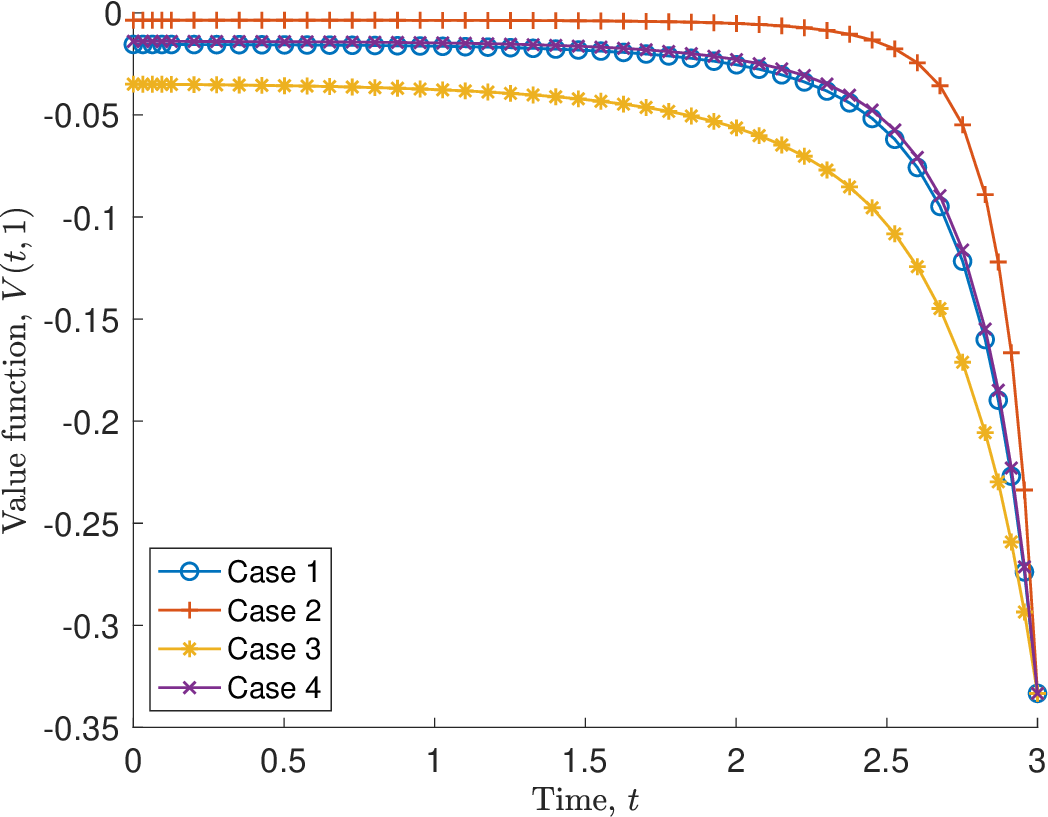}
        \caption{$\gamma = 4, V(t,1)$}
        \label{fig:det-valuefunction-comparison}
    \end{subfigure}
    \hfill
     \begin{subfigure}[b]{0.30\textwidth}
        \includegraphics[width = \textwidth]{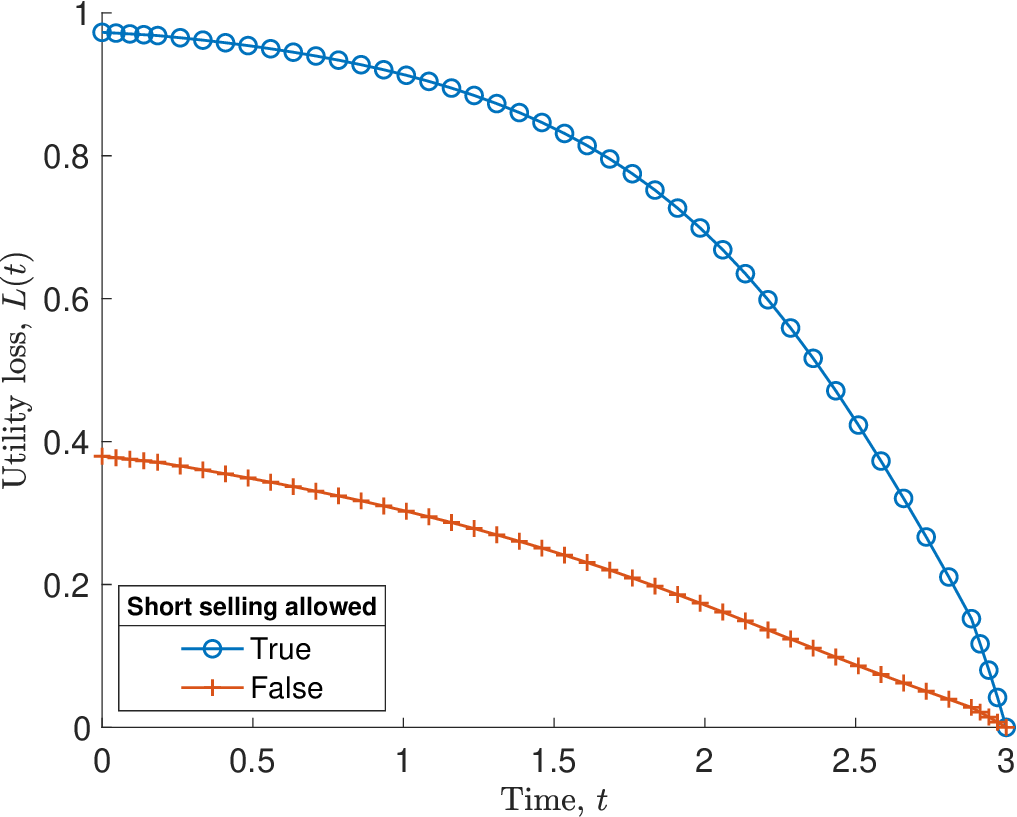}
        \caption{$\gamma = 4, L(t)$}
        \label{fig:det-util-loss-comparison}
    \end{subfigure}
    
    \begin{subfigure}[b]{0.30\textwidth}
        \includegraphics[width = \textwidth]{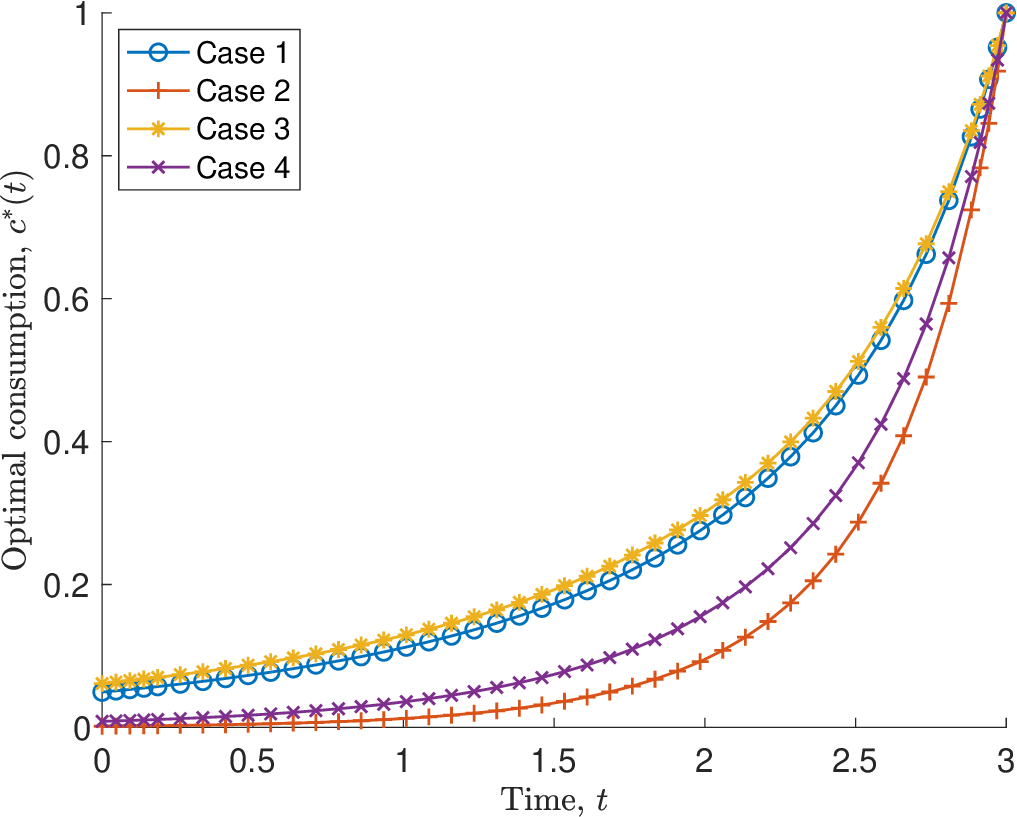}
        \caption{$\gamma = 0.9, c^*(t)$}
        \label{fig:det-opt-consumption-comparison-l1}
    \end{subfigure}
    \hfill
    \begin{subfigure}[b]{0.30\textwidth}
        \includegraphics[width = \textwidth]{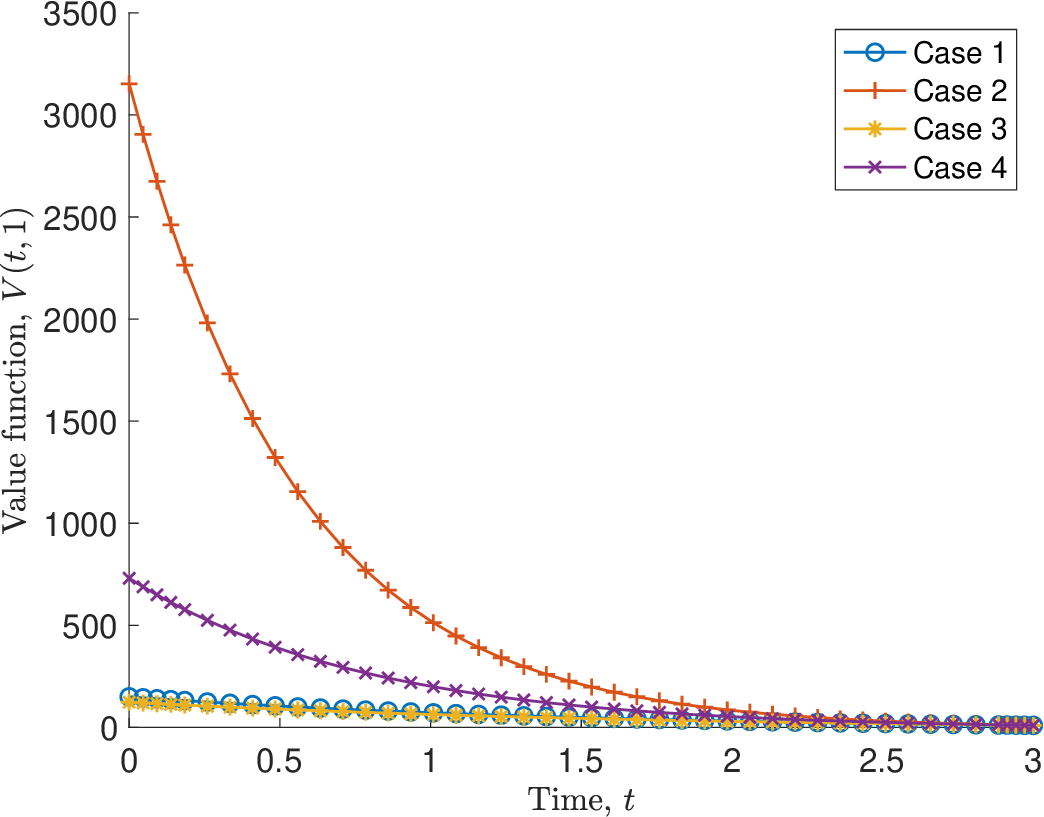}
        \caption{$\gamma = 0.9, V(t,1)$}
        \label{fig:det-valuefunction-comparison-l1}
    \end{subfigure}
    \hfill
    \begin{subfigure}[b]{0.30\textwidth}
        \includegraphics[width = \textwidth]{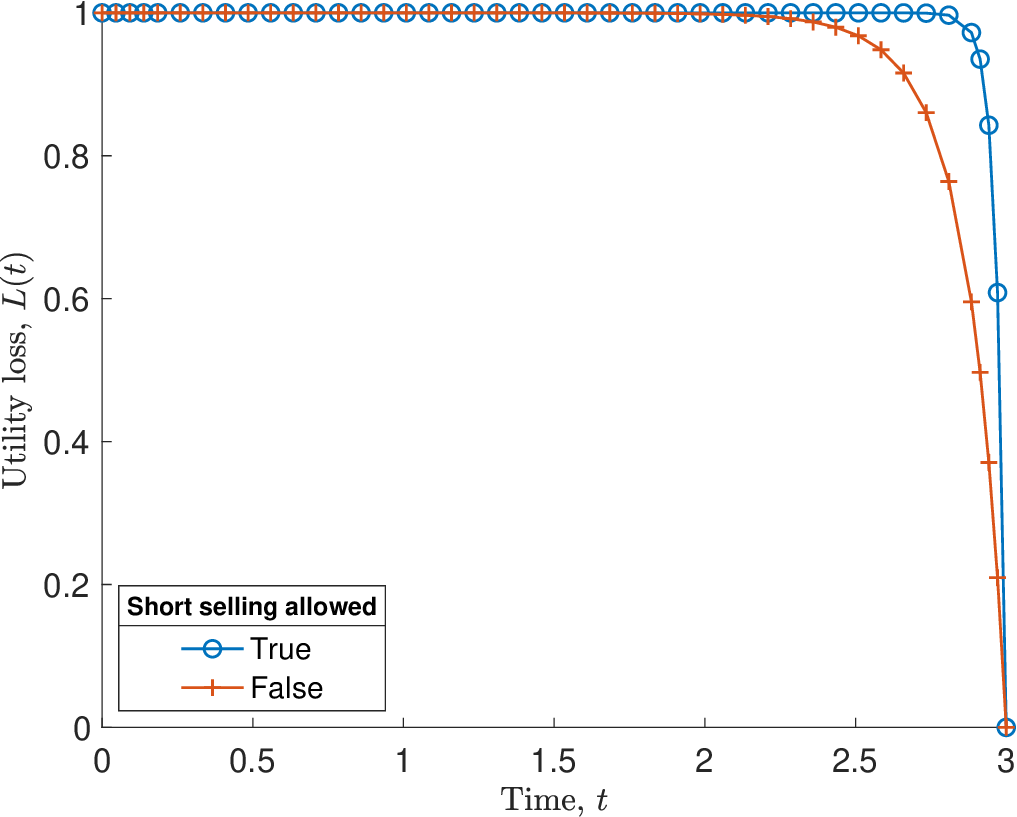}
        \caption{$\gamma = 0.9, L(t)$}
        \label{fig:det-util-loss-comparison-l1}
    \end{subfigure}
    \caption{The optimal consumption strategy $c^*(t)$ and the value function $V(t,1)$ for the cases in Table \ref{tab:det-comparison-cases} and the utility loss $L(t)$ when model uncertainty is ignored and when short-selling is allowed and not allowed.}
    \label{fig:det-comparison}
\end{figure}

As reflected by Figures \ref{fig:det-opt-consumption-comparison} and \ref{fig:det-valuefunction-comparison}, whether the investor is ambiguity-averse or not also has substantial implications on the optimal consumption strategy and the value function. When short-selling is allowed, the ambiguity-averse investor will consume less than the ambiguity-neutral investor. Consequently, the ambiguity-neutral investor is better off than the ambiguity-averse investor in terms of the value function. The ambiguity-neutral investor also tends to expose higher proportions of her wealth to the risk factors, as shown by the larger values of the optimal investment strategy $p^*$. We may attribute this to the ambiguity-neutral investor's complete faith in her belief of the risky asset price dynamics and her lower level of ``effective risk aversion'' (which, following \citet{Maenhout2004}, is characterised by the sum of $\gamma$ and the robustness preference parameters $\eta_i$'s). A similar effect takes place when short-selling is not allowed, although the overall optimal consumption strategy and the value function are lower compared to the case when short-selling is allowed.

The restriction of short-selling also has a pronounced impact on the investor's utility loss when she ignores model uncertainty, as evidenced by Figure \ref{fig:det-util-loss-comparison}. In particular, the investor faces a higher utility loss over time when short-selling is permitted. This is primarily due to the investor's optimal investment strategy (see $p^*(t)$ under Cases 1 and 3 in Table \ref{tab:det-comparison-cases}). When short-selling is not permitted (Case 3), the investor simply adopts the same positive exposure to the first two risk factors, as in Case 1, and drops the negative exposure to the third risk factor. As such, when short-selling is not permitted, the investor is concerned about the uncertainty in fewer risk factors compared to when short-selling is permitted, therefore implying a lower utility loss.

When the investor has a low degree of risk aversion ($\gamma = 0.9$), the short-selling constraint and the investor's aversion to ambiguity have the same impact on the investor's utility loss and value function as the case when $\gamma = 4$ (see Figures \ref{fig:det-util-loss-comparison-l1} and \ref{fig:det-valuefunction-comparison-l1}). However, the impact on the optimal consumption strategy is of the opposite direction compared to the case when $\gamma = 4$ (see Figure \ref{fig:det-opt-consumption-comparison-l1}). That is, when short-selling is not allowed, the investor (marginally) \textit{increases} consumption, whether or not they are ambiguity-averse. Likewise, whether or not short-selling is permitted, ambiguity-averse investors tend to consume \textit{more} of their wealth compared to an ambiguity-neutral investor. This is likely due to the investor's higher elasticity of intertemporal substitution,
i.e. $\frac{1}{\gamma}$, when she is less risk-averse, with the effect compounded by the investor's lack of faith in the dynamics of the risk factors (and hence the evolution of the wealth process) if she is ambiguity-averse. Indeed, the effect of higher levels of ambiguity aversion is explored in the succeeding numerical experiments, illustrating the formal results established in Proposition \ref{prop:av-coefficnet}.

% {\color{red} [For these basic comparisons (i.e. Figures \ref{fig:det-opt-consumption-utility-loss-comparison} and \ref{fig:det-valuefunction-comparison}), should we also show the results for $\gamma = 0.9 < 1$?]}

% \begin{figure}
%     \centering
%     \includegraphics[width = 0.5\textwidth]{Figures/DetCase/Comparison_valuefunc.eps}
%     \caption{The agent's value function $V(0,x)$ at time $t = 0$ for the cases in Table \ref{tab:det-comparison-cases}.}
%     \label{fig:det-valuefunction-comparison}
% \end{figure}

% \begin{figure}[h]
%     \centering
%     \begin{subfigure}[b]{0.40\textwidth}
%         \includegraphics[width = \textwidth]{Figures/DetCase/Comparison_valuefunc_time.eps}
%         \caption{$\gamma = 4$}
%         \label{fig:det-valuefunnction-comparison-g1}
%     \end{subfigure}
%     \hfill
%     \begin{subfigure}[b]{0.40\textwidth}
%         \includegraphics[width = \textwidth]{Figures/DetCase/Comparison_valuefunc_time_l1.eps}
%         \caption{$\gamma = 0.9$}
%         \label{fig:det-valuefunnction-comparison-11}
%     \end{subfigure}
%     \caption{The agent's value function $V(t,1)$ for the cases in Table \ref{tab:det-comparison-cases}.}
%     \label{fig:det-valuefunction-comparison}
% \end{figure}

Next, we illustrate the effect of the ambiguity aversion parameters $\eta_i$ on the optimal consumption strategy and the value function, as formally discussed in Proposition \ref{prop:av-coefficnet}. To this end, we keep $\eta_2$ and $\eta_3$ fixed at their respective values in Table \ref{tab:det-par-values} and vary $\eta_1$ over the set $\{0,1,2,3,4,5\}$. The results for $\gamma = 4$ and $\gamma = 0.9$ are shown in Figures \ref{fig:deter-comparisons-ambiguityaversion} and \ref{fig:deter-comparisons-ambiguityaversion_l1}, respectively.

\begin{figure}[h]
    \centering
    \begin{subfigure}[b]{0.30\textwidth}
        \includegraphics[width = \textwidth]{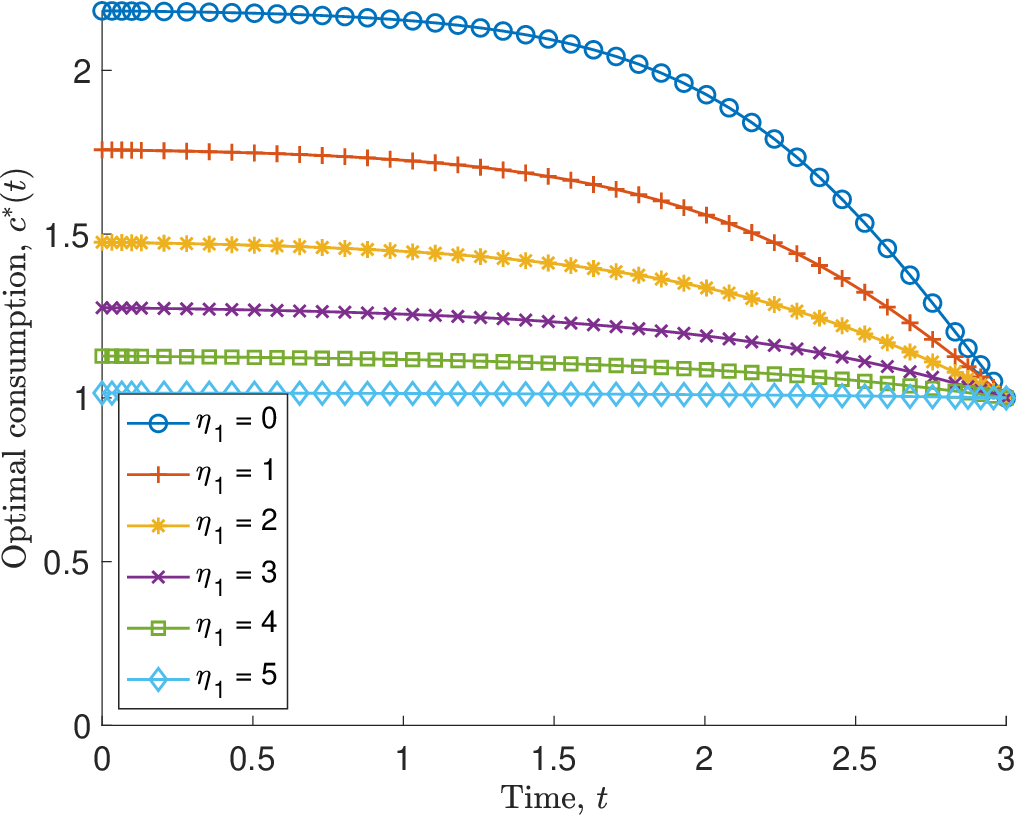}
        \caption{Optimal consumption $c^*(t)$}
    \end{subfigure}
    \qquad
    % \begin{subfigure}[b]{0.40\textwidth}
    %     \includegraphics[width = \textwidth]{Figures/Prop41/Comparison_valuefunc.eps}
    %     \caption{Value function $V(0,x)$}
    % \end{subfigure}
    \begin{subfigure}[b]{0.30\textwidth}
        \includegraphics[width = \textwidth]{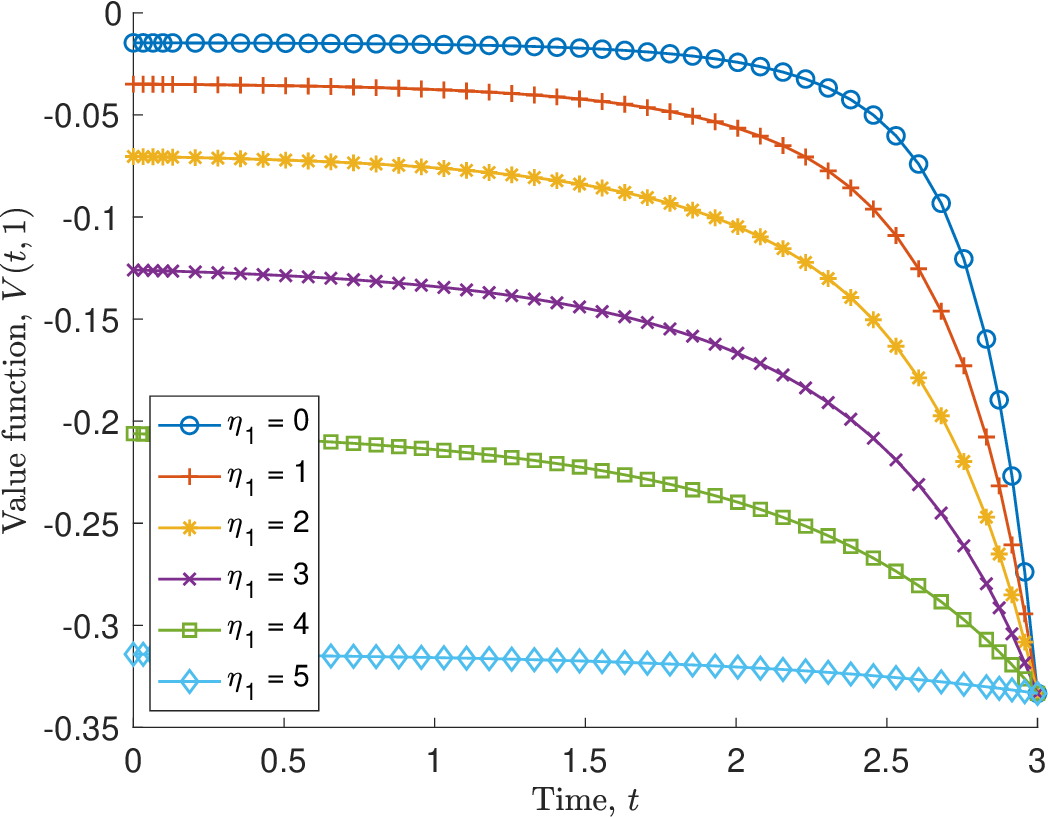}
        \caption{Value function $V(t,1)$}
    \end{subfigure}
    \caption{The optimal consumption strategy $c^*(t)$ and the value function $V(t,1)$ for $\gamma = 4$ and various values of the ambiguity aversion parameter $\eta_1$.}
    \label{fig:deter-comparisons-ambiguityaversion}
\end{figure}

\begin{figure}[h]
    \centering
    \begin{subfigure}[b]{0.30\textwidth}
        \includegraphics[width = \textwidth]{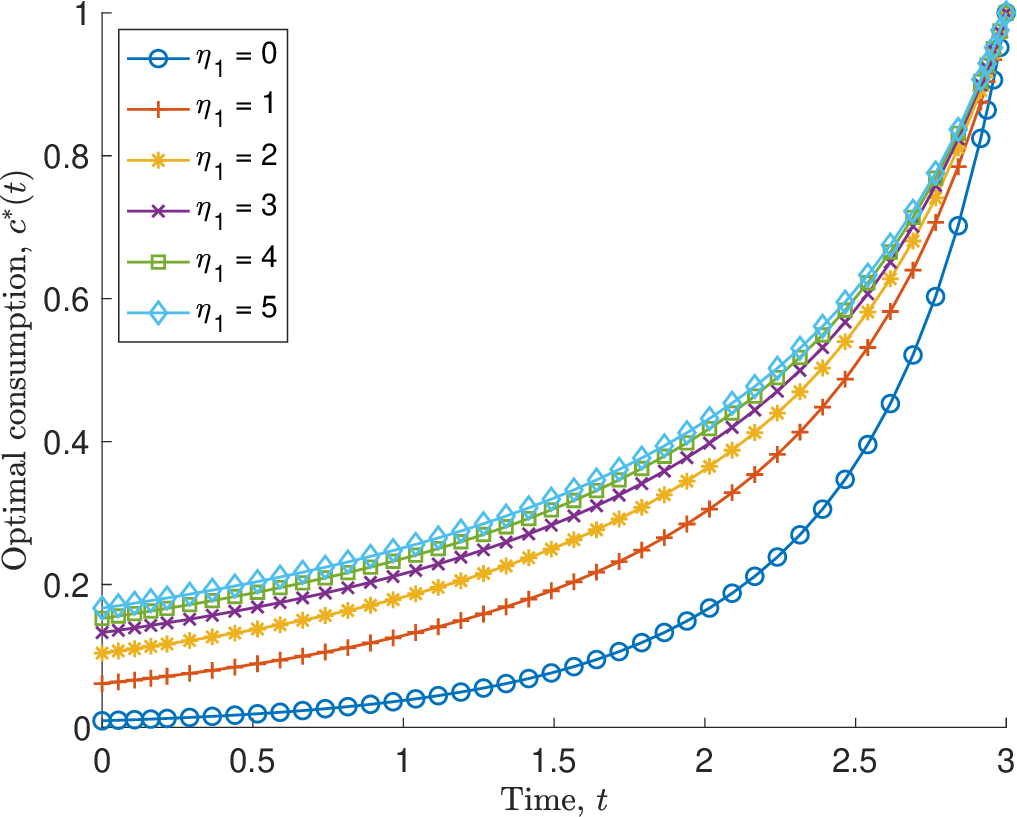}
        \caption{Optimal consumption $c^*(t)$}
    \end{subfigure}
    \qquad
    % \begin{subfigure}[b]{0.40\textwidth}
    %     \includegraphics[width = \textwidth]{Figures/Prop41/Comparison_valuefunc_l1.eps}
    %     \caption{Value function $V(0,x)$ (log-scale) {\color{red} [Comment: Why is the value function so large and why do we have to use log-scale? For instance, if $x = 1$, then the value function is $V(0,x) = 5 [Y (0)]^{0.8}$. Does that imply that $Y (0)$ has a large value? Anyway, I do not like Figures 4(b) and 6(b). For all figures related to the value function, could we present the value function against time (when the current wealth $X (t)$ is fixed at a constant level $x$) instead? This is also because the monotonicity 
    %     of the value function w.r.t. $x$ is obvious.]}}
    % \end{subfigure}
    \begin{subfigure}[b]{0.30\textwidth}
        \includegraphics[width = \textwidth]{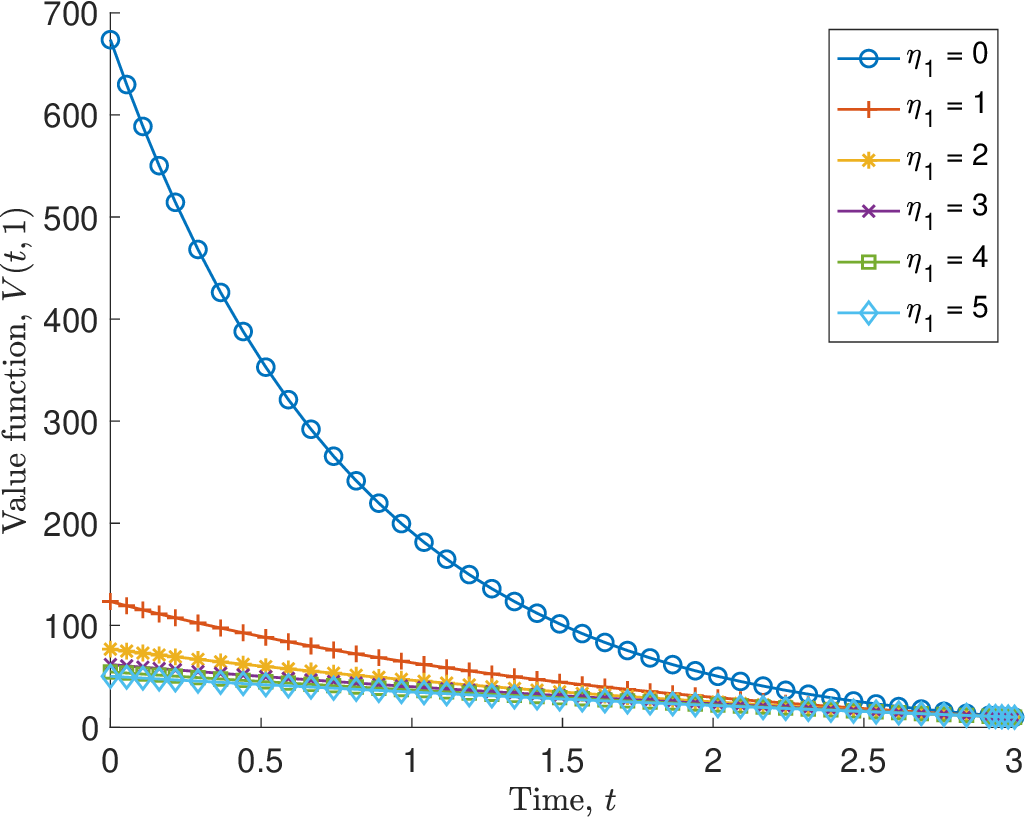}
        \caption{Value function $V(t,1)$ %{\color{red} [Comment: Why is the value function so large and why do we have to use log-scale? For instance, if $x = 1$, then the value function is $V(0,x) = 5 [Y (0)]^{0.8}$. Does that imply that $Y (0)$ has a large value? Anyway, I do not like Figures 4(b) and 6(b). For all figures related to the value function, could we present the value function against time (when the current wealth $X (t)$ is fixed at a constant level $x$) instead? This is also because the monotonicity of the value function w.r.t. $x$ is obvious.]}
        }
    \end{subfigure}
    \caption{The optimal consumption strategy $c^*(t)$ and the value function $V(t,1)$ for $\gamma = 0.9$ and various values of the ambiguity aversion parameter $\eta_1$.}
    \label{fig:deter-comparisons-ambiguityaversion_l1}
\end{figure}

% {\color{red} [\textsc{Insert discussion here}]}

The numerical experiments indeed verify the formal conclusions drawn in Proposition \ref{prop:av-coefficnet}. That is, when the investor has a high degree of risk aversion ($\gamma = 4$, see Figure \ref{fig:deter-comparisons-ambiguityaversion}), the optimal consumption strategy $c^* (t)$ decreases with the ambiguity aversion parameter $\eta_1$. In other words, as the investor prefers stronger robustness in the model for the first risk factor (i.e. first Brownian motion), her optimal consumption also decreases for a fixed $t$. Likewise, her welfare as measured by the value function also decreases with the ambiguity aversion parameter $\eta_1$, as she will only consider alternative models that are closer to the reference model. However, the variations in both the optimal consumption strategy and the value function decrease as time  increases. In contrast, when the investor has a low degree of risk aversion ($\gamma = 0.9$, see Figure \ref{fig:deter-comparisons-ambiguityaversion_l1}), as she increases her ambiguity aversion to the model uncertainty for the first risk factor, her optimal consumption increases for a fixed $t$. As is the case for $\gamma > 1$, the investor's welfare decreases as her ambiguity aversion increases. In further contrast to the case $\gamma > 1$, the magnitude of the changes in the optimal consumption strategy as $\eta_1$ increases is lower when $\gamma < 1$. At the same time, the magnitude of decrease in the value function as $\eta_1$ increases is much larger compared to when $\gamma > 1$. Mathematically speaking, setting $\eta_1 = 0$ drastically increases the exponent of $Y(t)$ in Equation \eqref{eq:Y-sol-deter} due to the division by $\gamma < 1$ (see also the simplified form of the exponent in the proof of Proposition \ref{prop:av-coefficnet}). This then leads to a very large value of the value function, which dramatically decreases as $\eta_1$ increases.

The above analysis illustrates a possible link between a power utility investor's consumption path and her ambiguity aversion, as mentioned below Proposition \ref{prop:av-coefficnet}. An ambiguity-neutral power utility investor with a high risk aversion ($\gamma > 1$) has a low elasticity of intertemporal substitution ($\frac{1}{\gamma} < 1$), and thus she prefers present consumption over future consumption. However, when the investor becomes more ambiguity-averse, her belief about the uncertainty of the financial market leads to lower levels of consumption and welfare. The opposite effect is true for the investor with a low risk aversion ($0<\gamma<1$); as she becomes more ambiguity-averse, her preference for present consumption (over future consumption) increases, compared to an ambiguity-neutral investor with a low degree of risk aversion. In either case, an increase in the investor's ambiguity averison leads to a lower level of welfare, as measured by the value function. Indeed, such welfare loss is the 
cost of robustifying optimal strategies to combat model uncertainty.
% {\color{cyan} [Comment: Could we add explanation to the difference between $\gamma>1$ and $\gamma<1$?]}

Furthermore, we investigate numerically the impact of investment and (upper and lower) consumption constraints on the investor's optimal consumption strategy and value function, as formally shown in Proposition \ref{prop:det-comparison}. The consumption constraints used to illustrate Proposition \ref{prop:det-comparison} are provided in Table \ref{tab:det-comparison-par-values}. In this analysis, we impose more stringent bounds on consumption compared to the earlier analyses, where we have set $[\underline{c}, \overline{c}] = [0,\infty)$. That is, we impose a more stringent upper (resp. lower) bound of 1 (resp. 0.2) on the consumption rate, implying that the investor can consume at most 100\% (resp. must consume at least 20\%) of her current wealth. Figure \ref{fig:deter-comparisons-gammaG1} shows the optimal consumption strategy, the utility loss, and the value function for the different cases considered in Proposition \ref{prop:det-comparison} for $\gamma = 4$, while Figure \ref{fig:deter-comparisons-gammaL1} exhibits the same quantities for $\gamma = 0.9$. The illustrations verify the results stated in Proposition \ref{prop:det-comparison} in the setting implied by Tables \ref{tab:det-par-values} and \ref{tab:det-comparison-par-values}, although we note that equality holds in some comparisons.

\begin{table}[h]
\caption{Assumed investment constraints and upper and lower consumption constraints for the different cases in Proposition \ref{prop:det-comparison}}
\label{tab:det-comparison-par-values}
\centering
\begin{tabular}{@{}cccc@{}}
\toprule
Case & Short-selling & $\underline{c}$ & $\overline{c}$ \\ \midrule
$\rm C_1$ & Not allowed & 0 & 1 \\
$\rm C_2$ & Not allowed & 0.2 & $\infty$ \\
$\rm C_3$ & Not allowed & 0 & $\infty$ \\
$\rm C_4$ & Allowed & 0 & 1 \\
$\rm C_5$ & Allowed & 0.2 & $\infty$ \\
NC & Allowed & 0 & $\infty$ \\ \bottomrule
\end{tabular}
\end{table}

\begin{figure}[h]
    \centering
    % \begin{subfigure}[b]{0.40\textwidth}
    %     \includegraphics[width = \textwidth]{Figures/GammaG1/Comparison_BSDEsol.eps}
    %     \caption{Solution $Y(t)$ of \eqref{eq:Y-sol-deter}}
    % \end{subfigure}
    % \hfill
    \begin{subfigure}[b]{0.30\textwidth}
        \includegraphics[width = \textwidth]{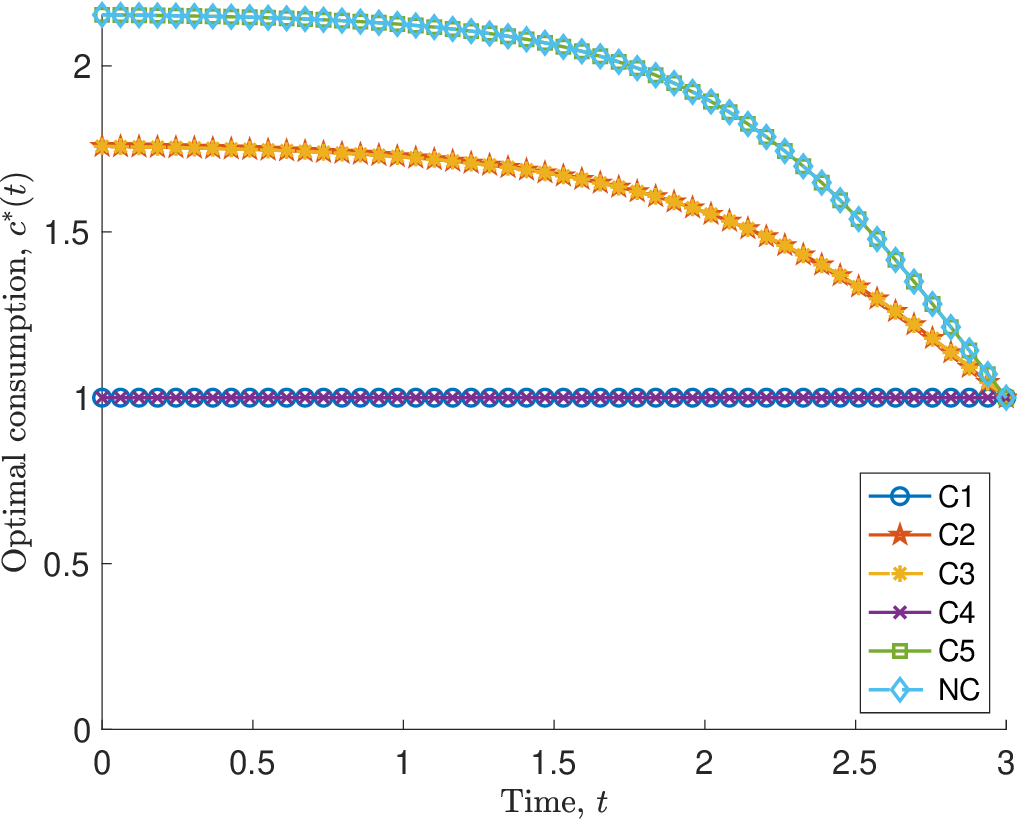}
        \caption{Optimal consumption $c^*(t)$}
    \end{subfigure}
    \hfill
    % \begin{subfigure}[b]{0.40\textwidth}
    %     \includegraphics[width = \textwidth]{Figures/Prop42GammaG1/Comparison_value_func.eps}
    %     \caption{Value function $V(0,x)$}
    % \end{subfigure}
    \begin{subfigure}[b]{0.30\textwidth}
        \includegraphics[width = \textwidth]{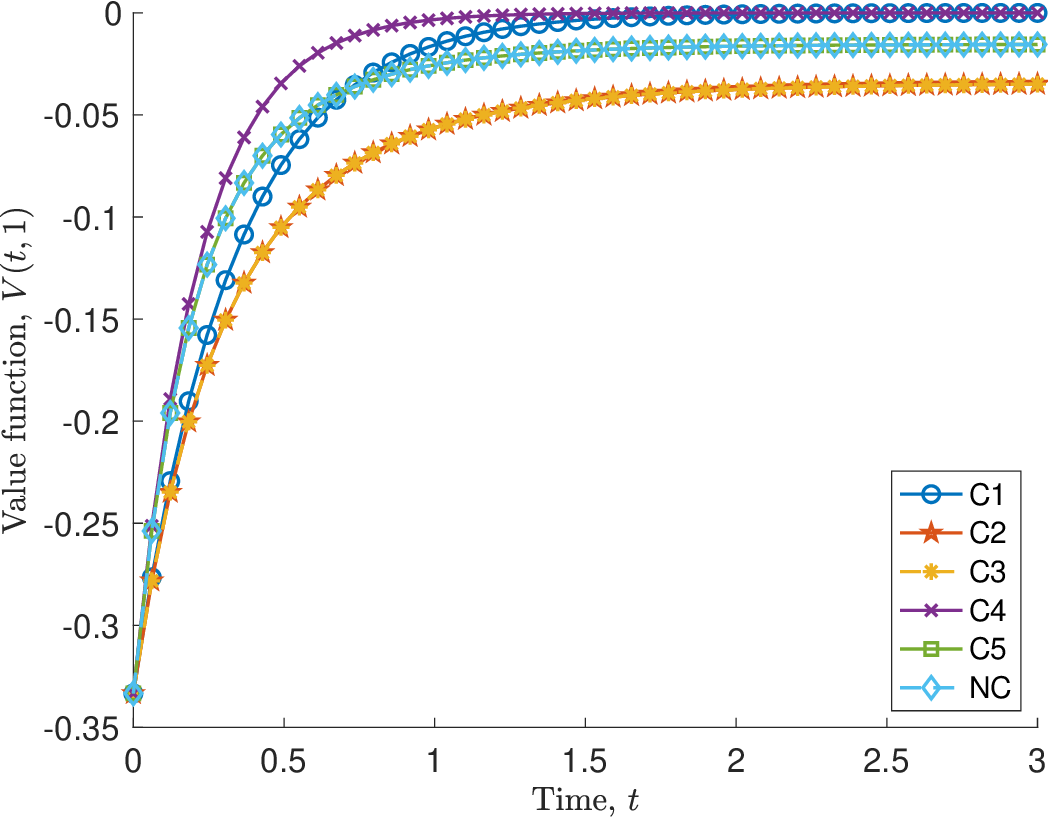}
        \caption{Value function $V(t,1)$}
    \end{subfigure}
    \hfill
    \begin{subfigure}[b]{0.30\textwidth}
        \includegraphics[width = \textwidth]{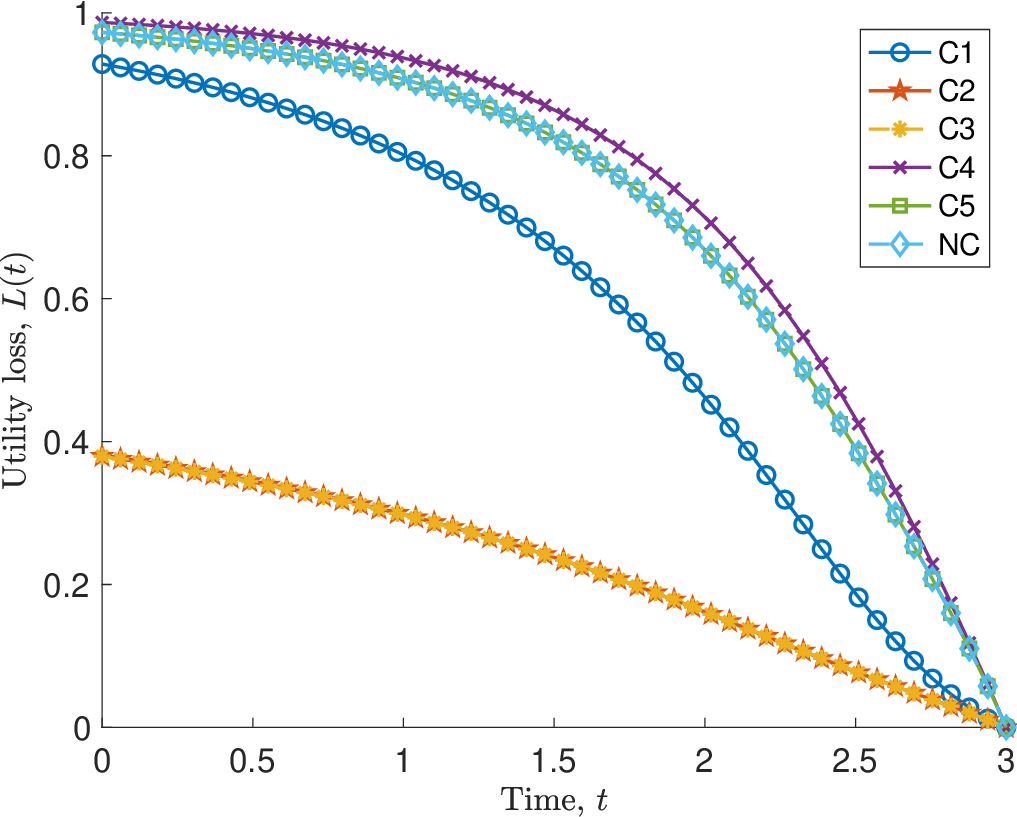}
        \caption{Utility loss $L(t)$}
    \end{subfigure}
    
    \caption{The optimal consumption strategy $c^*(t)$, the value function $V(t,1)$, and the utility loss $L(t)$ for each case considered in Proposition \ref{prop:det-comparison} for $\gamma = 4$.}
    \label{fig:deter-comparisons-gammaG1}
\end{figure}

\begin{figure}[h]
    \centering
    % \begin{subfigure}[b]{0.40\textwidth}
    %     \includegraphics[width = \textwidth]{Figures/GammaL1/Comparison_BSDEsol.eps}
    %     \caption{Solution $Y(t)$ of \eqref{eq:Y-sol-deter}}
    % \end{subfigure}
    % \hfill
    \begin{subfigure}[b]{0.30\textwidth}
        \includegraphics[width = \textwidth]{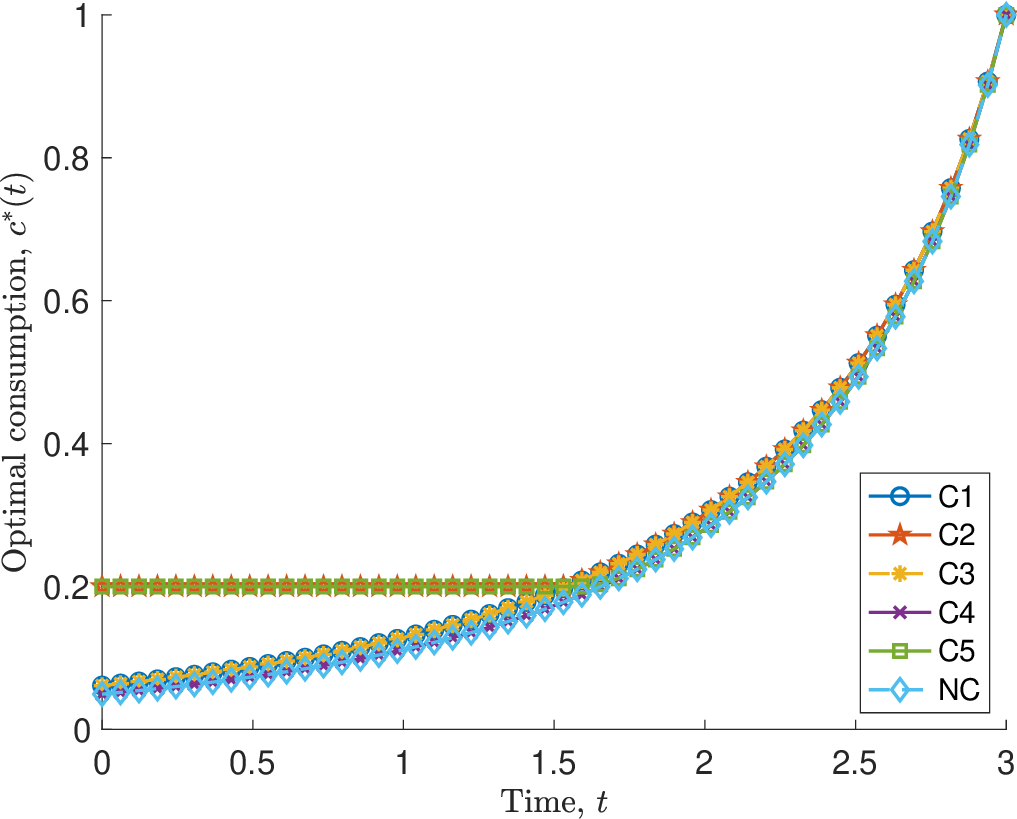}
        \caption{Optimal consumption $c^*(t)$}
    \end{subfigure}
    \hfill
    % \begin{subfigure}[b]{0.40\textwidth}
    %     \includegraphics[width = \textwidth]{Figures/Prop42GammaL1/Comparison_value_func.eps}
    %     \caption{Value function $V(0,x)$}
    % \end{subfigure}
    \begin{subfigure}[b]{0.30\textwidth}
        \includegraphics[width = \textwidth]{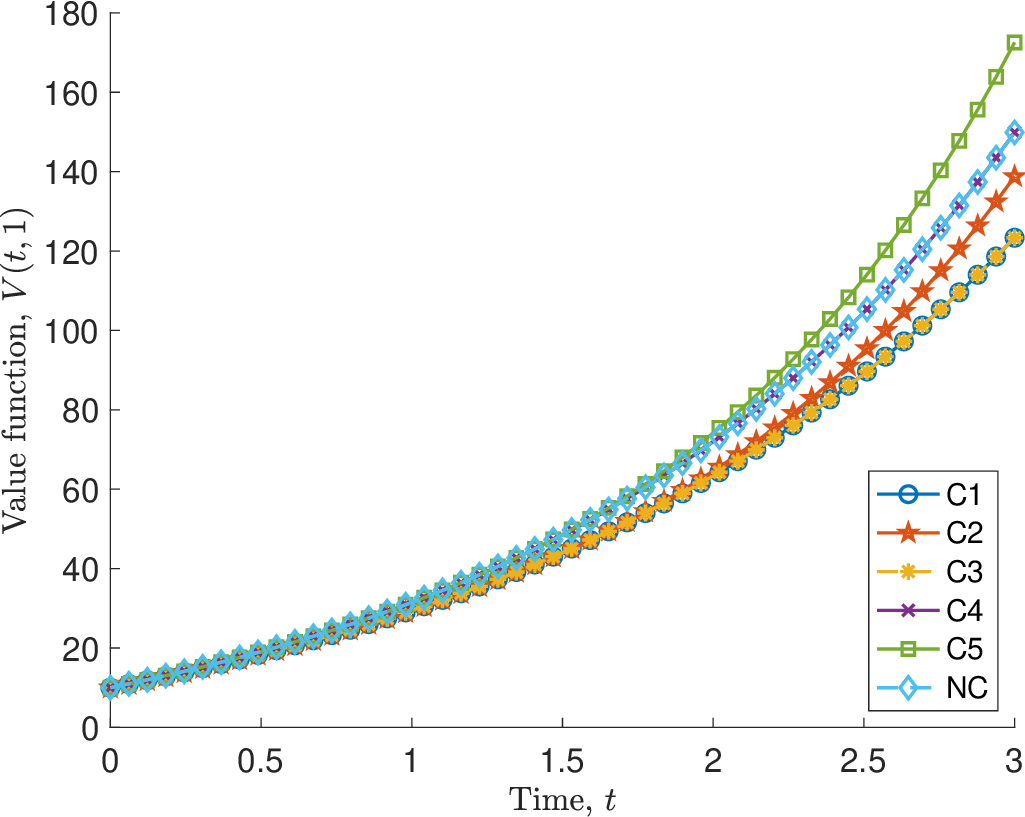}
        \caption{Value function $V(t,1)$}
    \end{subfigure}
    \hfill
    \begin{subfigure}[b]{0.30\textwidth}
        \includegraphics[width = \textwidth]{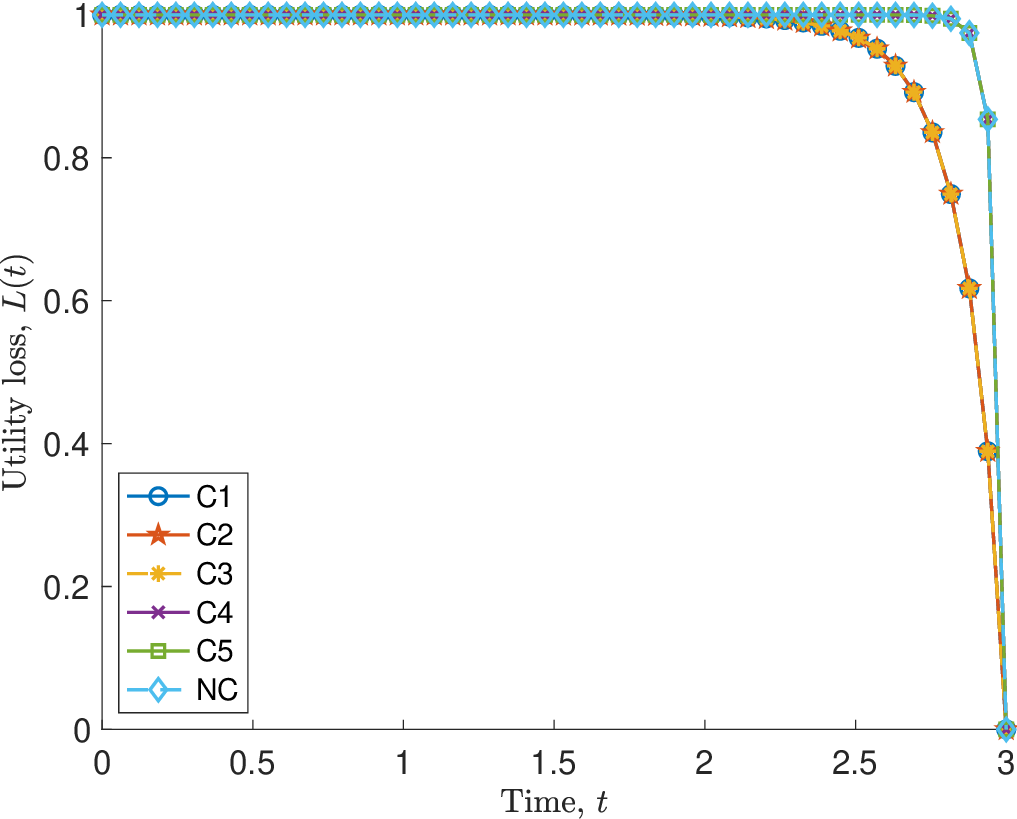}
        \caption{Utility loss $L(t)$}
    \end{subfigure}
    
    \caption{The optimal consumption strategy $c^*(t)$, the value function $V(t,1)$, and the utility loss $L(t)$ for each case considered in Proposition \ref{prop:det-comparison} for $\gamma = 0.9$.}
    \label{fig:deter-comparisons-gammaL1}
\end{figure}

% {\color{red} [\textsc{Insert discussion here}]}

On the one hand, we consider the case in which the investor has a high degree of risk aversion ($\gamma = 4$, see Figure \ref{fig:deter-comparisons-gammaG1}). When only an upper consumption constraint is in place, the optimal consumption strategies are the same whether or not short-selling is allowed ($c_{\rm C_1}^*(t)$ and $c_{\rm C_4}^*(t)$). %\sout{\color{red} We also note that the investor's value functions are identical in this situation. [Comment: The value functions are identical after around time $t=2$, but they are not the same before that time. Please double check and add an explanation if necessary.]} 
However, the value function corresponding to the case when short-selling is allowed is higher (for some values of $t$) compared to when short-selling is prohibited. In contrast, when there is only a lower consumption constraint or no consumption constraint at all, there is a clear gap between the optimal consumption strategies when short-selling is allowed and not allowed ($c_{\rm C_2}^*(t)$ and $c_{\rm C_5}^*(t)$; $c_{\rm C_3}^*(t)$ and $c_{\rm NC}^*(t)$). In Assertion (iii) of Proposition \ref{prop:det-comparison}, the aforementioned inequalities are strict for the value functions except for when short-selling is not allowed and when there may or may not be a lower consumption constraint; that is, we have $V_{\rm C_2}(t,1) = V_{\rm C_3}(t,1)$.

We also observe that short-selling restrictions and consumption constraints have varying effects on the investor's utility loss when she has a high degree of risk aversion $(\gamma = 4)$. As seen in Figure \ref{fig:det-util-loss-comparison}, the investor generally faces a lower utility loss when short-selling is not allowed. Furthermore, setting a lower consumption constraint leads to the same utility loss as the case when there are no consumption constraints (compare cases $\rm C_2$ and $\rm C_3$ and cases $\rm C_5$ and $\rm NC$). However, imposing an upper consumption constraint (cases $\rm C_1$ and $\rm C_4$) leads to an increase in the utility loss. Since the investor's consumption over time is capped at 100\% of her current wealth, the welfare is influenced more by her utility from the terminal wealth, which is more heavily affected by the model uncertainty in the risk factor dynamics. In this example, the increase in the utility loss is more substantial when short-selling is not allowed (comparing case $\rm C_1$ to cases $\rm C_2$ and $\rm C_3$) than when short-selling is permitted (comparing case $\rm C_4$ to cases $\rm C_5$ and $\rm NC$). This suggests a possible delicate interplay between the short-selling restriction and the upper consumption constraint and their effect on the utility loss faced by the investor.

On the other hand, when the investor has a low degree of risk aversion ($\gamma = 0.9$, see Figure \ref{fig:deter-comparisons-gammaL1}), the gaps in the optimal consumption strategies are much smaller, although the relevant assertions in Proposition \ref{prop:det-comparison} still hold. The imposition of a short-selling constraint does not induce a change in the optimal consumption path when there is (only) a lower consumption constraint present (that is, we have $c_{\rm C_2}^*(t) = c_{\rm C_5}^*(t)$). All other inequalities stated in Assertion (i') of Proposition \ref{prop:det-comparison} hold strictly, even if the gaps are very small. In contrast, all inequalities stated in Assertion (ii) of Proposition \ref{prop:det-comparison} hold strictly. However, the value functions do not change between cases $\rm C_1$ and $\rm C_3$ (the imposition of an upper consumption constraint only when short-selling is not allowed) and cases $\rm C_4$ and $\rm NC$ (the imposition of a short-selling constraint when there are no consumption constraints in place). All other inequalities in Assertion (iii') in Proposition \ref{prop:det-comparison} hold as strict inequalities in other cases. In this case, only the short-selling restriction has an effect on the utility loss; the investor faces a slightly lower utility loss when short-selling is not allowed (cases $\rm C_1$, $\rm C_2$, and $\rm C_3$) compared to when short-selling is permitted (cases $\rm C_4$, $\rm C_5$, and $\rm NC$).

These results show the significant effect of an upper (resp. lower) consumption constraint when the investor has a high (resp. low) degree of risk aversion. When the investor has a high degree of risk aversion ($\gamma > 1$), she tends to favor present consumption more and consume a proportion of current wealth that is higher than 1. Although a short-selling constraint reduces the overall level of consumption, it is reduced to a level that still exceeds unity. However, if an upper consumption constraint is imposed, the investor finds it optimal to adopt a level of consumption equal to the upper consumption constraint. Imposing a lower consumption constraint produces an immaterial effect in the optimal consumption path. When the investor's risk aversion is low ($0<\gamma<1$), she tends to favor future consumption and consume minimally in the beginning. As such, imposing an upper consumption constraint has no effect, whereas imposing a lower consumption constraint leads to an optimal consumption path which coincides with the lower constraint until some time in the future when the investor decides to increase her consumption.

\begin{figure}[h]
    \centering
    \begin{subfigure}[b]{0.30\textwidth}
        \includegraphics[width = \textwidth]{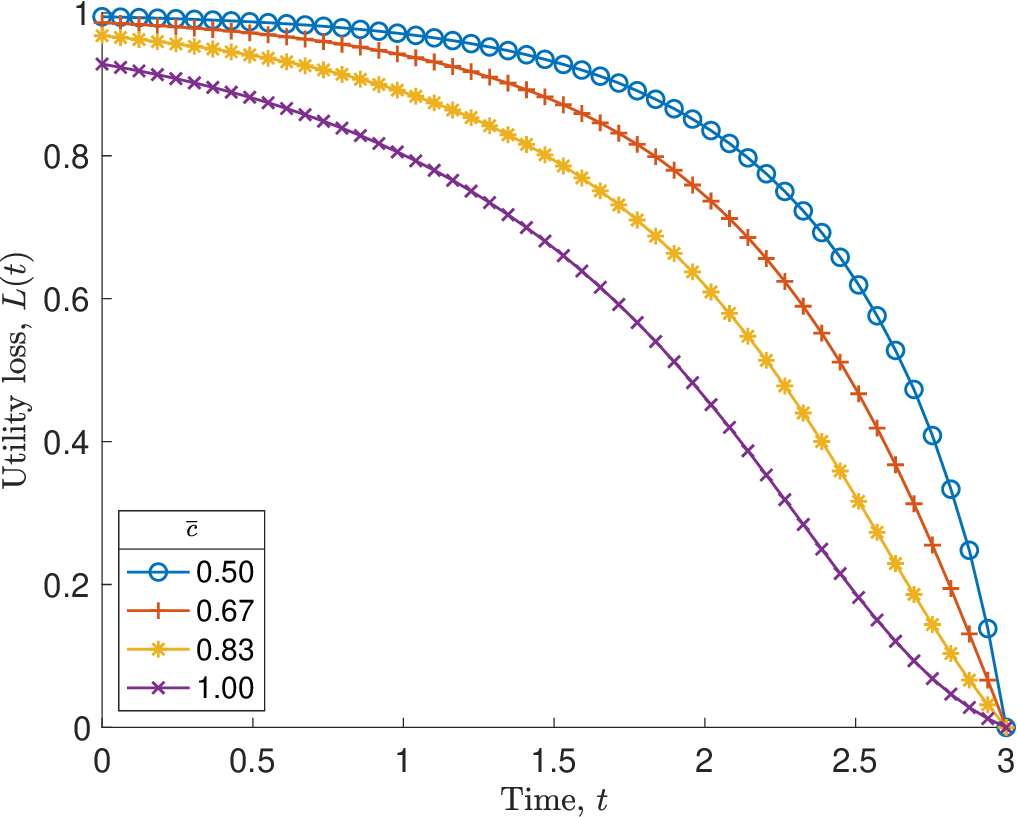}
        \caption{$\gamma = 4$}
    \end{subfigure}
    \qquad
    \begin{subfigure}[b]{0.30\textwidth}
        \includegraphics[width = \textwidth]{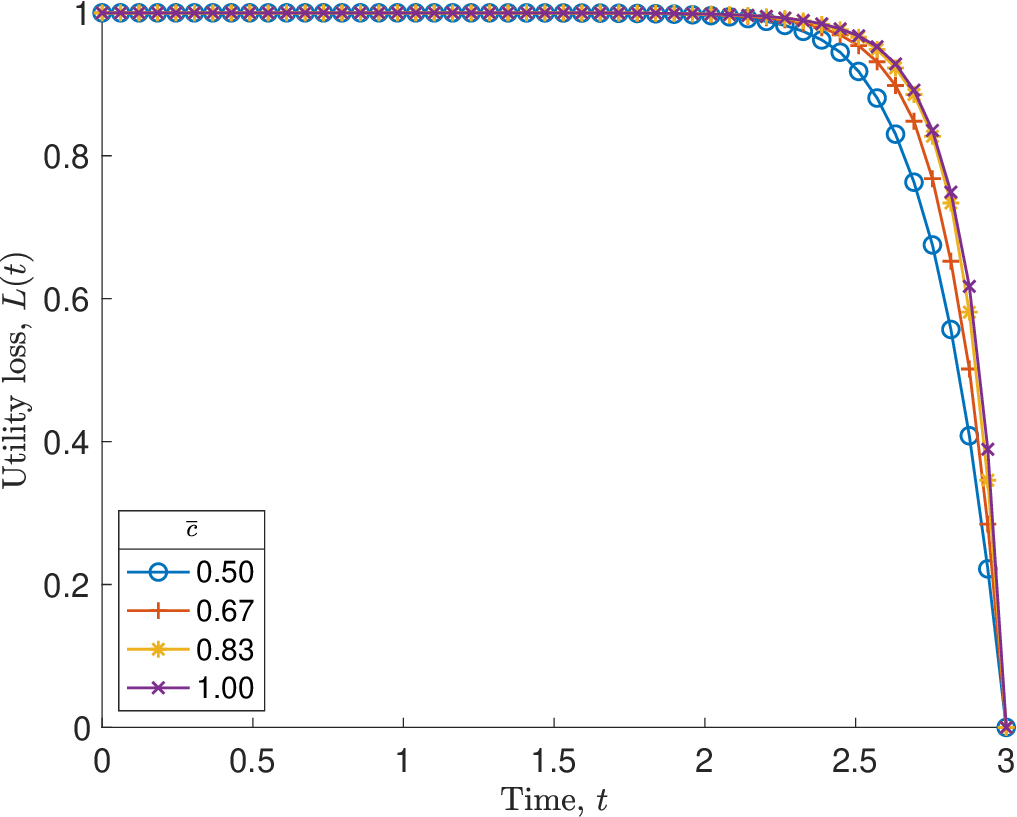}
        \caption{$\gamma = 0.9$}
    \end{subfigure}
    \caption{The effect of changing the upper consumption constraint $\overline{c}$ on the utility loss when short-selling is not allowed}
    \label{fig:deter-util-loss-ucc-comparison}
\end{figure}

Finally, in Figure \ref{fig:deter-util-loss-ucc-comparison}, we illustrate numerically the impact of imposing more stringent upper consumption constraints on the investor's utility loss when short-selling is not allowed. When $\gamma = 4$, decreasing the upper consumption constraint leads to an increasing utility loss when model ambiguity is ignored. 
% This implies that as the upper consumption constraint is made more stringent, less (initial) wealth is required for an investor who adopts the robust optimal consumption and investment strategy to realize the same welfare (as measured by the value function) as an investor who adopts the sub-optimal consumption and investment plans that ignore model uncertainty. [Note: This might be obvious as this statement follows immediately from the definition of utility loss $L(t)$.] 
We also recall from Figure \ref{fig:det-opt-consumption-comparison} that when short-selling is not allowed, an ambiguity-averse investor tends to consume less than an ambiguity-neutral investor. In other words, the investor tends to increase her consumption when she ignores model ambiguity. If the upper consumption constraint is lowered, the investor is unable to increase her consumption as much as if there were no upper consumption constraint. Thus, 
the more stringent is the upper consumption constraint, 
the higher is the investor's utility loss. 
When $\gamma = 0.9$, decreasing the upper consumption constraint leads to a decreasing utility loss when model ambiguity is ignored.
This is not unexpected due to that the investor now tends to decrease her consumption when she ignores model ambiguity as observed in Figure \ref{fig:det-opt-consumption-comparison-l1}.
Therefore, when facing a more stringent upper consumption constraint, the investor is able to decrease her consumption as much as needed, which leads to a smaller utility loss. 

\section{Conclusions}
\label{sec:conclusion}

There is now a wide recognition of the prominent role played by model ambiguity or uncertainty, in addition to the role of \textit{known} risks, in investors' decision-making. In this light, this paper investigates and solves the investment-consumption problem when the investor is both risk- and ambiguity-averse in a general continuous-time financial market with possibly stochastic coefficients. The investor's optimisation problem is formulated as a robust optimal control problem, as in \citet{AndersonHansenSargent2003}, using the homothetic robustness preference formulation proposed by \citet{Maenhout2004}. As the model coefficients are allowed to be stochastic, our study accommodates joint uncertainty from multiple sources, such as the expected return, the volatility, and interest rates. In addition, we consider the effect of imposing investment and consumption constraints on the investor's decision.

To tackle the robust optimal control problem, we solve the corresponding stochastic HJBI equation. In Proposition \ref{prop:HJBI-solution}, we show that the solution of the stochastic HJBI equation are represented by the solution of a BSDE. Furthermore, in Theorem \ref{thm:main} using the theory of BMO martingales, we establish that the candidate strategies obtained in the solution of the stochatic HJBI equation are indeed the robust optimal investment-consumption strategy and the worst-case distortion process that solve the robust optimal control problem. We also find in Proposition \ref{prop:utilloss} that the investor faces a non-negative utility loss when she follows a sub-optimal strategy that ignores model uncertainty.

An analysis of the deterministic case reveals significant relationships between the investor's degree of ambiguity aversion and her optimal consumption strategy (Proposition \ref{prop:av-coefficnet}) and interactions between investment and/or consumption constraints and the investor's optimal consumption strategy and welfare (Proposition \ref{prop:det-comparison}). A series of numerical experiments are then conducted to illustrate the effects of upper and lower consumption constraints and a short-selling restriction on the robust optimal consumption strategy and value function. We find that, when the investor has a high (resp. low) degree of risk aversion, the robust optimal consumption strategy increases (resp. decreases) when the investor is also ambiguity-averse, whether or not there is a short-selling restriction. We also find that, when the investor ignores model uncertainty, she faces a smaller utility loss when short-selling is prohibited, regardless of the degree of risk aversion. Furthermore, if the investor has a high (resp. low) risk aversion, the optimal consumption decreases (resp. increases) when her degree of ambiguity aversion increases. We also show numerically that imposing a more stringent upper consumption constraint leads to a higher (resp. lower) utility loss when the investor has a high (resp. low) risk aversion.

From the numerical experiments on the effect of the investor's ambiguity aversion on consumption, we find that the investor's elasticity of intertemporal substitution (EIS) may interact with her ambiguity aversion in determining robust optimal investment-consumption strategy. Thus, future work is concerned with solving the robust optimal control problem under recursive preferences, which allows us to disentangle the effects of risk aversion and EIS. Another future research topic is the effect of stochastic parameters (e.g. stochastic volatility and/or stochastic interest rates) on the investor's robust optimal strategies. Doing so allows us to investigate the evolution of the optimal investment strategy over time and the effect of model parameters and its interaction with risk and/or ambiguity aversion on the investor's decision making.

%\citet{SVDNF-2023}

%\section{Exponential utility}
%
%In this section, we consider
%\begin{eqnarray*}
%U (x) = - \exp(-\alpha x) .
%\end{eqnarray*}
%Assume that $r (t) = 0$ throughout this section. We try the ansatz
%\begin{eqnarray}
%V (t, x) = - \beta_2 \exp(-\alpha x) Y (t),
%\end{eqnarray}
%with
%\begin{eqnarray*}
%Y (T) = 1 .
%\end{eqnarray*}
%Then
%\begin{eqnarray}
%p^* (t) = \frac{1}{\alpha} \big [ I - \Gamma^{-1} \big ]^{-1}
%\bigg [ \theta (t) + \frac{Z (t)}{Y (t)} \bigg ]  ,
%\end{eqnarray}
%\begin{eqnarray}
%c^* (t) = x - \frac{1}{\alpha} \log (Y (t)) + \frac{1}{\alpha} \log \bigg ( \frac{\beta_1}{\beta_2} \bigg ) ,
%\end{eqnarray}
%and
%\begin{eqnarray}
%\phi^* (t) = - \big [ I - \Gamma \big ]^{-1}
%\bigg [ \theta (t) + \frac{Z (t)}{Y (t)} \bigg ]
%\end{eqnarray}

\bibliographystyle{apalike}
\begin{spacing}{1}
\setlength{\bibsep}{0pt plus 0.3ex}
\bibliography{references}
\end{spacing}

\end{document}